\documentclass{article}

\usepackage[utf8]{inputenc}   
\usepackage[english]{babel}      
\usepackage{amsmath, amssymb, amsthm, mathrsfs}     
\usepackage[a4paper, margin=3cm]{geometry} 
\usepackage{graphicx}         
\usepackage{hyperref}         
\usepackage{tikz}
\allowdisplaybreaks
\usepackage{float} 
\usepackage{graphicx, pgfplots, xcolor}
\usepackage{bm}
\definecolor{granate}{RGB}{128, 0, 64} 

\usepackage[T1]{fontenc}
\usepackage{subcaption}
\usepackage{booktabs}
\usepackage[authoryear]{natbib} 
\newtheorem{theorem}{Theorem}

\newtheorem{definition}[theorem]{Definition}
\newtheorem{lemma}[theorem]{Lemma}

\newtheorem{proposition}[theorem]{Proposition}

\begin{document}
\title{Robust Estimation in Step-Stress Experiments under Weibull Lifetime Distributions.}
\author{Mar\'{i}a Jaenada$^{(1)}, $ Juan Manuel Mill\'{a}n$^{(2)}$ and Leandro Pardo$^{(2)}$ \\
{\small $^{(1)}$ Department of Statistics, O.R. and N.A., UNED, Madrid, Spain}\\
{\small $^{(2)}$ Department of Statistics and O.R., Complutense University of Madrid, Spain}
}
\date{}
\maketitle

\begin{abstract}
	Many modern products are highly reliable, often exhibiting long lifetimes. As a result, conducting experiments under normal operating conditions can be prohibitively time-consuming to collect sufficient failure data for robust statistical inference. Accelerated life tests (ALTs) offer a practical solution by inducing earlier failures, thereby reducing the required testing time. In step-stress experiments, a stress factor that accelerates product degradation is identified and systematically increased at predetermined time points, while remaining constant between intervals. Failure data collected under these elevated stress levels is analyzed, and the results are then extrapolated to normal operating conditions.
	
	Traditional estimation methods for such data, such as the maximum likelihood estimator (MLE), are highly efficient under ideal conditions but can be severely affected by outlying or contaminated observations. To address this, we propose the use of Minimum Density Power Divergence Estimators (MDPDEs) as a robust alternative, offering a balanced trade-off between efficiency and resistance to contamination. The MDPDE framework is extended to mixed distributions and its theoretical properties, including the asymptotic distribution of the model parameters, are derived assuming Weibull lifetimes. The effectiveness of the proposed approach is illustrated through extensive simulation studies, and its practical applicability is further demonstrated using real-world data.
\end{abstract}


\section{Introduction}
The continuous improvement in the manufacturing of industrial and consumer products has led to a significant increase in their reliability, resulting in longer expected lifetimes until failure. These highly reliable products are commonly found in modern devices such as electronic components, batteries, sensors, and mechanical parts used in automotive and aerospace industries. However, these advancements have also made lifetime testing more challenging, since longer product lifetimes render tests under normal operating conditions both time and resource intensive.
To overcome this limitation, experiments designed to induce early failures are required. One widely used approach in industry, introduced in the early second half of the twentieth century, is the Accelerated Life Test (ALT). In these experiments, one or more stress factors responsible for product degradation and failure are identified and then increased to accelerate the failure process. For instance, for electronic components, the operating temperature can be increased, whereas for batteries, a higher voltage may be applied to accelerate the degradation process.
Data collected from an ALT are statistically analyzed, and results are then extrapolated to estimate reliability characteristics under normal conditions. Specifically, the lifetime at normal usage stress can be inferred through extrapolation using a stress–response regression model.

The statistical methodology for analyzing Accelerated Life Tests (ALTs) data emerged prominently in the early 1980s, primarily through the work of \cite{nelson1978, nelson1980, nelson1990}. Key initial contributions include \cite{nelson1978}, which established the theoretical foundation for optimum censored tests under Weibull and Extreme Value distributions; \cite{nelson1980}, which provided an important overview of step-stress models and their corresponding data analysis techniques; and \cite{nelson1990}, a reference that organized statistical models and test plans for the field. Since this foundational work, a vast body of literature has explored statistical methods across various ALTs models and censoring schemes.

Significant references in the development of step-stress ALTs methodology are numerous. For example, \cite{miller1983} determined optimal simple step-stress plans for ALT, and \cite{bai1989} proposed optimum simple step-stress ALTs with censoring. \cite{meeker1998} stands as a textbook consolidating statistical methods for reliability data, complementing the mentioned literature. Later contributions include \cite{balakrishnan2007}, who provided point and interval estimation for the simple step-stress model under Type-II censoring with exponential lifetimes, \cite{balakrishnan2009a}, who presented point and interval estimation for the Lognormal lifetime distribution with Type-I censoring, and \cite{balakrishnan2009b}, who proposed exact inference for the Exponential distribution under time constraints.

Censoring is an inherent feature in reliability experiments, as failure times are often unobservable within the duration of the study. Depending on how the experiment termination criterion is defined, two conventional censoring schemes are typically considered. Under Type-I censoring, the experiment is terminated at a predetermined time, making the number of observed failures a random variable. Conversely, in Type-II censoring, the test ends after a predetermined number of failures, resulting in a random total test duration. These two schemes constitute the basis for most statistical analyses of reliability data. 
Besides, the assumed censoring scheme directly affects both the variables involved in the analysis and the structure of the statistical model. In this study, we adopt the first approach, in which the test duration is predetermined.

Parametric inference relies on the assumption that the lifetime distribution belongs to a specific parametric family, such as the exponential, Weibull, gamma, or log-normal model, and so that its form is known except for a parameter vector, that we aim to estimate from observed data.
One of the most common distributions used to model lifetimes, generalizing the exponential distribution, is the Weibull distribution with scale parameter $\lambda$ and shape parameter $\eta$, $W(\lambda, \eta)$. Its strength lies in its flexibility, which allows different relationships between stress levels and reliability to be captured. The Weibull distribution has been shown to fit many real-world products well and is frequently employed in analyses involving the proportional hazards property. It includes the exponential distribution when the shape parameter is $\eta=1.$ Because of its practical relevance, the Weibull distribution will be considered in this work.

Four different types of stress loading are commonly described in the ALT literature: Constant-stress tests are those in which the stress applied to all test units remains constant throughout the entire experiment. This design allows for straightforward modeling but may require longer testing periods to observe failures.

Important references in constant-stress ALTs are \cite{meeker1975, kielpinski1975, nelson1976, nelson2009}. \cite{kielpinski1975} and \cite{nelson1976} addressed the theory and optimum censored tests for Normal and Lognormal lifetime distributions. \cite{meeker1975} developed optimum tests for Weibull and Extreme Value distributions. \cite{bai1991} presented an optimum design for the Exponential distribution. \cite{yang1994} developed optimum constant-stress test plans. Work on Partially Accelerated Life Testing, a related constant-stress approach, includes \cite{bai1993}, which developed an optimum design for the Lognormal distribution under Type-I censoring. \cite{hyun2015} studied the Log-Logistic distribution. \cite{abdel2016} compared estimation methods under exponentiated distributions. \cite{nassar2018} compared estimation methods for the Exponentiated Rayleigh distribution. \cite{elsagheer2018} developed inference procedures based on progressive Type-II censoring. \cite{aldayian2021} addressed Bayesian estimation and prediction for the Topp Leone–Inverted Kumaraswamy distribution. More recent work includes \cite{kumar2022}, who proposed estimation procedures for the Generalized Inverse Lindley distribution.

In step-stress ALTs (SSALTs), all units are subjected to the same stress pattern, where the stress level is abruptly increased at predetermined times or after a fixed number of failures, until either all samples have failed or the duration at the maximum stress level has elapsed. Key references in SSALTs include \cite{xiong1998}, which develops inference for the simple step-stress model with Type-II censored exponential data. \cite{balakrishnan2007} provides point and interval estimation under Type-II censoring. \cite{kateri2008} performs inference for the Weibull distribution under Type-II censoring. \cite{balakrishnan2009a} addresses inference for the Lognormal distribution under Type-I censoring. \cite{kateri2024} proposes the product of spacings estimation method. More recent developments focus on robustness, such as \cite{balakrishnan2023a}, which studies robust estimation for non-destructive one-shot devices.

In progressive-stress ALTs (PSALTs), the stress level continuously increases over time. This approach was fundamentally introduced by \cite{starr1961} as an accelerated method for voltage endurance testing. Early studies include \cite{yin1987}, which examined acceleration laws and data analysis methods for PSALTs. Subsequent work focused on inference under specific distributions and censoring schemes. \cite{abdel2011} developed inference procedures for the Weibull distribution under progressive Type-II censoring. \cite{abdel2015} addressed inference for the Exponentiated Exponential distribution under progressive hybrid Type-II censoring. \cite{alhussaini2015} studied Bayesian prediction intervals for the Half-Logistic distribution. \cite{mohie2017} developed classical and Bayesian inference for the Extension of the Exponential distribution. \cite{ronghua2004} investigated statistical inference for the Weibull distribution under a tampered failure rate model. \cite{bai1992} developed optimum simple ramp-test designs for the Weibull distribution.

Finally, in cyclic-stress ALTs (CySALTs) the stress level changes periodically. Important contributions include reliability modeling of one-shot units under thermal cyclic stresses \cite{cheng2017}, the application of accelerated temperature cycle tests and the Coffin-Manson model \cite{cui2005}, the proposal of a general model for age acceleration during thermal cycling \cite{nachlas1986}, and reliability estimation for one-shot devices under cyclic ALT as well as modeling using the Generalized Gamma lifetime \cite{zhu2021, zhu2022}.

Each stress design has its own advantages and disadvantages, and the choice of design may also depend on the nature of the product. In this study, we focus on step-stress tests under continuous monitoring, which have been shown to achieve the same level of accuracy as constant-stress tests but with shorter experimental times and, consequently, lower costs.

Bringing all the above considerations together, this study focuses on the SSALT model under Type-I censoring, assuming Weibull lifetime distribution. This model has been extensively investigated in the literature; to mention a few representative contributions, \cite{khamis1997} compared constant stress levels and simple step-stress ALTs under a Weibull distribution. \cite{bai1993} studied the optimal SSALTs for Weibull distributions. \cite{zhao2005} developed inference methods using Weibull and lognormal distributions.  These studies rely on the Maximum Likelihood Estimator (MLE) to perform inference, as it performs well in the absence of data contamination; however, they may become unreliable when the data under study contain outlier observations. In other words, the MLE is asymptotically efficient, but non-robust. To address this gap, a class of robust estimators for the SSALT model based on the Weibull lifetime distribution is developed in this paper.

To the best of our knowledge, the only studies that investigate robust estimators for the aforementioned model is under Weibull distribution\cite{balakrishnan2025}, where minimum density power divergence estimators (MDPDE) were developed. These works assume the failures observed are interval-censored, and so the which consisted of partitioning the support of the distribution that characterizes the lifetime (in our case the support of the Weibull distribution) of the device under consideration according to inspection times
In these studies, failures are assumed to be interval-censored, so the observed data consist of counts of failures within the inspected intervals. Accordingly, the support of the Weibull distribution is discretized according to the scheduled inspection times of the units under test, and the resulting likelihood is modeled through a multinomial distribution with probabilities obtained from the continuous lifetime distribution. Thus, the MDPDEs proposed therein were derived by minimizing the Density Power Divergence (DPD) between the empirical and model-based probability vectors of the corresponding multinomial formulation.
The interval-monitored experimental design has also been studied under other commonly used lifetime distributions, with theoretical developments addressing both estimation and hypothesis testing. For example, \cite{balakrishnan2023a,balakrishnan2023b, balakrishnan2023c}  used a exponential model, while \cite{balakrishnan2024b} and  \cite{balakrishnan2024a} assumed gamma and lognormal lifetimes, respectively. A semi-parametric approach was adopted in \cite{balakrishnan2023d}.
 However, in the experimental design with interval censoring, the exact failure times are not available, but only a count of failures is recorded. This limitation entails a loss of information that is relevant to the analysis, if continuous monitoring is feasible.
The use of exact failure times, together with inference based on MDPDEs, results in a novel modeling framework, first introduced for exponential lifetime data in \cite{jaenada2025}. We extend the theory here for Weibull lifetime distributions.

The rest of the paper is organized as follows: 
 Section 2 introduces the Weibull model with continuous monitoring time and two stress levels. Section 3 develops the Minimum Density Power Divergence Estimator (MDPDE). Section 4 studies its asymptotic properties and establishes both point and interval estimation methods for characteristics of interest. Finally, Section 5 provides an analysis with simulated data, and Section 6 presents a study with real data.

\section{The Step-Stress Model under Weibull Distributions}

In a SSALT under Type I censoring, the total duration of the study is partitioned into stress periods determined by the times of stress change. For simple SSALTs with only two stress levels,  $ x_1 $ and $ x_2, $ these  time points are denoted by $ \tau_0 = 0, \tau_1 $, and $ \tau_2 $.  Figure~\ref{fig:step-stress} describes the stress profile applied during the experiment.
Consistent with the assumption that higher stress accelerates degradation, the stress profile is assumed to increase monotonically.

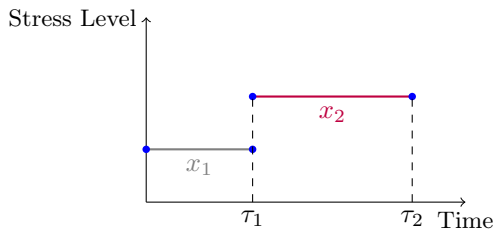
\begin{figure}[htb]
\centering
\begin{tikzpicture}[scale=0.7]
  \draw[->] (0, 0) -- (6, 0) node[below] {\small Time};
  \draw[->] (0, 0) -- (0, 3.5) node[left] {\small Stress Level};

  \draw[thick, gray] (0, 1) -- (2, 1) node[pos=0.5, below] {$x_1$};
  \draw[thick, purple] (2, 2) -- (5, 2) node[pos=0.5, below] {$x_2$};

  \fill[blue] (0, 1) circle(2pt);
  \fill[blue] (2, 1) circle(2pt);
  \fill[blue] (2, 2) circle(2pt);
  \fill[blue] (5, 2) circle(2pt);

  \node at (2, -0.3) {$\tau_1$};
  \node at (5, -0.3) {$\tau_2$};

  \draw[dashed] (2, 0) -- (2, 2);
  \draw[dashed] (5, 0) -- (5, 2);
\end{tikzpicture}
  \caption{Step-stress test design with Type-I censoring.}
  \label{fig:step-stress}
\end{figure}

Let $n$ units be placed under test, and denote their observed failure times by $ t_j, \, j=1,\ldots,n $.  Units that do not fail before the termination time, $\tau_2$ are right-censored, consistent with a Type-I censoring scheme.
Moreover,  let $ n_i $ denote the number of failures observed under a stress level $ x_i $, $ i=1,2 $. At the end of the experiment, the number of censored observations is  
$
n - n_1 - n_2,
$
and the ordered failure times satisfy  
\begin{equation}
\tau_0 = 0 < t_{1:n} < \cdots < t_{n_1:n} < \tau_1 < t_{n_1+1:n} < \cdots < t_{n_1+n_2:n} < \tau_2.
\label{eq:order}
\end{equation}
Assuming a Weibull distribution for the lifetimes with scale $\lambda > 0$ and shape $\eta >0$, the cumulative distribution function (c.d.f.) of a device lifetime, $T,$ is given by
\begin{equation}
	F(t|\lambda,\eta)=1-\exp\left(-\left(\frac{t}{\lambda}\right)^{\eta}\right)  , \quad t>0. \nonumber
\end{equation}
Since the step-stress testing procedure involves changing the stress level applied to the units during the experiment, it is necessary to specify two models for the analysis; one linking the stress level to the lifetime distribution, and another describing how the distribution evolves as the stress level increases. Several approaches have been proposed in the literature; among them, the log-linear relationship is commonly adopted for the stress–lifetime link, and the cumulative exposure model (CEM) for the stress-transition mechanism.

Under the assumption that higher stress levels accelerate product degradation, the parameters of the Weibull distribution under constant stress are expected to be stress-dependent. In the step-stress literature employing the Weibull model, the scale parameter is commonly modeled as a log-linear function of the stress level, while the shape parameter is typically assumed to remain constant across stress conditions.
 That is, the scale of the Weibull distribution under stress level $x_i$ is given by
\begin{equation}
	\lambda_i = \exp(a_0 + a_1 x_i)\ \text{for } i=1, 2,
	\label{eq:funcrelation}
\end{equation}
where the regression vector $(a_0, a_1) \in \mathbb{R} \times \mathbb{R}^{-}.$ The restricted  parameter space reflects the expected decrease in mean lifetime with increased stress.

On the other hand, under the CE model, the c.d.f. during the experiment is built as a continuous piecewise function based on the c.d.f. at constant stress, where accumulated wear from prior stresses is carried over to subsequent levels by adjusting their respective distributions. Therefore, the c.d.f. for device lifetimes under simple SSALT, $F_T$ , is given by
\begin{equation}
	F_T(t) = \begin{cases}
		F_{x_1}(t)  = 1-\exp\left(-\left(-\frac{t}{\lambda_1}\right)^{\eta}\right)  & 0 < t < \tau_1\\ \nonumber
		F_{x_2}(t+h) = 1-\exp\left(-\frac{1}{\lambda_2^\eta} \left(t+h\right)^{\eta}\right) & \tau_1 \leq t < \tau_2  \nonumber
	\end{cases},
\end{equation}
where $F_{x_i}$ denotes the c.d.f. under stress level $x_i, i =1,2,$ and the shifting time $h$ ensures the continuity of the c.d.f.,
\begin{equation}
F_{x_1}(\tau_1) = F_{x_2}(\tau_1 + h).  \nonumber
\end{equation}
The cumulative distribution function can be extended to infinity by simply setting $\tau_2=\infty.$
The shifting time $h$ can be directly calculated using the properties of these scale parameters as
\begin{equation}
h = \frac{\lambda_2}{\lambda_1} \tau_1 - \tau_1.   \nonumber
\end{equation}


From the above c.d.f., we directly obtain the probability density function (p.d.f.) $f_T$ under the step-stress experiment as follows
\begin{equation}
f_T(t \mid \lambda_1,\lambda_2,\eta)=
\begin{cases}
\begin{aligned}
f_{x_1}(t) &=
\frac{\eta}{\lambda_1}
\left(\frac{t}{\lambda_1}\right)^{\eta-1}
\exp\!\left[
-\left(\frac{t}{\lambda_1}\right)^{\eta}
\right],
\\[-0.2em]
&\hfill 0 \le t < \tau_1
\end{aligned}
\\[0.6em]
\begin{aligned}
f_{x_2}(t) &=
\frac{\eta}{\lambda_2}
\left(
\frac{t + \frac{\lambda_2}{\lambda_1}\tau_1 - \tau_1}
{\lambda_2}
\right)^{\eta-1}
\\
& \times
\exp\!\left[
-\left(
\frac{t + \frac{\lambda_2}{\lambda_1}\tau_1 - \tau_1}
{\lambda_2}
\right)^{\eta}
\right],
\\[-0.2em]
&\hfill \tau_1 \le t < \infty
\end{aligned} \nonumber
\end{cases}
\end{equation}

Now, given the observed failure times $t_{1:n}, \dots, t_{n_1:n}, t_{n_1+1:n}, \dots, t_{n_1+n_2:n}=\bm{T}$, the log-likelihood of the SSALT model is given by
\begin{align}
\ell(\lambda_1,\lambda_2,\eta;
\bm{T})
&= \log(n!)
 - \log\bigl((n-n_1-n_2)!\bigr)
\nonumber\\
&
+ \sum_{i=1}^{n_1}
\log\!\bigl(f_{x_1}(t_i \mid \lambda_1,\eta)\bigr)
+\sum_{i=n_1+1}^{n_1+n_2}
\log\!\bigl(f_{x_2}(t_i+h \mid \lambda_1,\lambda_2,\eta)\bigr)
\nonumber\\
&
+ (n-n_1-n_2)
\log\!\Bigl(
1 - F_{x_2}(\tau_2+h \mid \lambda_1,\lambda_2,\eta)
\Bigr)  \nonumber
\\ 
&= \log(n!)
 - \log\bigl((n-n_1-n_2)!\bigr)
+ (n_1+n_2)\log(\eta)\nonumber
\\
&- \eta\!\left(
n_1\log(\lambda_1)
+ n_2\log(\lambda_2)
\right)+ (\eta-1)\!\left(
\sum_{i=1}^{n_1}\frac{t_i}{\lambda_1}
+ \sum_{i=n_1+1}^{n_1+n_2}\frac{t_i+h}{\lambda_2}
\right)
\nonumber\\
&
- \sum_{i=1}^{n_1}
\left(\frac{t_i}{\lambda_1}\right)^{\eta}
- \sum_{i=n_1+1}^{n_1+n_2}
\left(\frac{t_i+h}{\lambda_2}\right)^{\eta}- (n-n_1-n_2)
\left(\frac{\tau_2+h}{\lambda_2}\right)^{\eta}.
\label{eq:likelihood}
\end{align}

Therefore, the MLE for the SSALT is defined as
\begin{equation*}
	\left(\hat{\lambda}_1^{MLE}, \hat{\lambda}_2^{MLE},\hat{\eta}^{MLE} \right) = \arg \max_{(\lambda_1, \lambda_2,\eta) \in \mathbb{R}^{3+}} \ell(\lambda_1, \lambda_2,\eta ; \bm{T}),
	\label{eq:max_likelihood_lambda}
\end{equation*}
Importantly, note that the MLEs for $ \lambda_1 $ and $ \lambda_2 $ can only be obtained simultaneously if $ n_1 \geq 1 $ and $ n_2 \geq 1 $. Therefore, the existence of the MLE requires that failures occur in both intervals, a condition that will be assumed throughout this paper.

As noted above, we adopt a log-linear link between stress and the Weibull scale parameter$\lambda_i = \exp(a_0 + a_1 x_i)\ \text{for } i=1, 2,$ 
which models the stress–lifetime relationship and enables extrapolation to nominal operating conditions. Besides, when more than two step-stress levels are used, this parametrization reduces dimensionality by replacing a separate scale at each level with the two regression coefficients. Therefore, under the assumption of common Weibull shape parameter $\eta$ across stress levels, the SSALT parameter vector is $\boldsymbol{\theta} = (a_0,a_1,\eta),$ and the corresponding MLE is given by
\begin{align}
	\boldsymbol{\widehat{\theta}}^{MLE} =& \left(\hat{a}_0^{MLE}, \hat{a}_1^{MLE},\hat{\eta}^{MLE} \right)  \arg \max_{\bm{\theta} \in A} \ell(\exp(a_0 + a_1 x_1), \exp(a_0 + a_1 x_2),\eta ; \bm{T}), \nonumber
	\label{eq:max_likelihood_params}
\end{align}
with
$A=\left\{ (a_0, a_1,\eta) \in \mathbb{R} \times  \mathbb{R}^{-} \times  \mathbb{R}^{+} \right\}$. 


 \section{The minimum density power divergence estimator \label{sec:MDPDE}}

Divergence-based estimators form a broad class of estimators obtained by minimizing a statistical divergence between the empirical distribution of the observed data and a postulated parametric model, rather than maximizing the assumed model likelihood. Different divergences can be selected to trade off robustness and efficiency. For instance, $\phi$-divergences often lead to efficient estimators under correct specification, whereas the density power divergence (DPD) and the log-DPD yield robust estimators under model misspecification or data contamination. In particular, using the Kullback–Leibler divergence, which is a special case within these  families, results in the MLE as a minimum-distance solution.
Divergences are typically defined in terms of the empirical c.d.f. for continuous data or in terms of the empirical probability vector for discrete data. To the best of our knowledge, only a few studies have addressed the mixed-distribution case. For instance, \cite{Nair2007} defined the Shannon entropy for mixed distributions as a combination of its continuous and discrete components, and a similar approach was followed in \cite{jaenada2025}.
In this section, we derive the MDPDE for the step-stress model under continuous monitoring. To that end, we define the DPD family (which contains the MLE as the limit 
for the mixed distribution induced by our observation scheme—namely, a continuous component for exact failure times and a point mass at the termination time $\tau_2$
for right-censored units.

Before introducing the Kullback-Leibler divergence, let us motivate the use of a mixed distribution for the observable lifetime. Although units are continuously monitored during the test, at termination only the number of survivors is recorded; their exact failure times are unobserved. Consequently, the observed time $T^{\ast} = \operatorname{min}\left(T, \tau_2\right)$ follows a mixed distribution, denoted by $F_{T^\ast},$ having a continuous component on during the experimental time $(0,\tau_2)$ with density $f_T$, and an atom at $\tau_2$ of size $ 1 - F_T(\tau_2). $ This discrete mass equals the probability that a unit survives until the end of the experiment.
Therefore, for Weibull lifetimes, the observable c.d.f. is defined as:
\begin{equation}
F_{T^{\ast}}(t \mid \boldsymbol{\lambda},\eta)=
\begin{cases}
0, \quad t < 0,
\\
\begin{aligned}
F_{x_1}(t)
&= 1 - \exp\!\left[
-\left(\dfrac{t}{\lambda_1}\right)^{\eta}
\right],
\\[-0.2em]
& 0 < t < \tau_1,
\end{aligned}
\\[0.6em]
\begin{aligned}
F_{x_2}(t+h)
&= 1 - \exp\!\left[
-\dfrac{1}{\lambda_2^{\eta}}
\left(
t + h
\right)^{\eta}
\right],
\\
& \tau_1 \le t < \tau_2,
\end{aligned}
\\
1, \quad \tau_2 \le t.
\end{cases}
\label{eq:funF}
\end{equation}

where $\boldsymbol{\lambda} =\left(\lambda_1, \lambda_2 \right) $ and $\lambda_i, i =1,2$  follows the log-linear relation stated in \eqref{eq:funcrelation}.
Similarly, the true underlying c.d.f. of the lifetime, truncated at the end of the experiment will also be given as a piecewise distribution of the form
\begin{equation}
	G_{T^\ast}(t)=
	\begin{cases}
		0 \quad & t<0
		\\
		G_{x_1}(t) \quad & 0 \leq t <\tau_1
		\\
		G_{x_2}(t) \quad & \tau_1 \leq t < \tau_2
		\\
		1 \quad & \tau_2 \leq t
	\end{cases},
\label{eq:funG}
\end{equation}
where $G_{x_1}(\cdot)$ and $G_{x_2}(\cdot)$ denote the c.d.f.s under constant stresses $x_1$ and $x_2$, respectively.
Note that, rather than assuming a specific model relating the c.d.f. at different stresses, we only consider that it may change when the stress level varies.
\\
Besides, for mathematical simplicity, for $0 < t <\tau_1$
\begin{align}
	f_{x_1}(t|\lambda_1,\eta)&=\frac{\eta}{\lambda_1^{\eta}}t^{\eta-1}\exp\left[-\left(\frac{t}{\lambda_1}\right)^{\eta}\right] \nonumber
	\\
	 g_{x_1}(t)&=\frac{\partial G_{x_1}(t)}{\partial t} \nonumber,
\end{align}
and for $\tau_1<t<\tau_2$
\begin{align}
	f_{x_2}(t+h|\boldsymbol{\lambda},\eta)&=\frac{\eta}{\lambda_2^{\eta}}\left(t+h\right)^{\eta-1}\exp\left[-\left(\frac{t+h}{\lambda_2}\right)^{\eta}\right] , \nonumber
	\\
	g_{x_2}(t)&=\frac{\partial G_{x_2}(t)}{\partial t} \nonumber . 
\end{align}
These functions are not proper density functions themselves, but are proportional to them up to a constant factor, as they represent the continuous component of the mixed distribution.
From the above, we can define the Kullback-Leibler divergence between the mixed distributions $F_{T^\ast}(\cdot|\boldsymbol{\lambda},\eta)$ and $G_{T^\ast}(\cdot)$ as defined in \eqref{eq:funF} and \eqref{eq:funG} is given by
\begin{equation}
\begin{aligned}
d_{\mathrm{KL}}\!\left(
G_{T^{\ast}}, F_{T^{\ast}}
\right)
&=
\int_{0}^{\tau_1}
g_{x_1}(t)
\log\!\frac{g_{x_1}(t)}
{f_{x_1}(t \mid \lambda_1)}
\, dt
\\
&
+ \int_{\tau_1}^{\tau_2}
g_{x_2}(t)
\log\!\frac{g_{x_2}(t)}
{f_{x_2}(t+h \mid \boldsymbol{\lambda},\eta)}
\, dt
\\
&
+ \bigl(1-G_{x_2}(\tau_2)\bigr)
\\
&\times
\log\!\frac{1-G_{x_2}(\tau_2)}
{1-F_{x_2}(\tau_2+h \mid \boldsymbol{\lambda},\eta)} .
\end{aligned}
\label{eq:kl_def}
\end{equation}

By separating terms that do not depend on the unknown parameters, expression \eqref{eq:kl_def} can be rewritten as follows
\begin{equation}
\begin{aligned}
d_{\mathrm{KL}}\!\left(
G_{T^{\ast}}, F_{T^{\ast}}
\right)
&=
- \int_{0}^{\tau_1}
\log f_{x_1}(t \mid \lambda_1)\,
dG_{x_1}(t)
\\
&
\quad- \int_{\tau_1}^{\tau_2}
\log f_{x_2}(t+h \mid \boldsymbol{\lambda},\eta)\,
dG_{x_2}(t)
\\
&
\quad- \bigl(1-G_{x_2}(\tau_2)\bigr)
\log\!\bigl(
1-F_{x_2}(\tau_2+h \mid \boldsymbol{\lambda},\eta)
\bigr)
\\
&
\quad+ K .
\end{aligned}
\end{equation}

where K a constant that does not depends on the Weibull parameters $(\boldsymbol{\lambda}, \eta).$

Minimum divergence estimates of the SSALT parameters $ \boldsymbol{\theta} = (a_0, a_1, \eta)^T$ are those that make the assumed parametric  c.d.f. fit the true distribution as closely as possible, as measured by the Kullback-Leibler divergence. Because the true c.d.f.'s $G_{x_1}(\cdot)$ and $G_{x_2}(\cdot)$ are unknown, given a ordered random sample from the experiment, $(t_{1:n}, ..., t_{n1+n2:n})$, we can approximate them by their corresponding empirical distribution functions associated. Then we have
\begin{align}
d_{\mathrm{KL}}\bigl(G_{T^\ast}, F_{T^\ast}\bigr)
&\approx
- \frac{n_1}{n} \sum_{i=1}^{n_1} \log f_{x_1}(t_{i:n} \mid \lambda_1)
\nonumber\\
&
- \frac{n_2}{n} \sum_{i=n_1+1}^{n_1+n_2} \log f_{x_2}(t_{i:n}+h \mid \boldsymbol{\lambda},\eta)
\nonumber\\
&
- \frac{n-n_1-n_2}{n} \log\bigl(1 - F_{x_2}(\tau_2+h \mid \boldsymbol{\lambda},\eta)\bigr)
+ K \nonumber
\\
&= -\frac{n_1+n_2}{n} \log(\eta)
+ \frac{\eta}{n}\bigl(n_1 \log(\lambda_1) + n_2 \log(\lambda_2)\bigr)
\nonumber\\
&
- \frac{\eta-1}{n} \left(
\sum_{i=1}^{n_1} \frac{t_i}{\lambda_1}
+ \sum_{i=n_1+1}^{n_1+n_2} \frac{t_i+h}{\lambda_2}
\right)
\nonumber\\
&
+ \frac{1}{n} \sum_{i=1}^{n_1} \left(\frac{t_i}{\lambda_1}\right)^\eta
+ \frac{1}{n} \sum_{i=n_1+1}^{n_1+n_2} \left(\frac{t_i+h}{\lambda_2}\right)^\eta
\nonumber\\
&
+ \frac{n-n_1-n_2}{n} \left(\frac{\tau_2+h}{\lambda_2}\right)^\eta
+ K
\nonumber\\
&= K - \frac{1}{n} \ell(\lambda_1, \lambda_2, \eta; \bm{T}).
\end{align}
where $\ell(\lambda_1, \lambda_2,\eta ; n)$ is the log-likelihood of the model, as defined in \eqref{eq:likelihood}.
Therefore, the values of $(\lambda_1, \lambda_2,\eta)$ that minimizes the estimated Kullback-Leibler divergence is the same as the one that maximizes the log-likelihood function. Assuming the log-linear relationship \eqref{eq:funcrelation} and reparametrizing accordingly, we can equivalently defined the MLE as the minimum Kullback-Leibler estimator,
\begin{equation}
\left(\widehat{a}_0^{MLE}, \widehat{a}_1^{MLE},\widehat{\eta}^{MLE}\right)= \arg \min_{\bm{\theta} \in A}d_{\text{KL}} \left(G_{T^\ast}, F_{T^*}(\cdot|\lambda_1,\lambda_2,\eta)\right).
\label{eq:kl_min}
\end{equation}
with $A=\left\{ (a_0, a_1,\eta) \in \mathbb{R} \times  \mathbb{R}^{-} \times  \mathbb{R}^{+} \right\}$ and 
$\lambda_i=\exp\left(a_0+a_1x_i\right).$

As noted earlier, while the MLE enjoys attractive asymptotic properties, it lacks robustness: even mild contamination can severely affect the estimates. To mitigate this sensitivity, robust alternatives can be employed. Among various approaches, minimum DPD estimators offer an appealing trade-off in practice, balancing efficiency and robustness and allowing for a controlled trade-off through a tuning parameter $\beta$.
The DPD was considered for the first time in \cite{basu1997}, and have proved to perform well for reliability applications \cite{balakrishnan2023b,balakrishnan2023c,balakrishnan2023d,balakrishnan2023e,balakrishnan2024a,balakrishnan2024b,balakrishnan2025,jaenada2025,gonzalez2025}.

Following a similar approach to that used for the KL divergence for mixed distributions, we define the DPD between $G_{T^\ast}(\cdot)$ and $F_{T^\ast}(\cdot|\boldsymbol{\lambda},\eta)$ as a combination of the DPD between both continuous and discrete components as follows
\begin{align}
d_{\beta}\!\left(G_{T^\ast}(\cdot), F_{T^\ast}(\cdot|\boldsymbol{\lambda},\eta)\right) \nonumber
&=
d_{\beta}^{0,\tau_1}
+
d_{\beta}^{\tau_1,\tau_2}
+
d_{\beta}^{\tau_2},\nonumber
\end{align}

where
\begin{align}
d_{\beta}^{0,\tau_1}\nonumber
&=
\int_0^{\tau_1}
\Biggl[
f_{x_1}(t|\lambda_1,\eta)^{\beta+1}
-
\left(1+\frac{1}{\beta}\right)
f_{x_1}(t|\lambda_1,\eta)^{\beta}
g_{x_1}(t)\nonumber
\\
&+
\frac{1}{\beta}
g_{x_1}(t)^{\beta+1}
\Biggr]
\,dt,\nonumber
\\
d_{\beta}^{\tau_1,\tau_2}
&=
\int_{\tau_1}^{\tau_2}
\Biggl[
f_{x_2}(t+h|\boldsymbol{\lambda},\eta)^{\beta+1}
-
\left(1+\frac{1}{\beta}\right)
f_{x_2}(t+h|\boldsymbol{\lambda},\eta)^{\beta}
g_{x_2}(t)\nonumber
\\
&+
\frac{1}{\beta}
g_{x_2}(t)^{\beta+1}
\Biggr]
\,dt,\nonumber
\\
d_{\beta}^{\tau_2}
&=
\Bigl(1-F_{x_2}(\tau_2+h|\boldsymbol{\lambda},\eta)\Bigr)^{\beta+1}
-
\left(1+\frac{1}{\beta}\right)
\Bigl(1-F_{x_2}(\tau_2+h|\boldsymbol{\lambda},\eta)\Bigr)^{\beta}\nonumber
\\
&\times \Bigl(1-G_{x_2}(\tau_2)\Bigr)
+
\frac{1}{\beta}
\Bigl(1-G_{x_2}(\tau_2)\Bigr)^{\beta+1}.\nonumber
\end{align}

Now, given a ordered random sample from the SSALT, $(t_{i:n}, ..., t_{n1+n2:n})$, and substituting the true (unknown) c.d.f. for its empirical estimate, the DPD is given by
\begin{align}
	d_{\beta}\left(G_{T^\ast}(\cdot), F_{T^\ast}(\cdot|\boldsymbol{\lambda},\eta)\right)&\sim 
\int_{0}^{\tau_1} f_{x_1}(t|\lambda_1,\eta)^{\beta+1} dt \nonumber
\\&+ \int_{\tau_1}^{\tau_2} f_{x_2}(t+h|\boldsymbol{\lambda},\eta)^{\beta+1} dt
\nonumber\\
&
+ \Bigl(1-F_{x_2}(\tau_2+h|\boldsymbol{\lambda},\eta)\Bigr)^{\beta+1}
\nonumber\\
&
- \frac{\beta+1}{\beta n} \Biggl\{
\sum_{i=1}^{n_1} f_{x_1}(t_{i:n}|\lambda_1,\eta)^\beta \nonumber
\\
& + \sum_{i=n_1+1}^{n_1+n_2} f_{x_2}(t_{i:n}+h|\boldsymbol{\lambda},\eta)^\beta
\nonumber\\
&
+ (n-n_1-n_2)  \nonumber
\\
& \times \Bigl(1-F_{x_2}(\tau_2+h|\boldsymbol{\lambda},\eta)\Bigr)^{\beta}
\Biggr\}
+ K.   \nonumber
\end{align}
where $K$ is a constant independent of the Weibull parameters. In computing the minimum–divergence estimator, parameter-free terms (i.e., those involving only the true distribution) can be ignored on the optimization; hence the constant $K$ may be omitted from the objective, and the MDPDE  with  tuning parameter $\beta >0$ for a SSALT model 
 is defined as
\begin{equation*}
\left(\hat{\lambda}_1^{\beta},\hat{\lambda}_2^{\beta},\hat{\eta}^{\beta}\right)=arg \min_{\lambda_1, \lambda_2,\eta}H_{n}^{\beta}\left(\lambda_1, \lambda_2,\eta\right),
\end{equation*}
being
\begin{equation} \label{eq:H_n}
H_{n}^{\beta}\left(\lambda_1, \lambda_2,\eta\right)=h_{n}^{1,\beta}\left(\lambda_1, \lambda_2,\eta\right) -\frac{\beta+1}{\beta n} h_{n}^{2,\beta}\left(\lambda_1, \lambda_2,\eta\right) 
\end{equation}
where
\begin{align*}
h_{n}^{1,\beta}\left(\lambda_1, \lambda_2,\eta\right)&=\int_{0}^{\tau_1} f_{x_1}\left(t|\lambda_1,\eta\right)^{\beta+1}dt
\\
&+\int_{\tau_1}^{\tau_2}f_{x_2}\left(t+h|\lambda_1,\lambda_2,\eta\right)^{\beta+1}dt
\\
&+\left(1-F_{x_2}\left(\tau_2+h|\lambda_1,\lambda_2,\eta\right)\right)^{\beta+1},
\end{align*}
and
\begin{align*}
h_{n}^{2,\beta}\left(\lambda_1, \lambda_2,\eta\right)&=\sum_{i=1}^{n_1}f_{x_1}\left(t|\lambda_1,\eta\right)^{\beta}+\sum_{i=n_1+1}^{n_1+n_2} f_{x_2}\left(t|\lambda_1,\lambda_2,\eta\right)^{\beta}  \\
&+\left(n-n_1-n_2\right)\left(1-F_{x_2}\left(\tau_2+h|\lambda_1,\lambda_2,\eta\right)\right)^{\beta} .
\end{align*}
In the limit as $\beta \rightarrow 0$, the DPD reduces to the Kullback-Leiber divergence; consequently, the corresponding MDPDE at $\beta=0$ coincides with the MLE as the limiting case of the MDPDE family.
  
%

For estimation of the SSALT model parameters $\boldsymbol{\theta} = \left(a_0,a_1,\eta\right)^T$, closed-form expressions for the right-hand-side terms in \eqref{eq:H_n} under the Weibull model and log-linear stress–scale relationship are required. Since the Weibull scale parameter is expressed through the regression vector $(a_0,a_1)$, we will henceforth write 
$h_{n}^{i,\beta}\left(a_0,a_1,\eta\right), i =1,2$  to denote the corresponding functions after reparameterization.

Let us first introduce an auxiliary Lemma.
\begin{lemma}
Under Weibull lifetimes, we have
\begin{align}
\zeta_{\alpha,\beta}^{\tau_1}\left(\lambda_1,\lambda_2,\eta\right)&=\int_{0}^{\tau_1} \left(\frac{t}{\lambda_1}\right)^{\alpha} f_{x_1}\left(t|\lambda_1,\eta\right)^{\beta+1}dt  \nonumber \\
&=\left(\frac{\eta}{\lambda_1}\right)^{\beta}\frac{1}{\left(\beta+1\right)^{\frac{\alpha+\left(\eta-1\right)(\beta+1)+1}{\eta}}} \nonumber
\\
&\times \gamma\left(\frac{\alpha+(\eta-1)(\beta+1)+1}{\eta},\left(\frac{\tau_1}{\lambda_1}\right)^{\eta}(\beta+1)\right),
 \\
\zeta_{\alpha,\beta}^{\tau_1,\tau_2}\left(\lambda_1,\lambda_2,\eta\right)&=\int_{\tau_1}^{\tau_2}\left(\frac{t+h}{\lambda_2}\right)^{\alpha}f_{x_2}\left(t+h|\lambda_1,\lambda_2,\eta\right)^{\beta+1}dt \nonumber \\
&=\left(\frac{\eta}{\lambda_2}\right)^{\beta}\frac{1}{\left(\beta+1\right)^{\frac{\alpha+(\eta-1)(\beta+1)+1}{\eta} } }\nonumber
\\ 
&\times \left\{\gamma\left(\frac{\alpha+(\eta-1)(\beta+1)+1}{\eta} , \frac{\left(\tau_2+h\right)^{\eta}}{\lambda_2^{\eta}} (\beta+1) \right) \right. \nonumber \\
&\left. -\gamma \left(\frac{\alpha+(\eta-1)(\beta+1)+1}{\eta},\left(\frac{\tau_1}{\lambda_1}\right)^{\eta}(\beta+1)\right) \right\}.
\end{align}
Where $\gamma(s,a)$ is the incomplete Gamma function:
\begin{align*}
\gamma(s,a)=\int_0^{a} t^{s-1} \exp(-t)dt.
\end{align*}
\end{lemma}
\begin{proof}
See Appendix
\end{proof}

The next result gets the expression for $h_{n}^{1,\beta}\left(\lambda_1,\lambda_2,\eta\right)$ under Weibull lifetimes,
\begin{proposition}
For the Weibull-lifetime we have:
\begin{align*}
h_{n}^{1,\beta}(\lambda_1,\lambda_2,\eta)&=\left(\frac{\eta}{\lambda_1}\right)^{\beta}\frac{1}{(\beta+1)^{\frac{(\beta+1)(\eta-1)+1}{\eta}}}
\\
&\times \gamma\left(\frac{(\beta+1)(\eta-1)+1}{\eta}, \left(\frac{\tau_1}{\lambda_1}\right)^{\eta}(\beta+1)\right)
\\
&+\left(\frac{\eta}{\lambda_2}\right)^{\beta}\frac{1}{(\beta+1)^{\frac{(\beta+)(\eta-1)+1}{\eta}}}
\\
&\times \left\{\gamma\left(\frac{(\beta+1)(\eta-1)+1}{\eta},\frac{\left(\tau_2+h\right)^{\eta}}{\lambda_2^{\eta}}(\beta+1)\right)\right.
\\
&-\left. \gamma\left(\frac{(\eta-1)(\beta+1)+1}{\eta},\left(\frac{\tau_1}{\lambda_1}\right)^{\eta}(\beta+1)\right) \right\}
\\
&+\exp\left[-\frac{\left(\tau_2+h\right)^{\eta}}{\lambda_2^{\eta}}(\beta+1)\right].
\end{align*}

\label{thm:Hn}
\end{proposition}
\begin{proof}
See Appendix
\end{proof}
And the expression for $h_{n}^{2,\beta}(\lambda_1,\lambda_2,\eta)$  is
\begin{align*}
h_{n}^{2,\beta}\left(\lambda_1,\lambda_2,\eta\right)
&=\sum_{i=1}^{n_1}\frac{\eta^{\beta}}{\lambda_1^{\eta \beta}}\,t_i^{\beta\left(\eta-1\right)}
\exp\left[-\beta\left(\frac{t_i}{\lambda_1}\right)^{\eta}\right] \\
& +\sum_{i=n_1+1}^{n_1+n_2} \frac{\eta^{\beta}}{\lambda_1^{\eta \beta}}\left(t_i+h\right)^{\beta\left(\eta-1\right)}\\
&\times\exp\left[-\beta\left(\frac{t_i+h}{\lambda_2}\right)^{\eta}\right] \\
& +\left(n-n_1-n_2\right)\exp\left[-\beta\frac{\left(\tau_2+h\right)^{\eta}}{\lambda_2^\eta} \right].
\end{align*}
\begin{definition}
The MDPDE with tuning parameter $\beta$ for the SSALT model under Weibull lifetime is defined by
\begin{align*}
\left(\hat{\lambda}_1^{\beta},\hat{\lambda}_2^{\beta},\hat{\eta}^{\beta}\right)=arg \min_{\lambda_1,\lambda_2,\eta}H_{n}^{\beta}\left(\lambda_1,\lambda_2,\eta\right).
\end{align*}
\end{definition}
where
\begin{equation}
\begin{aligned}
H_{n}^{\beta}(\lambda_1,\lambda_2,\eta)
&=\left(\frac{\eta}{\lambda_1}\right)^{\beta}
\frac{1}{(\beta+1)^{\frac{(\beta+1)(\eta-1)+1}{\eta}}}\\
&\gamma\!\left(
\frac{(\beta+1)(\eta-1)+1}{\eta},
(\beta+1)\left(\frac{\tau_1}{\lambda_1}\right)^{\eta}
\right)
\\
&+\left(\frac{\eta}{\lambda_1}\right)^{\beta}
\frac{1}{(\beta+1)^{\frac{(\beta+1)(\eta-1)+1}{\eta}}}
\\
&\Bigg[
\gamma\!\left(
\frac{(\beta+1)(\eta-1)+1}{\eta},
\frac{\left(\tau_2+h\right)^{\eta}}{\lambda_2^{\eta}}
(\beta+1)
\right)
\\
&
-\gamma\!\left(
\frac{(\eta-1)(\beta+1)+1}{\eta},
(\beta+1)\left(\frac{\tau_1}{\lambda_1}\right)^{\eta}
\right)
\Bigg]
\\
&+\exp\!\left[
-(\beta+1)\frac{\left(\tau_2+h\right)^{\eta}}{\lambda_2^{\eta}}
\right]
\\
&-\frac{\beta+1}{\beta n}
\Bigg[
\sum_{i=1}^{n_1}
\frac{\eta^{\beta}}{\lambda_1^{\eta\beta}}
t_i^{\beta(\eta-1)}
\exp\!\left[-\beta\left(\frac{t_i}{\lambda_1}\right)^{\eta}\right]
\\
&
+\sum_{i=n_1+1}^{n_1+n_2}
\frac{\eta^{\beta}}{\lambda_1^{\eta\beta}}
(t_i+h)^{\beta(\eta-1)}
\exp\!\left[-\beta\left(\frac{t_i+h}{\lambda_1}\right)^{\eta}\right]
\\
&
+(n-n_1-n_2)
\exp\!\left[
-\beta\frac{\left(\tau_2+h\right)^{\eta}}{\lambda_2^\eta}
\right].   \nonumber
\end{aligned}
\end{equation}

The MDPDE for $(a_0,a_1,\eta)$ is given by:

\begin{align*}
\left(\hat{a}_0^{\beta},\hat{a}_1^{\beta},\hat{\eta}^{\beta}\right)=arg \min_{a_0,a_1,\eta}H_{n}^{\beta}\left(\exp\left(a_0+a_1x_1\right),\exp\left(a_0+a_1x_2\right),\eta\right).
\end{align*}
\section{Asymptotic distribution of the minimum density power divergence estimator \label{sec:asymptotic}}

In this section, the asymptotic distribution of the MDPDEs for the SSALT model is presented.
\begin{theorem}\label{thm:asymp}
The MDPDE with tuning parameter $\beta$ for the SSALT model under Weibull lifetimes satisfies
\begin{equation*}
\sqrt{n}\left(\left(\hat{a}_0^{\beta},\hat{a}_1^{\beta},\hat{\eta}^{\beta}\right)^T-(a_0,a_1,\eta)\right)\xrightarrow[N \to \infty]{\mathscr{L}} \mathcal{N}\left(\bm{0},\bm{\Sigma}\right)
\end{equation*}
where the variance-covariance asymptotic matrix is given by
\begin{equation*}
\boldsymbol{\Sigma}=\boldsymbol{J}_{\beta}^{-1}\boldsymbol{K}_{\beta}\boldsymbol{J}_{\beta}^{-1},
\end{equation*}
the matrix $\boldsymbol{J}_{\beta} = \left(J^{\beta}(a_0,a_1, \eta)\right)_{i,j} i, j=1,2,3$ is the Jacobian matrix of the DPD-loss \eqref{eq:H_n} 
and
\begin{equation*}
\mathbf{K}_{\beta}=\mathbf{J}_{2\beta}- \bm{\xi}_{\beta}\bm{\xi}_{\beta}^{T};
\quad 
\bm{\xi}_{\beta}=\left(
		\xi_1(a_0), 
		\xi_2(a_1), 
		\xi_3(\eta)
\right)^T.
\end{equation*}
Where each component of $\bm{\xi}_{\beta}$ is defined as:
\begin{equation*}
\xi_{\beta}(s)=\int \frac{\partial \log\left[f_{T}(t|\bm{\theta})\right] }{\partial s} f_{T}(t|\bm{\theta})^{\beta+1} dt \quad s \in \left\{a_0,a_1,\eta\right\}
\end{equation*}

\end{theorem}

\begin{proof}
This result is based on Basu et al. \cite{basu1997}.
\end{proof}
In the following propositions we will establish explicit expression of the elements of $\mathbf{J}_{\beta}$ and $\mathbf{\xi}_{\beta}$ for the SSALT model under Weibull lifetimes. All proofs are presented in the Appendix.

\begin{proposition}\label{thm:asympa0}
For the Weibull-lifetime, the value of the element $J_{11}^{\beta}(a_0)$ is
\begin{equation*}
J_{11}^{\beta}(a_0,a_1,\eta) = J_{11}^{\beta}(a_0)=J_{11,\tau_1}^{\beta}(a_0)+J_{11,\tau_1\tau_2}^{\beta}(a_0)+J_{11,\tau_2}^{\beta}(a_0),
\end{equation*}
with
\begin{equation}
\begin{aligned}
J_{11,\tau_1}^{\beta}(a_0)
&=\eta^2\Big[
\zeta_{0,\beta}^{\tau_1}(a_0,a_1,\eta)
+\zeta_{2\eta,\beta}^{\tau_1}(a_0,a_1,\eta)
\\
&-2\zeta_{\eta,\beta}^{\tau_1}(a_0,a_1,\eta)
\Big],
\\
J_{11,\tau_1\tau_2}^{\beta}(a_0)
&=\eta^{2}\Big[
\zeta_{0,\beta}^{\tau_1,\tau_2}(a_0,a_1,\eta)
+\zeta_{2\eta,\beta}^{\tau_1,\tau_2}(a_0,a_1,\eta)
\\
&-2\zeta_{\eta,\beta}^{\tau_1,\tau_2}(a_0,a_1,\eta)
\Big],
\\
J_{11,\tau_2}^{\beta}(a_0)
&=\eta^2
\left(
\frac{\tau_2+\frac{\lambda_2}{\lambda_1}\tau_1-\tau_1}{\lambda_2}
\right)^{2\eta}
\\
&\times
\exp\!\Bigg[
-(\beta+1)
\left(
\frac{\tau_2+\frac{\lambda_2}{\lambda_1}\tau_1-\tau_1}{\lambda_2}
\right)^{\eta}
\Bigg].
\end{aligned}  \nonumber
\end{equation}

\end{proposition}
\begin{proof}
See Appendix
\end{proof}
\begin{proposition}\label{thm:asympa1}
For the Weibull-lifetime, the value of the element $J_{22}^{\beta}(a_1)$ is
\begin{equation*}
J_{22}^{\beta}(a_0,a_1,\eta) = J_{22}^{\beta}(a_1)=J_{22,\tau_1}^{\beta}(a_1)+J_{22,\tau_1\tau_2}^{\beta}(a_1)+J_{22,\tau_2}^{\beta}(a_1),
\end{equation*}
with
\begin{equation}
\begin{aligned}
J_{22,\tau_1}^{\beta}(a_1)
&=x_1^2 J_{11,\tau_1}^{\beta}(a_0),
\\
J_{22,\tau_1\tau_2}^{\beta}(a_1)
&=\eta^2 x_2^2
\Big[
\zeta_{0,\beta}^{\tau_1,\tau_2}
+\zeta_{2\eta,\beta}^{\tau_1,\tau_2}
-2\zeta_{\eta,\beta}^{\tau_1,\tau_2}
\Big]
\\
&
+\frac{\tau_1^2}{\lambda_1^2}(x_2-x_1)^2
\Big[
(\eta-1)^2 \zeta_{-2,\beta}^{\tau_1,\tau_2}
+\eta^2 \zeta_{2(\eta-1),\beta}^{\tau_1,\tau_2}
\\
&
-2\eta(\eta-1)\zeta_{\eta-2,\beta}^{\tau_1,\tau_2}
\Big]
\\
&
+\frac{\tau_1}{\lambda_1}(x_2-x_1)x_2
\Big[
-2\eta(\eta-1)\zeta_{-1,\beta}^{\tau_1,\tau_2}
\\
&
+2\eta(2\eta-1)\zeta_{\eta-1,\beta}^{\tau_1,\tau_2}-2\eta^2\zeta_{2\eta-1,\beta}^{\tau_1,\tau_2}
\Big],
\\
J_{22,\tau_2}^{\beta}(a_1)
&=\eta^2
\left(
\frac{\tau_2+h}{\lambda_2}
\right)^{2\eta-2}
\left(
\frac{\frac{x_1\tau_1\lambda_2}{\lambda_1}
+x_2(\tau_2-\tau_1)}{\lambda_2}
\right)^2
\\
&\times
\exp\!\Bigg[
-(\beta+1)
\left(
\frac{\tau_2+h}{\lambda_2}
\right)^{\eta}
\Bigg].
\end{aligned}  \nonumber
\end{equation}

\end{proposition}
\begin{proof}
See Appendix
\end{proof}

\begin{proposition}\label{thm:asympeta}
For the Weibull-lifetime, the value of the element $J_{33}^{\beta}(\eta)$ is
\begin{equation*}
J_{33}^{\beta}(a_0,a_1,\eta) = J_{33}^{\beta}(\eta)=J_{33,\tau_1}^{\beta}(\eta)+J_{33,\tau_1\tau_2}^{\beta}(\eta)+J_{33,\tau_2}^{\beta}(\eta).
\end{equation*}
Where
\begin{equation}
\begin{aligned}
J_{33,\tau_1}^{\beta}(\eta)
&=\frac{1}{\eta^2}
\Big[
H_{0,0,\beta}^{\tau_1}
+2\eta H_{0,1,\beta}^{\tau_1}
+\eta^2 H_{0,2,\beta}^{\tau_1}
\Big]
\\
&
-\frac{2}{\eta}
\Big[
H_{\eta,1,\beta}^{\tau_1}
+\eta H_{\eta,2,\beta}^{\tau_1}
\Big]
+H_{2\eta,2,\beta}^{\tau_1},
\\[0.6em]
J_{33,\tau_1\tau_2}^{\beta}(\eta)
&=\frac{1}{\eta^2}
\Big[
H_{0,0,\beta}^{\tau_1,\tau_2}
+2\eta H_{0,1,\beta}^{\tau_1,\tau_2}
+\eta^2 H_{0,2,\beta}^{\tau_1,\tau_2}
\Big]
\\
&
-\frac{2}{\eta}
\Big[
H_{\eta,1,\beta}^{\tau_1,\tau_2}
+\eta H_{\eta,2,\beta}^{\tau_1,\tau_2}
\Big]
+H_{2\eta,2,\beta}^{\tau_1,\tau_2},
\\
J_{33,\tau_2}^{\beta}(\eta)
&=
\left(
\frac{\tau_2+\frac{\lambda_2}{\lambda_1}\tau_1-\tau_1}{\lambda_2}
\right)^{2\eta}
\left[
\log\!\left(
\frac{\tau_2+h}{\lambda_2}
\right)
\right]^2
\\
&\times
\exp\!\Bigg[
-(\beta+1)
\left(
\frac{\tau_2+h}{\lambda_2}
\right)^{\eta}
\Bigg].
\end{aligned}  \nonumber
\end{equation}

\end{proposition}
\begin{proof}
See Appendix
\end{proof}
\begin{proposition}
For the Weibull-lifetime, the value of the element $J_{12}^{\beta}(a_0,a_1)$ is
\begin{equation*}
J_{12}^{\beta}(a_0,a_1)=J_{12,\tau_1}^{\beta}(a_0,a_1)+J_{12,\tau_1\tau_2}^{\beta}(a_0,a_1)+J_{12,\tau_2}^{\beta}(a_0,a_1).
\end{equation*}
Where
\begin{equation}
\begin{aligned}
J_{12,\tau_1}^{\beta}(a_0,a_1)
&=\eta^2 x_1
\Big[
\zeta_{0,\beta}^{\tau_1}
+\zeta_{2\eta,\beta}^{\tau_1}
-2\zeta_{\eta,\beta}^{\tau_1}
\Big],
\\
J_{12,\tau_1\tau_2}^{\beta}(a_0,a_1)
&=\eta^2 x_2
\Big[
\zeta_{0,\beta}^{\tau_1,\tau_2}
+\zeta_{2\eta,\beta}^{\tau_1,\tau_2}
-2\zeta_{\eta,\beta}^{\tau_1,\tau_2}
\Big]
\\
&
+\frac{\tau_1}{\lambda_1}(x_2-x_1)
\Big[
-\eta(\eta-1)\zeta_{-1,\beta}^{\tau_1,\tau_2}
\\
&
+\eta(2\eta-1)\zeta_{\eta-1,\beta}^{\tau_1,\tau_2}-\eta^2\zeta_{2\eta-1,\beta}^{\tau_1,\tau_2}
\Big],
\\
J_{12,\tau_2}^{\beta}(a_0,a_1)
&=\eta^2
\left(
\frac{\tau_2+h}{\lambda_2}
\right)^{2\eta-1}
\frac{\frac{x_1\tau_1\lambda_2}{\lambda_1}
+x_2(\tau_2-\tau_1)}{\lambda_2}
\\
&\quad\times
\exp\!\Bigg[
-(\beta+1)
\left(
\frac{\tau_2+h}{\lambda_2}
\right)^{\eta}
\Bigg].
\end{aligned}  \nonumber
\end{equation}

\end{proposition}
\begin{proof}
See Appendix
\end{proof}
In order to simplify notation e introduce two functions previously defined in Basu et al. \cite{basu2017}.

\begin{lemma}
	For the Weibull-lifetime model, we have:
\begin{equation}
\begin{aligned}
H_{\alpha,\gamma,\beta}^{\tau_1}(a_0,a_1,\eta)
&:=\int_{0}^{\tau_1}
\left(\frac{t}{\lambda_1}\right)^{\alpha}
\left[
\log\!\left(\frac{t}{\lambda_1}\right)
\right]^{\gamma}
g_{x_1}(t)^{\beta+1}\,dt
\\
&=\lambda_1
\left(\frac{\eta}{\lambda_1}\right)^{\beta+1}
\int_{0}^{\frac{\tau_1}{\lambda_1}}
\Big[l^{\alpha+(\eta-1)(\beta+1)}
\big[\log(l)\big]^{\gamma}
\\
&\times
\exp\!\Big[
-(\beta+1)l^{\eta}
\Big]\,\Big]dl ,
\end{aligned}  \nonumber
\end{equation}
\begin{equation}
\begin{aligned}
H_{\alpha,\gamma,\beta}^{\tau_1,\tau_2}(a_0,a_1,\eta)
&:=\int_{\tau_1}^{\tau_2}\Big[
\left(
\frac{t+\frac{\lambda_2}{\lambda_1}\tau_1-\tau_1}{\lambda_2}
\right)^{\alpha}
\\
&\times
\left[
\log\!\left(
\frac{t+h}{\lambda_2}
\right)
\right]^{\gamma}
g_{x_2}(t)^{\beta+1}\,\Big]dt
\\
&=\lambda_2
\left(\frac{\eta}{\lambda_2}\right)^{\beta+1}
\int_{\frac{\tau_1}{\lambda_1}}^{\frac{\tau_2+h}{\lambda_2}}
\Big[l^{\alpha+(\eta-1)(\beta+1)}
\big[\log(l)\big]^{\gamma}
\\
&\times
\exp\!\Big[
-(\beta+1)l^{\eta}
\Big]\,\Big]dl .
\end{aligned}  \nonumber
\end{equation}

\end{lemma}
\begin{proof}
See Appendix
\end{proof}
\begin{proposition}
For the Weibull-lifetime, the value of the element $J_{13}^{\beta}(a_0,\eta)$ is
\begin{equation*}
J_{13}^{\beta}(a_0,\eta)=J_{13,\tau_1}^{\beta}(a_0,\eta)+J_{13,\tau_1\tau_2}^{\beta}(a_0,\eta)+J_{13,\tau_2}^{\beta}(a_0,\eta).
\end{equation*}
Where
\begin{equation}
\begin{aligned}
J_{13,\tau_1}^{\beta}(a_0,\eta)
&=
-\zeta_{0,\beta}^{\tau_1}
+\zeta_{\eta,\beta}^{\tau_1}
\\
&
-\eta
\Big[
H_{0,1,\beta}^{\tau_1}
-2 H_{\eta,1,\beta}^{\tau_1}
+H_{2\eta,1,\beta}^{\tau_1}
\Big],
\\
J_{13,\tau_1\tau_2}^{\beta}(a_0,\eta)
&=
-\zeta_{0,\beta}^{\tau_1,\tau_2}
+\zeta_{\eta,\beta}^{\tau_1,\tau_2}
\\
&
-\eta
\Big[
H_{0,1,\beta}^{\tau_1,\tau_2}
-2 H_{\eta,1,\beta}^{\tau_1,\tau_2}
+H_{2\eta,1,\beta}^{\tau_1,\tau_2}
\Big],
\\
J_{13,\tau_2}^{\beta}(a_0,\eta)
&=
-\eta
\log\!\left(
\frac{\tau_2+h}{\lambda_2}
\right)
\left(
\frac{\tau_2+h}{\lambda_2}
\right)^{2\eta}
\\
&\times
\exp\!\Bigg[
-(\beta+1)
\left(
\frac{\tau_2+h}{\lambda_2}
\right)^{\eta}
\Bigg].
\end{aligned}
\end{equation}

\end{proposition}
\begin{proof}
See Appendix
\end{proof}
\begin{proposition}
For the Weibull-lifetime, the value of the element $J_{23}^{\beta}(a_1,\eta)$ is
\begin{equation*}
J_{23}^{\beta}(a_1,\eta)=J_{23,\tau_1}^{\beta}(a_1,\eta)+J_{23,\tau_1\tau_2}^{\beta}(a_1,\eta)+J_{23,\tau_2}^{\beta}(a_1,\eta).
\end{equation*}
Where
\begin{equation}
\begin{aligned}
J_{23,\tau_1}^{\beta}(a_1,\eta)
&=x_1 J_{13,\tau_1}^{\beta}(a_0,\eta)
\\
&=x_1
\Big[
-\zeta_{0,\beta}^{\tau_1}
+\zeta_{\eta,\beta}^{\tau_1}
\\
&-\eta
\big(
H_{0,1,\beta}^{\tau_1}
-2H_{\eta,1,\beta}^{\tau_1}
+H_{2\eta,1,\beta}^{\tau_1}
\big)
\Big],
\\
J_{23,\tau_1\tau_2}^{\beta}(a_1,\eta)
&=
x_2
\Big[
-\zeta_{0,\beta}^{\tau_1,\tau_2}
+\zeta_{\eta,\beta}^{\tau_1,\tau_2}
\\
&-\eta
\big(
H_{0,1,\beta}^{\tau_1,\tau_2}
-2H_{\eta,1,\beta}^{\tau_1,\tau_2}
+H_{2\eta,1,\beta}^{\tau_1,\tau_2}
\big)
\Big]
\\
&
+\frac{\tau_1}{\lambda_1}(x_2-x_1)
\Big[
\frac{\eta-1}{\eta}\zeta_{-1,\beta}^{\tau_1,\tau_2}
-\zeta_{\eta-1,\beta}^{\tau_1,\tau_2}
\\
&
+(\eta-1)H_{-1,1,\beta}^{\tau_1,\tau_2}
-(2\eta-1)H_{\eta-1,1,\beta}^{\tau_1,\tau_2}
+\eta H_{2\eta-1,1,\beta}^{\tau_1,\tau_2}
\Big],
\\
J_{23,\tau_2}^{\beta}(a_1,\eta)
&=
-\eta
\frac{\frac{x_1\lambda_2\tau_1}{\lambda_1}
+x_2(\tau_2-\tau_1)}{\lambda_2}
\log\!\left(
\frac{\tau_2+h}{\lambda_2}
\right)
\\
&\times
\left(
\frac{\tau_2+h}{\lambda_2}
\right)^{2\eta-1}\exp\!\Bigg[
-(\beta+1)
\left(
\frac{\tau_2+h}{\lambda_2}
\right)^{\eta}
\Bigg].
\end{aligned}  \nonumber
\end{equation}

\end{proposition}
\begin{proof}
See Appendix
\end{proof}

The following propositions present explicit expressions of the vector $\boldsymbol{\xi}.$
\begin{proposition}
For the Weibull-lifetime, the value of the element $\xi_{1}^{\beta}(a_0)$ is
\begin{equation*}
\xi_{1}^{\beta}(a_0)=\xi_{1,\tau_1}^{\beta}(a_0)+\xi_{1,\tau_1\tau_2}^{\beta}(a_0)+\xi_{1,\tau_2}^{\beta}(a_0),
\end{equation*}
where
\begin{align*}
\xi_{1,\tau_1}^{\beta}(a_0)&=\eta\left\{-\zeta_{0,\beta}^{\tau_1}(a_0,a_1,\eta)+\zeta_{\eta,\beta}^{\tau_1}(a_0,a_1,\eta)\right\},
\\
\xi_{1,\tau_1\tau_2}^{\beta}(a_0)&=\eta\left\{-\zeta_{0,\beta}^{\tau_1,\tau_2}(a_0,a_1,\eta)+\zeta_{\eta,\beta}^{\tau_1,\tau_2}(a_0,a_1,\eta)\right\},
\\
\xi_{1,\tau_2}^{\beta}(a_0)&=\eta \left(\frac{\tau_2+h}{\lambda_2}\right)^{\eta} \exp\left[-\left(\frac{\tau_2+h}{\lambda_2}\right)^{\eta}(\beta+1)\right].
\end{align*}
\end{proposition}
\begin{proof}
See Appendix
\end{proof}
\begin{proposition}
For the Weibull-lifetime, the value of the element $\xi_{2}^{\beta}(a_1)$ is
\begin{equation*}
\xi_{2}^{\beta}(a_1)=\xi_{2,\tau_1}^{\beta}(a_1)+\xi_{2,\tau_1\tau_2}^{\beta}(a_1)+\xi_{2,\tau_2}^{\beta}(a_1).
\end{equation*}
Where
\begin{equation}
\begin{aligned}
\xi_{2,\tau_1}^{\beta}(a_1)
&=\eta x_1
\Big[
-\zeta_{0,\beta}^{\tau_1}
+\zeta_{\eta,\beta}^{\tau_1}
\Big],
\\
\xi_{2,\tau_1\tau_2}^{\beta}(a_1)
&=
\eta x_2
\Big[
-\zeta_{0,\beta}^{\tau_1,\tau_2}
+\zeta_{\eta,\beta}^{\tau_1,\tau_2}
\Big]
\\
&
+\frac{\tau_1}{\lambda_1}(x_2-x_1)
\Big[
(\eta-1)\zeta_{-1,\beta}^{\tau_1,\tau_2}
-\eta\zeta_{\eta-1,\beta}^{\tau_1,\tau_2}
\Big],
\\
\xi_{2,\tau_2}^{\beta}(a_1)
&=
-\eta
\left(
\frac{\tau_2+h}{\lambda_2}
\right)^{\eta}
\frac{\frac{x_1\tau_1}{\lambda_1}
+x_2(\tau_2-\tau_1)}{\lambda_2}
\\
&\times
\exp\!\Bigg[
-(\beta+1)
\left(
\frac{\tau_2+h}{\lambda_2}
\right)^{\eta}
\Bigg].
\end{aligned}  \nonumber
\end{equation}

\end{proposition}
\begin{proof}
See Appendix
\end{proof}
\begin{proposition}
For the Weibull-lifetime, the value of the element $\xi_{3}^{\beta}(\eta)$ is
\begin{equation*}
\xi_{3}^{\beta}(\eta)=\xi_{3,\tau_1}^{\beta}(\eta)+\xi_{3,\tau_1\tau_2}^{\beta}(\eta)+\xi_{3,\tau_2}^{\beta}(\eta),
\end{equation*}
where
\begin{equation}
\begin{aligned}
\xi_{3,\tau_1}^{\beta}(\eta)&=\frac{1}{\eta}\zeta_{0,\beta}^{\tau_1}(a_0,a_1,\eta)+H_{0,1,\beta}^{\tau_1}(a_0,a_1,\eta)
\\&-H_{\eta,1,\beta}^{\tau_1}(a_0,a_1,\eta),
\\
\xi_{3,\tau_1\tau_2}^{\beta}(\eta)&=\frac{1}{\eta}\zeta_{0,\beta}^{\tau_1,\tau_2}(a_0,a_1,\eta)+H_{0,1,\beta}^{\tau_1,\tau_2}(a_0,a_1,\eta)
\\&-H_{\eta,1,\beta}^{\tau_1,\tau_2}(a_0,a_1,\eta),
\\
\xi_{3,\tau_2}^{\beta}(\eta)&=-\log\left(\frac{\tau_2+h}{\lambda_2}\right)\left(\frac{\tau_2+h}{\lambda_2}\right)^{\eta}
\\
&\times \exp\left[-\left(\frac{\tau_2+h}{\lambda_2}\right)^{\eta}(\beta+1)\right].
\end{aligned}  \nonumber
\end{equation}
\end{proposition}
\begin{proof}
See Appendix
\end{proof}

\section{Lifetime Characteristics Estimation \label{sec:CI}}

Estimates for the model parameters $a_0$, $a_1$, and $\eta$ have been derived, along with their corresponding asymptotic distributions. As \cite{balakrishnan2025} shows these parameter estimates allow for the approximation of the cumulative distribution function (c.d.f.) of the lifetime, as well as related measures. Nonetheless, in many reliability analyses, the main focus is on quantifying specific features of the lifetime, such as the mean time to failure (MTTF),rather than on the c.d.f. itself. In what follows, we explore the robust estimation of essential lifetime attributes using the proposed MDPDEs, and we outline the development of corresponding confidence intervals (CIs) based on the asymptotic behavior of the estimators and the application of the Delta method.

The reliability of a system at a specified mission time $t$ can be assessed by evaluating the complement of the failure probability. In other words, assuming standard operating conditions, the reliability of a component at mission time $t$ is defined as the reliability function, under normal operating conditions, and is given by
\begin{equation}
R_0(t)=\exp\left(-\left(\frac{t}{\lambda_0}\right)^{\eta}\right)
=\exp\big(-t^{\eta}\exp\{-\eta(a_0+a_1x_0)\}\big),  \nonumber
\end{equation}
and the corresponding MDPDE-based estimator can be written as
\begin{equation}
\widehat{R}_0^{\beta}(t)
=\exp\left(-\left(\frac{t}{\widehat{\lambda}^{\beta}_0}\right)^{\widehat{\eta}^{\beta}}\right)
=\exp\big(-t^{\widehat{\eta}^{\beta}}\exp\{-\widehat{\eta}^{\beta}(\widehat{a}_0^{\beta}+\widehat{a}_1^{\beta}x_0)\}\big).  \nonumber
\end{equation}

Following the asymptotic properties of the MDPDE seen in theorem (\ref{thm:asymp}),
$\widehat{\boldsymbol{\theta}}^{\beta}=(\widehat{a}_0^{\beta},\widehat{a}_1^{\beta},\widehat{\eta}^{\beta})$
is asymptotically normal. Consequently,
\begin{equation}
\sqrt{n}\left(\widehat{R}_0^{\beta}(t)-R_0(t)\right)
\overset{\mathscr{L}}{\longrightarrow}
\mathcal{N}\left(0,\sigma_{\beta}(R_0(t))^2\right),  \nonumber
\end{equation}
where
\begin{equation}
\sigma_{\beta}(R_0(t))^2
=\nabla h_R(\boldsymbol{\theta}_0)^{\top}
\boldsymbol{J}^{-1}_{\beta}(\boldsymbol{\theta}_0)
\boldsymbol{K}_{\beta}(\boldsymbol{\theta}_0)
\boldsymbol{J}^{-1}_{\beta}(\boldsymbol{\theta}_0)
\nabla h_R(\boldsymbol{\theta}_0).  \nonumber
\end{equation}
The expressions of $\nabla h_R(\boldsymbol{\theta}_0)$ is given by
\begin{align*}
\nabla h_R(\boldsymbol{\theta}_0)&=\left(\frac{t}{\lambda_0}\right)^{\eta}R_0(t)\left(\eta,\eta x_0,-\log\left(\frac{t}{\lambda_0}\right)\right),  \nonumber
\end{align*}
and $\boldsymbol{J}_{\beta}(\boldsymbol{\theta})$ and $\boldsymbol{K}_{\beta}(\boldsymbol{\theta})$ where defined in theorem (\ref{thm:asymp}).

For a reliability level $1-\alpha$, the associated lifetime quantile is
\begin{equation}
Q_{1-\alpha}=R_0^{-1}(1-\alpha)
=-\big(\log(1-\alpha)\big)^{1/\eta}\lambda_0,  \nonumber
\end{equation}
with MDPDE estimator
\begin{equation}
\widehat{Q}^{\beta}_{1-\alpha}
=-\big(\log(1-\alpha)\big)^{1/\widehat{\eta}^{\beta}}
\widehat{\lambda}^{\beta}_0.  \nonumber
\end{equation}
Similarly,
\begin{equation}
\sqrt{n}\left(\widehat{Q}^{\beta}_{1-\alpha}-Q_{1-\alpha}\right)
\overset{\mathscr{L}}{\longrightarrow}
\mathcal{N}\left(0,\sigma_{\beta}(Q_{1-\alpha})^2\right),  \nonumber
\end{equation}
where $\sigma_{\beta}(Q_{1-\alpha})^2$ is defined as
	\begin{align*}
	\sigma_{\beta}\left( Q_{1-\alpha}\right)^2=\nabla h_Q(\boldsymbol{\theta}_0)^{T} \boldsymbol{J}^{-1}_{\beta}(\boldsymbol{\theta}_0)\boldsymbol{K}_\beta(\boldsymbol{\theta}_0)\boldsymbol{J}^{-1}_\beta(\boldsymbol{\theta}_0)\nabla h_Q(\boldsymbol{\theta}_0)
	\end{align*}
	and $\boldsymbol{K}_\beta(\boldsymbol{\theta_0})$ and $\boldsymbol{J}_\beta(\boldsymbol{\theta}_0)$ are defined in theorem \ref{thm:asymp}, and 
	\begin{align*}
	\nabla h_Q(\boldsymbol{\theta}_0)^{\top}=Q_{1-\alpha}\left( 1, x_0,\frac{\log(\log(1-\alpha))}{\eta}\right)
	\end{align*}
	is the gradient of the function $h_Q(\boldsymbol{\theta}) = Q_{1-\alpha} =- \left(\log\left(1-\alpha\right)\right)^{1}{\eta}\lambda_0$.

For Weibull lifetimes, the mean time to failure (MTTF) is
\begin{equation}
E_T=\lambda_0\Gamma\left(1+\frac{1}{\eta}\right)
=\exp(a_0+a_1x_0)\Gamma\left(1+\frac{1}{\eta}\right),
\end{equation}
and its robust estimator is obtained by substitution. Its asymptotic distribution is
\begin{equation}
\sqrt{n}\left(\widehat{E}_T^{\beta}-E_T\right)
\overset{\mathscr{L}}{\longrightarrow}
\mathcal{N}\left(0,\sigma_{\beta}(E_T)^2\right).
\end{equation}
where $\sigma_{\beta}(E_T)^2$ is defined as:
\begin{align*}
\sigma_{\beta}\left( E_{T}\right)^2=\nabla h_E(\boldsymbol{\theta}_0)^{T}, \boldsymbol{J}^{-1}_\beta(\boldsymbol{\theta}_0)\boldsymbol{K}_\beta(\boldsymbol{\theta}_0)\boldsymbol{J}^{-1}_\beta(\boldsymbol{\theta}_0)\nabla h_E(\boldsymbol{\theta}_0)
\end{align*}
and $\boldsymbol{K}_\beta(\boldsymbol{\theta})$ and $\boldsymbol{J}_\beta(\boldsymbol{\theta})$ are defined in theorem \ref{thm:asymp}, and 
\begin{equation}
\nabla h_E(\boldsymbol{\theta}_0)^{\top}=E_{T}\left( 1, x_0,-\frac{1}{\eta^2}\psi\left(1+\frac{1}{\eta}\right)\right)
\end{equation}
is the gradient of the function $E_{T}$ where $\psi\left(1+\frac{1}{\eta}\right) $ is the digamma function, de derivative of $\Gamma\left(1+\frac{1}{\eta}\right)$.

Approximate $(1-\alpha)\cdot100\%$ confidence intervals are therefore given by
\begin{equation}
\widehat{R}_0^\beta(t) \pm z_{\alpha/2}\frac{\sigma_{\beta}(R_0(t))}{\sqrt{n}}, \,
\widehat{Q}^\beta_{1-\alpha} \pm z_{\alpha/2}\frac{\sigma_{\beta}(Q_{1-\alpha})}{\sqrt{n}}, \,
\widehat{E}_T^{\beta} \pm z_{\alpha/2}\frac{\sigma_{\beta}(E_T)}{\sqrt{n}}.
\end{equation}

Since $R_0(t)\in[0,1]$, a logit transformation may be applied to avoid boundary issues. Writing 
$\phi_R = \log\left(\widehat{R}_0^\beta(t)/(1-\widehat{R}_0^\beta(t))\right)$, the CI on the logit scale is
\[
\phi_R \pm z_{\alpha/2}\frac{\sigma_{\beta}(R_0(t))}{\widehat{R}_0^\beta(t)(1-\widehat{R}_0^\beta(t))\sqrt{n}},
\]
and inverting back yields
\begin{equation}
\left[
\frac{\widehat{R}_0^\beta(t)}{\widehat{R}_0^\beta(t)+(1-\widehat{R}_0^\beta(t))S},
\;
\frac{\widehat{R}_0^\beta(t)}{\widehat{R}_0^\beta(t)+(1-\widehat{R}_0^\beta(t))/S}
\right],
\end{equation}
where
\[
S=\exp\left(\frac{z_{\alpha/2}}{\sqrt{n}}\frac{\sigma_{\beta}(R_0(t))}{\widehat{R}_0^\beta(t)(1-\widehat{R}_0^\beta(t))}\right).
\]

For $Q_{1-\alpha}$ and $E_T$, which must be positive, we instead apply the logarithmic transformation:
\[
\log(\widehat{Q}_{1-\alpha}^\beta)\pm z_{\alpha/2}\frac{\sigma_{\beta}(Q_{1-\alpha})}{\widehat{Q}_{1-\alpha}^\beta\sqrt{n}}, \,
\log(\widehat{E}_T^\beta)\pm z_{\alpha/2}\frac{\sigma_{\beta}(E_T)}{\widehat{E}_T^\beta\sqrt{n}},
\]
and exponentiating gives
\begin{equation}
\widehat{Q}_{1-\alpha}^\beta\exp\!\left(\pm\frac{z_{\alpha/2}}{\sqrt{n}}
\frac{\sigma_{\beta}(Q_{1-\alpha})}{\widehat{Q}_{1-\alpha}^\beta}\right),\,
\widehat{E}_T^\beta\exp\!\left(\pm\frac{z_{\alpha/2}}{\sqrt{n}}
\frac{\sigma_{\beta}(E_T)}{\widehat{E}_T^\beta}\right).
\end{equation}

\section{Simulation Study}
This section examines the robustness of the MDPDEs for the SSALT model under a Weibull lifetime distribution in the presence of outliers.
A simple SSALT and type- I censoring type will be simulate as follows:
A total of for a total of $n=200$ devices having Weibull lifetime distributions are set under a SSALT. During the first step, all units under test are subjected to a higher-than normal stress of $x_1=1$, and in the second step, to a stress of $x_2=2$. We take the nominal operating stress to be \(x_0=0.5\), i.e., one half of the first stress level and one quarter of the second.
The point of stress change is set at
$\tau_1=3$ and the test completion time is at $\tau_2=5.$ 
The true parameters of the SSALT model are set as $\boldsymbol{\theta} = (a_0, a_1, \eta)^T=(2, -0.8, 5.5)^T.$ 

To evaluate the robustness of the estimators, we contaminate the data by introducing an $\varepsilon$-fraction of outliers, with
$ \varepsilon \in \{0, 0.03, 0.05, 0.07, 0.08, 0.09, 0.10\}.$
The outliers are generated from a Weibull distribution with contaminated parameters $(\tilde{a}_0, \tilde{a}_1, \tilde{\eta})^T$, truncated to the interval $(0,1.5)$. 
We consider three contamination schemes, each targeting a single component of the SSALT model
\(\boldsymbol{\theta}=(a_0,a_1,\eta)\): (i) the intercept \(a_0\) (scale-at-reference stress),
(ii) the slope \(a_1\) (stress–scale sensitivity), and (iii) the Weibull shape \(\eta\).
For each scheme, the two non–contaminated parameters are fixed at their true values, and the contaminated parameter is determined as the
solution of the calibration equation 
\begin{align*}
\varepsilon = \mathbb{P}(0<W<1.5|\tilde{a}_0, \tilde{a}_1, \tilde{\eta}).
 \end{align*}
  Where $W  \sim \text{ Weibull}(\exp(\tilde{a}_0 + \tilde{a}_1 x_1), \tilde{\eta}),$
that is, the probability of failure within the interval $(0,1.5) $ must be equal to the proportion of outliers. 
\subsection{Performance of the MPDE}

To assess the performance of the MDPDEs, we compute the root mean squared error (RMSE) of the estimators for several values of the tuning parameter
\(\beta \in \{0,\,0.2,\,0.4,\,0.6,\,0.8,\,1\}\). As noted in Section \ref{sec:MDPDE}, the case \(\beta=0\) coincides with the MLE.
Let \(\widehat{\boldsymbol{\theta}}^{(r)}(\alpha)
=(\widehat{a}_0^{(r)}(\alpha),\,\widehat{a}_1^{(r)}(\alpha),\,\widehat{\beta}^{(r)}(\alpha))\)
denote the estimator in replication \(r=1,\ldots,M\).

Figure \ref{fig:graf-outlier} shows how the error varies in the contaminated parameter. The abscissa axis shows the proportion of outliers and the ordinate axis the RMSE. We observe that, in the absence of outliers, the MLE outperforms all other estimators. In contrast, even a small contamination level (e.g., \(3\%\)) markedly affects the MLE, causing it to underperform relative to the robust alternatives.
Moreover, the calibration role of the tuning parameter becomes evident: larger values of \(\beta\) yield better performance in the presence of outliers, whereas smaller values of \(\beta\) produce more efficient estimates in uncontaminated settings. Thus, there is an inherent trade-off between desired performance with and without contamination in terms of efficiency and robustness.

\begin{figure}[htb]
  \centering
  \includegraphics[width=1\textwidth]{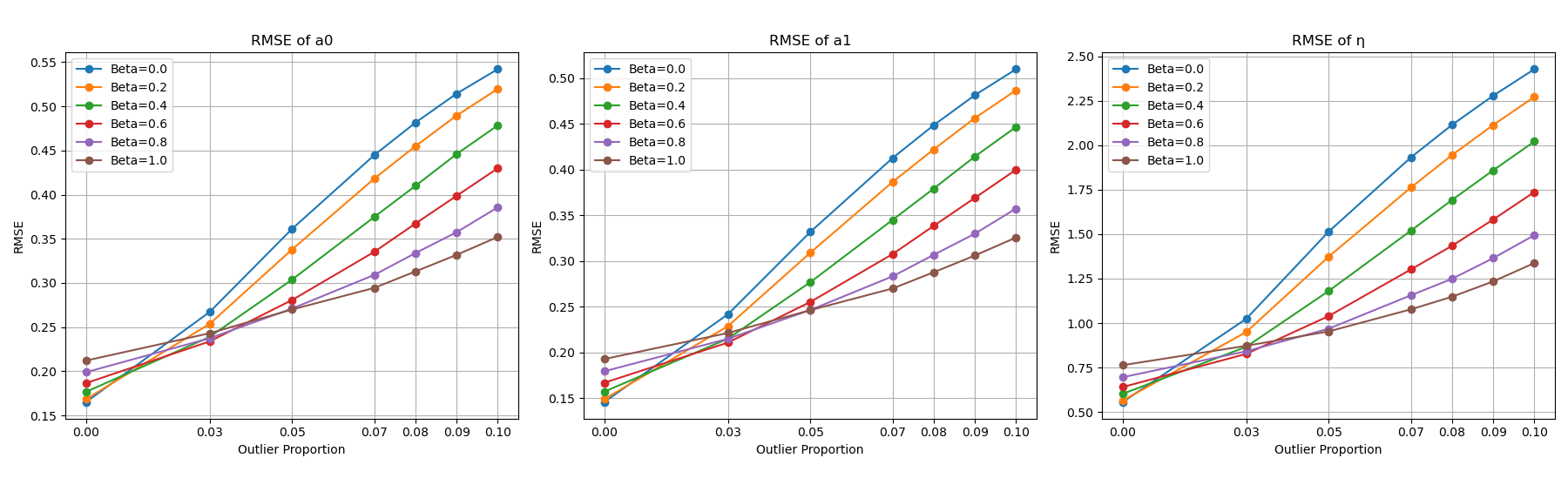}
  \caption{RMSE of  the MDPDEs under increasing contamination proportion for the estimation of a0, a1, eta}
  \label{fig:graf-outlier}
\end{figure}

Furthermore, by inverting the log–linear parametrization  \(\lambda(x)=\exp(a_0+a_1 x)\), we can also compute the RMSE for the scale parameters \(\lambda(x_i)\) at the
stress levels \(x_1=1\), \(x_2=2\), as well as at a nominal operating level \(x_0=0.5\).
Figure~\ref{fig:graf-outlier-lambda} reports the RMSE for the different \(\lambda\) values when the contaminated parameter is the regression coefficient \(a_1\).
The qualitative conclusions mirror those obtained for the model parameters.

\begin{figure}[htb]
  \centering
  \includegraphics[width=1\textwidth]{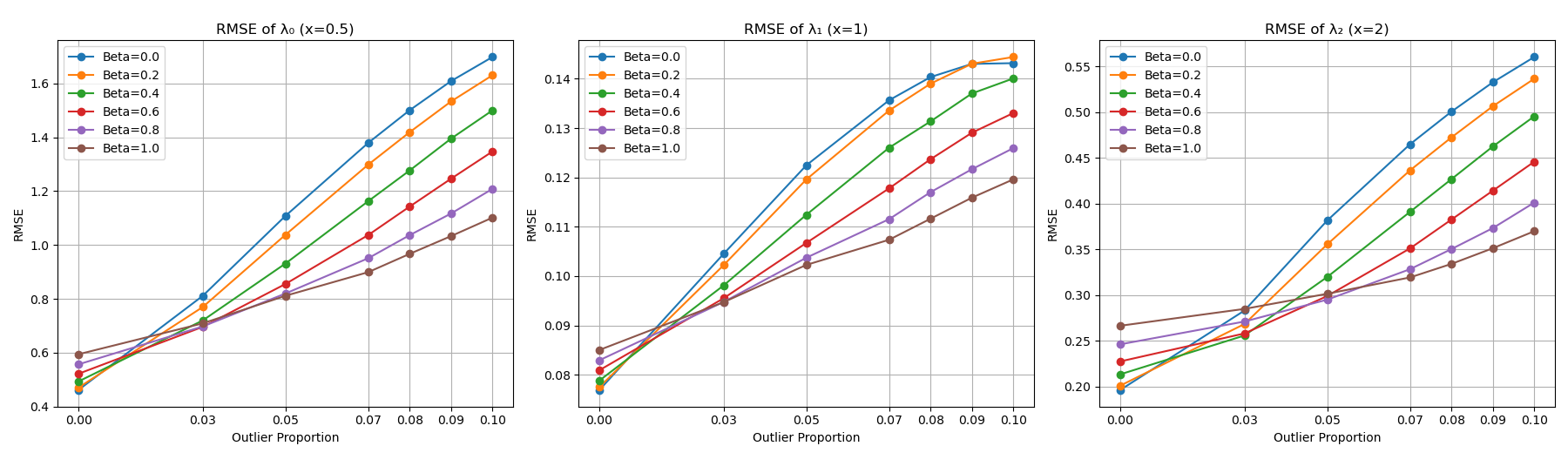}
  \caption{RMSE of $\lambda_i$ when the contaminated parameter is $a_1$.}
  \label{fig:graf-outlier-lambda}
\end{figure}

As discussed in Section \ref{sec:CI}, given the estimates of the SSALT model, we can estimate some values of interest such as the MTTF, the reliability at a certain mission time or the median time of failure. Figure  \ref{fig:graf-outlier-charfunctions} shows this values under normal operating conditions $x_0=0.5$ when the contaminated parameter is $a_1$ at a mission time $t=2.$
For this simulation, we also computed confidence intervals for the model parameters under parameter-specific contamination, as well as confidence intervals (on both the original and transformed methods) for the quantities of interest. Empirical coverage of the confidence intervals are presented in Table \ref{tab:table_CI_parameters}.
In general, the presence of outliers substantially reduces empirical coverage for all estimates. 
However,  intervals produced by MDPDE with large values of the tuning parameter \(\beta\) keep empirical coverage  markedly higher than those based on the MLE. In fact as it is shown in Tables \ref{tab:table_CI_MTTF} and \ref{tab:table_CI_survival}, for the median and the \(\mathrm{MTTF}\), the MDPDE attains satisfactory coverage for both transformed and direct intervals.

\begin{figure}[htb]
	\centering
	\includegraphics[width=1\textwidth]{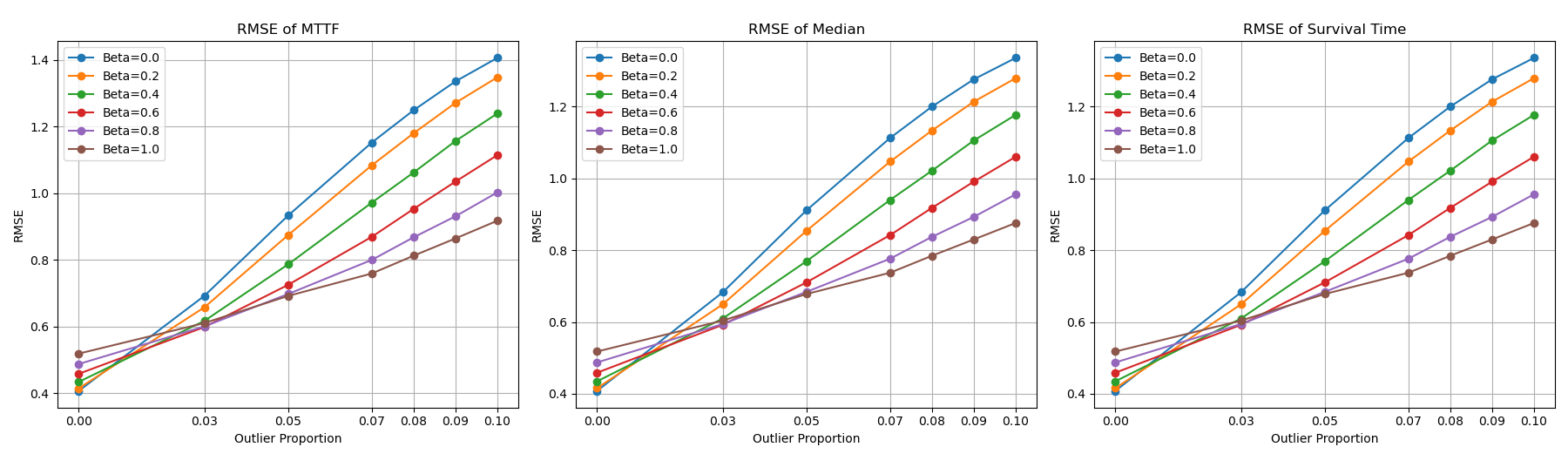}
	\caption{RMSE for the MTTF (top), the Reability at t=2 (center) and the median (bottom) when the contaminated parameter is $a_1$.}
	\label{fig:graf-outlier-charfunctions}
\end{figure}

\begin{table}[htbp]
\centering
\caption{Coverage of the CI for $(a_0,a_1,\eta)$ under contamination in each parameter. Each cell shows the column vector $\left( a_0 ,a_1 , \eta \right)^{\intercal}$.}
\label{tab:table_CI_parameters}
\renewcommand{\arraystretch}{1.15} 
\small

\begin{tabular}{ccccccc}
\hline
proportion $\setminus \beta$ 
& MLE & 0.2 & 0.4 & 0.6 & 0.8 & 1 \\ 
\hline

0 &
$\begin{matrix}95.9\\95.7\\94.7\end{matrix}$ &
$\begin{matrix}95.9\\95.8\\94.9\end{matrix}$ &
$\begin{matrix}95.3\\95.4\\95\end{matrix}$ &
$\begin{matrix}94.7\\94.6\\95.2\end{matrix}$ &
$\begin{matrix}94.2\\94.1\\95.3\end{matrix}$ &
$\begin{matrix}94.4\\94.2\\94.8\end{matrix}$
\\[3pt]\hline

0.03 &
$\begin{matrix}76.5\\75.4\\60.4\end{matrix}$ &
$\begin{matrix}81.4\\81.3\\72.4\end{matrix}$ &
$\begin{matrix}87\\86.7\\83.9\end{matrix}$ &
$\begin{matrix}89.5\\89.3\\90.1\end{matrix}$ &
$\begin{matrix}90.2\\90.5\\91.7\end{matrix}$ &
$\begin{matrix}91.7\\91.4\\92.5\end{matrix}$
\\[3pt]\hline

0.05 &
$\begin{matrix}54.1\\49.5\\9.9\end{matrix}$ &
$\begin{matrix}62.6\\59.9\\27.8\end{matrix}$ &
$\begin{matrix}75.4\\73.9\\57.3\end{matrix}$ &
$\begin{matrix}81.7\\80.2\\75.6\end{matrix}$ &
$\begin{matrix}84.8\\84.1\\84.1\end{matrix}$ &
$\begin{matrix}86.5\\86.8\\87.6\end{matrix}$
\\[3pt]\hline

0.07 &
$\begin{matrix}32.3\\24.3\\0.3\end{matrix}$ &
$\begin{matrix}43.5\\36.8\\5.4\end{matrix}$ &
$\begin{matrix}59.4\\56.1\\28.3\end{matrix}$ &
$\begin{matrix}73.3\\71.6\\57.3\end{matrix}$ &
$\begin{matrix}79.9\\79.6\\74.5\end{matrix}$ &
$\begin{matrix}83.9\\83.3\\83.2\end{matrix}$
\\[3pt]\hline

0.08 &
$\begin{matrix}23.8\\16.8\\0\end{matrix}$ &
$\begin{matrix}33.3\\26.3\\1.4\end{matrix}$ &
$\begin{matrix}52.4\\46.1\\16.5\end{matrix}$ &
$\begin{matrix}67.4\\65.2\\45.8\end{matrix}$ &
$\begin{matrix}76.8\\75.8\\67.7\end{matrix}$ &
$\begin{matrix}81.6\\80.4\\79.3\end{matrix}$
\\[3pt]\hline

0.09 &
$\begin{matrix}17.2\\11\\0\end{matrix}$ &
$\begin{matrix}26.3\\19.6\\0.1\end{matrix}$ &
$\begin{matrix}43.2\\36.8\\8.7\end{matrix}$ &
$\begin{matrix}60\\56.3\\34.5\end{matrix}$ &
$\begin{matrix}73.6\\72\\60.3\end{matrix}$ &
$\begin{matrix}79.1\\78.5\\73.7\end{matrix}$
\\[3pt]\hline

0.1 &
$\begin{matrix}12.5\\7.5\\0\end{matrix}$ &
$\begin{matrix}20.1\\13.5\\0\end{matrix}$ &
$\begin{matrix}34.9\\29\\3.9\end{matrix}$ &
$\begin{matrix}53.7\\49.8\\23.5\end{matrix}$ &
$\begin{matrix}67.2\\65.4\\50.8\end{matrix}$ &
$\begin{matrix}77.1\\75.9\\67.7\end{matrix}$
\\[3pt]\hline

\end{tabular}
\end{table}

\begin{table}[htbp]
\centering
\caption{Coverage of the direct and transformed CI for the MTTF under contamination in $a_1$ and $\eta$. Each cell shows $\left(a_1,\ \eta\right)^{\intercal}$.}
\label{tab:table_CI_MTTF}
\renewcommand{\arraystretch}{1.15}
\small

\begin{tabular}{c|cccccc}
\hline
\multicolumn{7}{c}{\textbf{Direct CI}} \\
\hline
proportion $\setminus \beta$ & MLE & 0.2 & 0.4 & 0.6 & 0.8 & 1 \\
\hline

0 &
$[96.4,96.4]$ &
$[96.2,96.2]$ &
$[95.4,95.4]$ &
$[95.1,95.1]$ &
$[94.5,94.5]$ &
$[94.3,94.3]$
\\

0.03 &
$[74.4,70.9]$ &
$[79.3,77.1]$ &
$[85.2,83.8]$ &
$[87.7,87.7]$ &
$[88.0,88.2]$ &
$[89.4,88.9]$
\\

0.05 &
$[54.3,44.7]$ &
$[62.7,57.5]$ &
$[73.8,73.2]$ &
$[80.2,81.0]$ &
$[83.8,85.0]$ &
$[85.4,87.3]$
\\

0.07 &
$[37.5,22.3]$ &
$[46.5,39.1]$ &
$[59.7,58.6]$ &
$[72.1,73.8]$ &
$[78.8,80.7]$ &
$[83.2,85.2]$
\\

0.08 &
$[31.0,13.8]$ &
$[39.2,29.3]$ &
$[54.7,52.1]$ &
$[67.4,70.2]$ &
$[75.7,78.8]$ &
$[80.7,83.6]$
\\

0.09 &
$[25.4,7.9]$ &
$[31.9,22.2]$ &
$[46.4,43.7]$ &
$[60.4,64.0]$ &
$[73.1,77.1]$ &
$[79.1,82.0]$
\\

0.1 &
$[20.4,5.5]$ &
$[27.8,13.9]$ &
$[40.6,35.9]$ &
$[54.9,57.8]$ &
$[67.4,72.6]$ &
$[77.1,80.5]$
\\
\end{tabular}


\vspace{0.5cm}

\begin{tabular}{c|cccccc}
\hline
\multicolumn{7}{c}{\textbf{Transformed CI}} \\
\hline
proportion $\setminus \beta$ & MLE & 0.2 & 0.4 & 0.6 & 0.8 & 1 \\
\hline

0 &
$[95.9,95.9]$ &
$[95.7,95.7]$ &
$[95.3,95.3]$ &
$[95.4,95.4]$ &
$[94.4,94.4]$ &
$[94.1,94.1]$
\\

0.03 &
$[78.5,75.7]$ &
$[83.0,81.7]$ &
$[87.8,87.5]$ &
$[89.9,89.4]$ &
$[90.7,89.9]$ &
$[91.8,91.4]$
\\

0.05 &
$[60.5,50.0]$ &
$[68.5,63.9]$ &
$[77.7,77.2]$ &
$[83.3,84.5]$ &
$[85.5,87.5]$ &
$[87.0,88.3]$
\\

0.07 &
$[43.5,27.6]$ &
$[52.5,44.6]$ &
$[64.3,63.5]$ &
$[76.3,77.8]$ &
$[81.9,83.2]$ &
$[85.3,86.8]$
\\

0.08 &
$[36.7,17.5]$ &
$[45.7,34.2]$ &
$[58.4,56.5]$ &
$[71.8,73.6]$ &
$[78.3,81.9]$ &
$[82.8,86.2]$
\\

0.09 &
$[30.7,10.6]$ &
$[37.9,26.3]$ &
$[53.2,49.4]$ &
$[66.0,69.5]$ &
$[76.5,79.6]$ &
$[81.6,84.0]$
\\

0.1 &
$[26.0,7.1]$ &
$[32.2,18.7]$ &
$[47.0,41.3]$ &
$[60.4,63.0]$ &
$[72.7,77.0]$ &
$[79.8,82.0]$
\\

\end{tabular}

\end{table}

\begin{table}[htbp]
\centering
\caption{Coverage of the direct and transformed CI for the survival rate under contamination in $a_1$ and $\eta$. Each cell shows $\left(a_1,\ \eta \right)^{\intercal}$.}
\label{tab:table_CI_survival}
\renewcommand{\arraystretch}{1.18}
\small

\begin{tabular}{c|cccccc}
\hline
\multicolumn{7}{c}{\textbf{Direct CI}} \\
\hline
proportion $\setminus \beta$ & MLE & 0.2 & 0.4 & 0.6 & 0.8 & 1 \\ 
\hline

0 &
$[94.6,94.6]$ &
$[94.7,94.7]$ &
$[94.0,94.0]$ &
$[93.6,93.6]$ &
$[93.6,93.6]$ &
$[93.5,93.5]$
\\

0.03 &
$[72.6,71.2]$ &
$[72.7,72.0]$ &
$[76.2,76.8]$ &
$[79.7,80.3]$ &
$[81.6,81.9]$ &
$[83.2,82.9]$
\\

0.05 &
$[41.8,35.9]$ &
$[44.4,42.1]$ &
$[52.4,53.0]$ &
$[58.9,61.3]$ &
$[64.9,67.9]$ &
$[68.4,71.3]$
\\

0.07 &
$[13.7,7.4]$ &
$[17.2,13.6]$ &
$[26.2,28.3]$ &
$[36.5,41.3]$ &
$[45.7,51.2]$ &
$[51.8,57.0]$
\\

0.08 &
$[6.6,1.6]$ &
$[8.6,6.0]$ &
$[14.8,16.3]$ &
$[25.7,31.7]$ &
$[34.6,42.2]$ &
$[42.6,48.7]$
\\

0.09 &
$[1.7,0.3]$ &
$[2.3,1.0]$ &
$[8.4,9.2]$ &
$[16.4,22.3]$ &
$[26.5,34.4]$ &
$[33.7,41.0]$
\\

0.1 &
$[0.3,0.0]$ &
$[0.4,0.1]$ &
$[2.2,2.9]$ &
$[9.0,13.9]$ &
$[18.3,26.4]$ &
$[26.3,33.5]$
\\

\end{tabular}

\vspace{0.6cm}

\begin{tabular}{c|cccccc}
\hline
\multicolumn{7}{c}{\textbf{Transformed CI}} \\
\hline
proportion $\setminus \beta$ & MLE & 0.2 & 0.4 & 0.6 & 0.8 & 1 \\ 
\hline

0 &
$[94.4,94.4]$ &
$[94.6,94.6]$ &
$[94.8,94.8]$ &
$[94.9,94.9]$ &
$[95.3,95.3]$ &
$[95.2,95.2]$
\\

0.03 &
$[87.8,87.2]$ &
$[88.7,87.7]$ &
$[90.0,90.3]$ &
$[91.7,92.0]$ &
$[92.2,92.5]$ &
$[92.8,92.9]$
\\

0.05 &
$[63.8,57.7]$ &
$[65.0,63.5]$ &
$[72.2,73.0]$ &
$[78.7,80.4]$ &
$[83.3,85.9]$ &
$[85.2,87.5]$
\\

0.07 &
$[32.9,22.0]$ &
$[34.8,31.2]$ &
$[45.9,46.8]$ &
$[56.6,61.5]$ &
$[66.1,70.1]$ &
$[72.2,76.0]$
\\

0.08 &
$[18.7,8.7]$ &
$[21.6,16.3]$ &
$[31.2,32.6]$ &
$[42.1,49.7]$ &
$[54.5,62.8]$ &
$[62.4,68.6]$
\\

0.09 &
$[9.3,2.2]$ &
$[11.5,7.4]$ &
$[18.3,19.9]$ &
$[30.0,38.5]$ &
$[43.0,52.4]$ &
$[51.7,60.9]$
\\

0.1 &
$[3.5,0.4]$ &
$[4.2,1.2]$ &
$[9.9,10.9]$ &
$[19.5,26.7]$ &
$[31.7,41.5]$ &
$[41.6,51.9]$
\\

\end{tabular}

\end{table}

Finally, to better understand the estimators’ behavior under contamination of \(a_1\), we plotted \(95\%\) confidence ellipses for pairs of parameters.
As shown in Figures Figures~\ref{fig:ellip_a0_a1}--\ref{fig:ellip_a1_eta}, the effect of the outliers does not induce a notable inflation in the estimator's variance, but rather a significant deviation of the estimates from the true parameter values.  This bias is effectively mitigated by the MDPDE.

\begin{figure}[htb]
  \centering
  \includegraphics[width=0.75\textwidth]{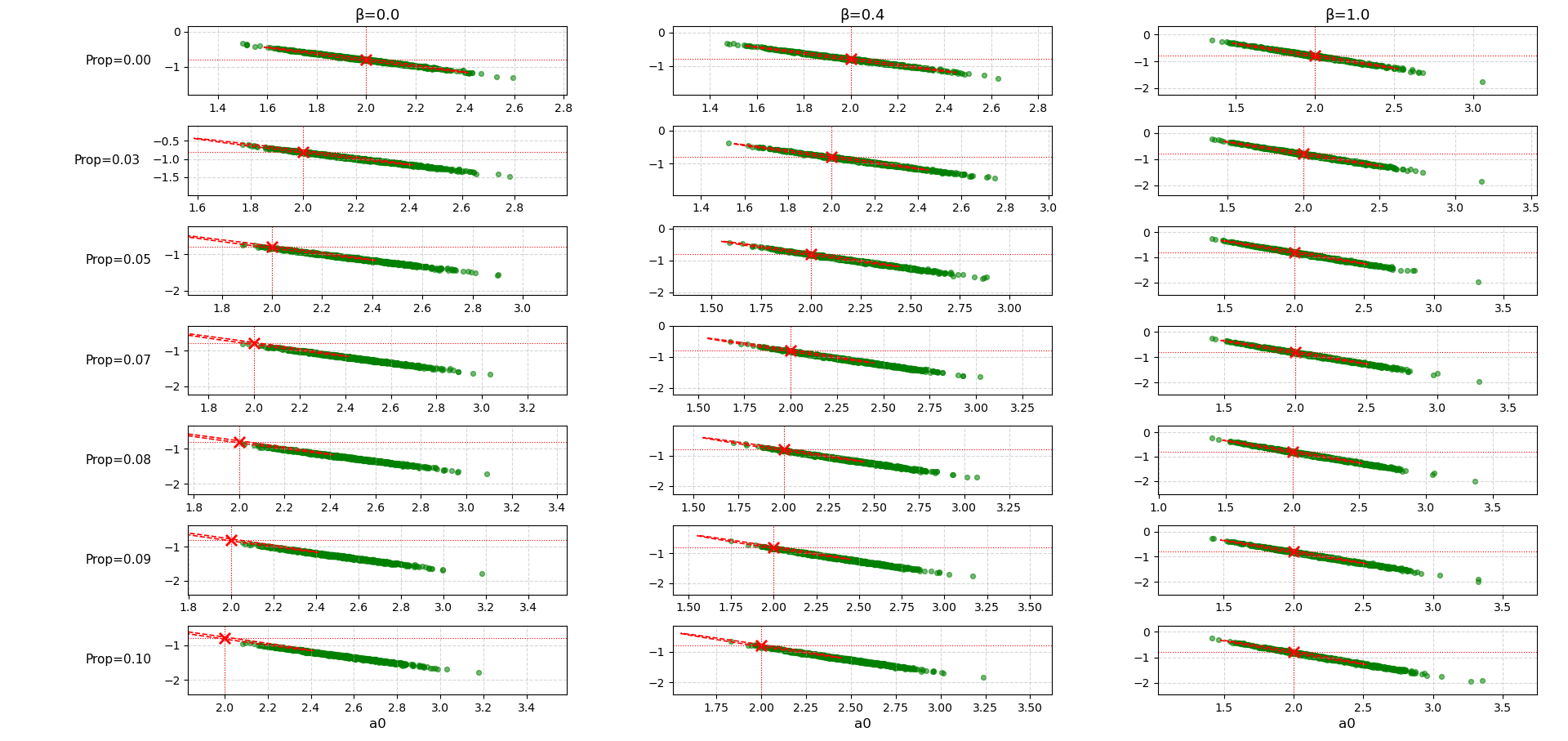}
  \caption{Ellipse of the estimations for $a_0$ and $a_1$ for the 95\% under contamination in  the parameter $a_1$.}
  \label{fig:ellip_a0_a1}
\end{figure}
\begin{figure}[htb]
  \centering
  \includegraphics[width=0.75\textwidth]{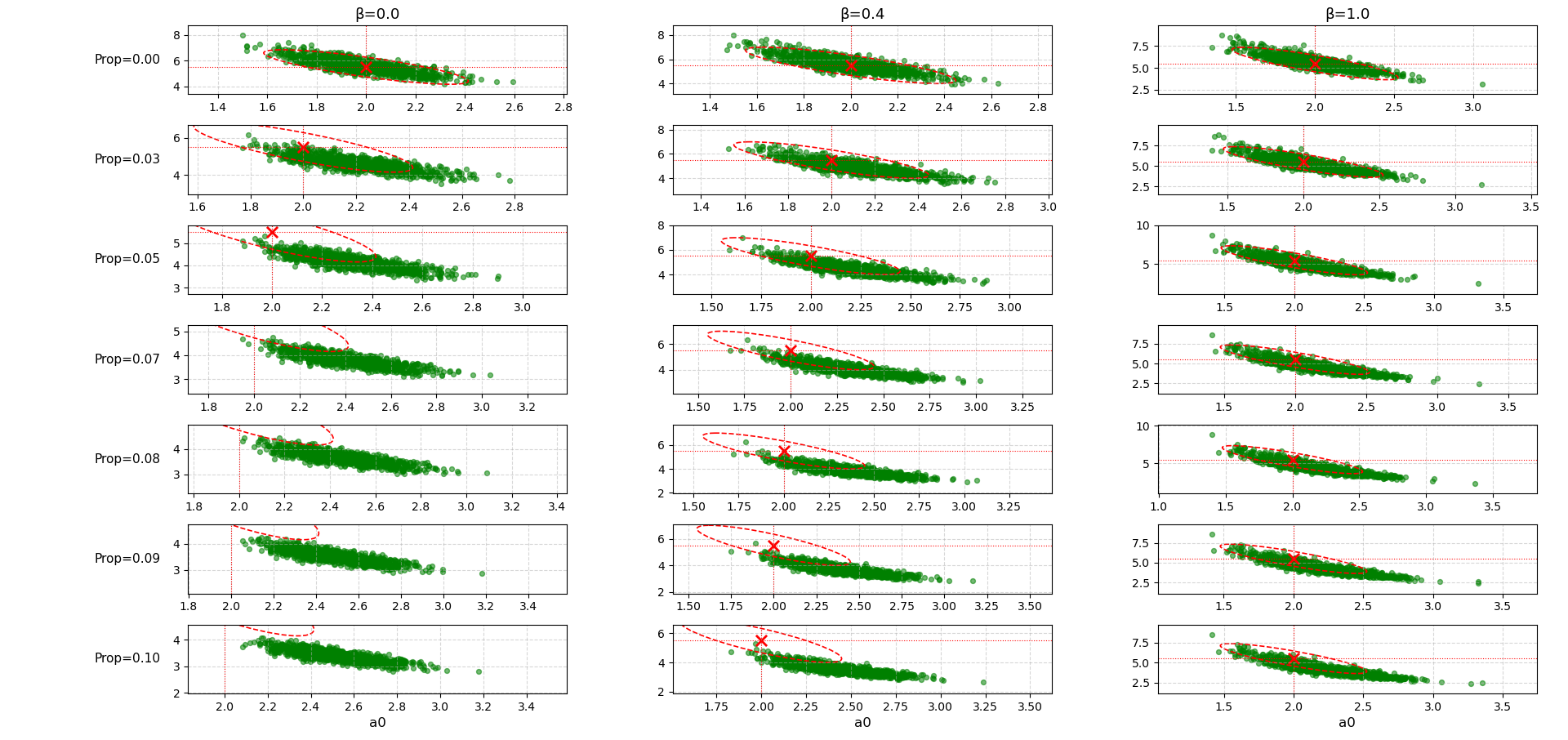}
  \caption{Ellipse of the estimations for $a_0$ and $\eta$ for the 95\% under contamination in  the parameter $a_1$.}
  \label{fig:ellip_a0_eta}
\end{figure}
\begin{figure}[htb]
  \centering
  \includegraphics[width=0.75\textwidth]{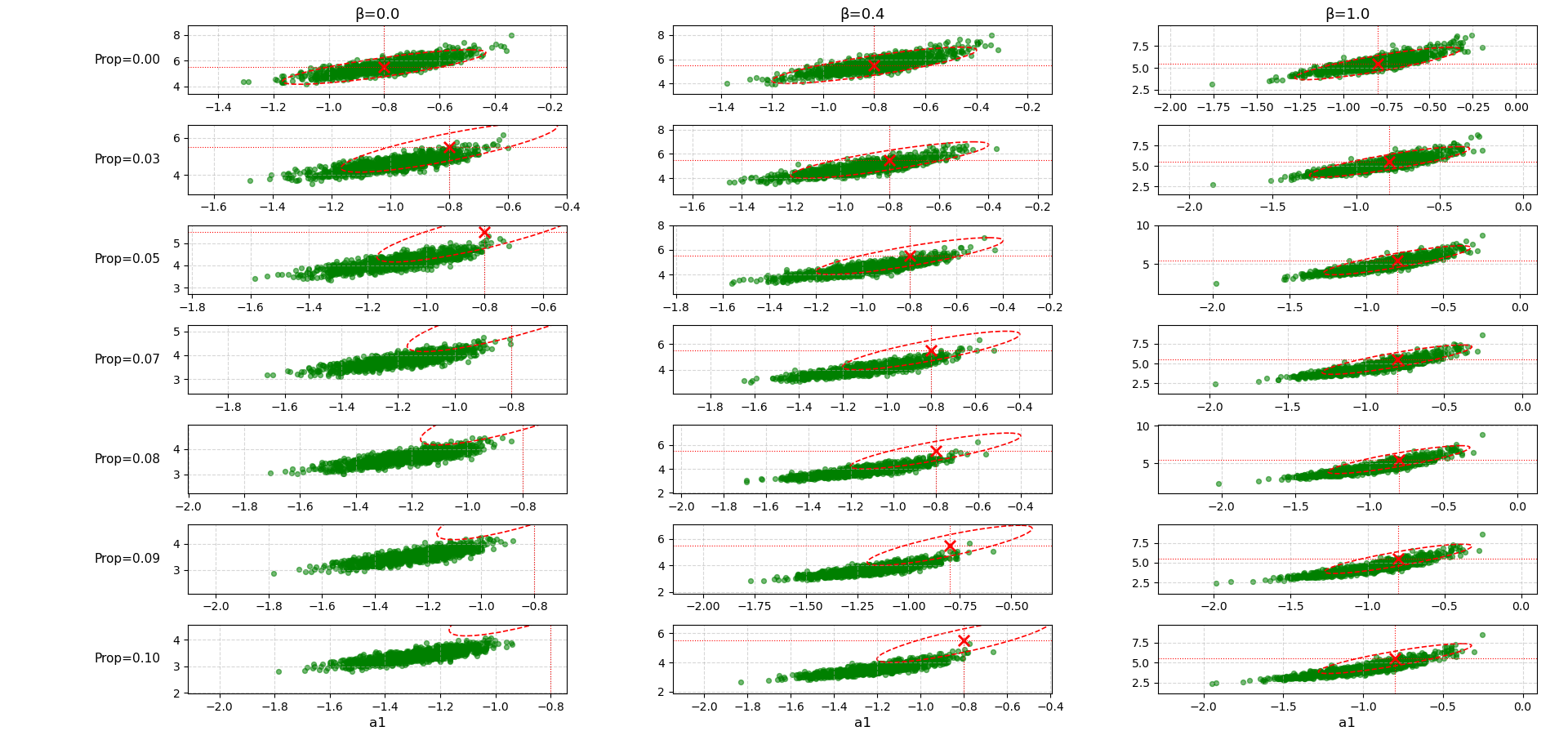}
  \caption{Ellipse of the estimations for $a_1$ and $\eta$ for the 95\% under contamination in  the parameter $a_1$.}
  \label{fig:ellip_a1_eta}
\end{figure}
\section{Real-Data Example \label{sec:example}}

This example was studied by \cite{han2015} to obtain the reliability indices of a kind of solar lightning devices at the normal we will try to do inference for a temperature of $x_0=288^\circ K$. A total of $N = 35$ apparatus were randomly selected for a SSALT with two stress levels $x_1 = 293^\circ K$ and $x_2 = 353^\circ K$. The stress change occur after 500 hours and the experiment ended when all the products failed. In order to make a type I censoring we will stablish an end of the experiment after 530 hours. So all the devices that in the original work failed after that will be considered as devices that survive the test. The original failures times are (in hundred of hours)

\begin{itemize}
  \item Failure times at the first stress level $x_1$:  0.140 0.783 1.324 1.582 1.716 1.794 1.883 2.293 2.660 2.674 2.725 3.085 3.924 4.396 4.612 4.892.
  \item Failure times at the second stress level $x_2$: 5.002 5.022 5.082 5.112 5.147 5.238 5.244 5.247 5.305 5.337 5.407 5.408 5.445 5.483 5.717.
\end{itemize}

The Table~\ref{tab:estimates_and_cis} shows the parameter estimates under different values of $\beta$, together with their approximate CIs, derived from the asymptotic distributions in theorems \ref{thm:asymp}:
\[
\hat{\theta}_i \pm \frac{\widehat{\sigma}_\beta(\widehat{\theta}^\beta_i)}{\sqrt{n}},
\]
with $\widehat{\sigma}_\beta(\widehat{\bm{\theta}}^\beta_i) = \left(J_{\beta}(\widehat{\bm{\theta}}^\beta)^{-1}K_\beta(\widehat{\bm{\theta}}^\beta)J_\beta(\widehat{\bm{\theta}}^\beta)^{-1}\right)_{ii}$ where the corresponding quantities
$J_\beta(\bm{\theta}_i)$ and $K_\beta(\bm{\theta}_i)$ are defined for $a_0$ in Proposition \ref{thm:asympa0}, for $a_1$ in Proposition \ref{thm:asympa1} and for $\eta$ in Proposition \ref{thm:asympeta}.
Larger values of the tuning parameter leads to wider CIs, but results are quite stable across the different $\beta$s.
\begin{table}[htb]
\centering
\caption{Estimations and CI of the parameters.}
\label{tab:estimates_and_cis}
\begin{tabular}{lcccccc}
\toprule
$\beta$ & $\hat{a}_0$ & IC($\hat{a}_0$) & $\hat{a}_1$ ($\times 10^2$) & IC($\hat{a}_1$) ($\times 10^2$) & $\hat{\eta}$ & IC($\hat{\eta}$) \\
\midrule
0.000 & 14.384 & $[9.677, 19.091]$ & -4.400 & $[-6.000, -2.800]$ & 2.793 & $[1.690, 3.896]$ \\
0.200 & 15.746 & $[10.507, 20.985]$ & -4.800 & $[-6.600, -3.100]$ & 1.463 & $[0.824, 2.102]$ \\
0.400 & 15.423 & $[10.098, 20.747]$ & -4.700 & $[-6.500, -3.000]$ & 1.471 & $[0.818, 2.123]$ \\
0.600 & 15.220 & $[9.857, 20.583]$ & -4.600 & $[-6.400, -2.900]$ & 1.472 & $[0.822, 2.121]$ \\
0.800 & 15.090 & $[9.396, 20.784]$ & -4.600 & $[-6.500, -2.700]$ & 1.475 & $[0.790, 2.160]$ \\
1.000 & 14.178 & $[7.783, 20.573]$ & -4.300 & $[-6.400, -2.100]$ & 1.643 & $[0.740, 2.546]$ \\
\bottomrule
\end{tabular}
\end{table}

As outlined in Section~\ref{sec:CI}, robust estimators of lifetime-related characteristics like MTTF, reliability at survival time $t$, and distribution quantiles, can be obtained from the model estimates.

Robust MTTF estimates based on MDPDEs, along with their corresponding direct and transformed approximate confidence intervals (CIs), are presented in Table~\ref{tab:estimates_MTTF_and_cis} for two stress levels:
\begin{itemize}
    \item $x_0=288^ \circ K$ (normal operating conditions)
    \item $x_1=293^ \circ K$ (increased stress).
\end{itemize}
A clear and expected trend is apparent: as the stress level increases ($x_0 < x_1 $), the estimated mean lifetime decreases significantly. Specifically, this reduction is by a factor of $\exp\left(-\hat{a}_1\right)$. For a constant stress level, the MTTF estimates are relatively stable across different $\beta$ values, which suggests a possible absence of outlying observations.

Importantly, under normal operating conditions, the lower bounds of the approximate direct confidence intervals exceed zero, and thus need to be truncated. In contrast, transformed CIs for the MTTF provide positive interval bounds, which are more physically meaningful.

When the stress is increased from $x_0=283^\circ K$ to $x_1=293^\circ K$, the MTTF is reduced to 80\% of its original value. More generally, a $5^\circ K$ increase in environmental temperature may lead to a reduction in the MTTF of the devices by $\left(1-\exp\left(5\widehat{a}_1^\beta \right)\right) \cdot 100\%$, which is approximately 19.75\%.

\begin{table}[htb]
\centering
\caption{Estimated mean lifetime and asymptotic (direct and transformed) confidence intervals (in hundrer of hours) of the lightning apparatus under two  constant temperatures}
\label{tab:estimates_MTTF_and_cis}
\begin{tabular}{lcccc}
\toprule
& Mean lifetime & Direct CI & Transformed CI \\
\midrule
\multicolumn{4}{c}{$x_0 = 288$} \\
\midrule
MLE & 4.818 & $[3.842, 5.794]$ & $[3.935, 5.900]$ \\
0.2 & 5.489 & $[3.224, 7.755]$ & $[3.633, 8.294]$ \\
0.4 & 5.549 & $[3.246, 7.851]$ & $[3.664, 8.402]$ \\
0.6 & 5.642 & $[3.267, 8.017]$ & $[3.703, 8.595]$ \\
0.8 & 5.744 & $[3.236, 8.252]$ & $[3.711, 8.889]$ \\
1 & 5.643 & $[3.299, 7.987]$ & $[3.725, 8.549]$ \\
\midrule
\multicolumn{4}{c}{$x_1 = 293$} \\
\midrule
MLE & 3.865 & $[3.234, 4.496]$ & $[3.288, 4.550]$ \\
0.2 & 4.309 & $[2.777, 5.841]$ & $[3.020, 6.149]$ \\
0.4 & 4.381 & $[2.803, 5.959]$ & $[3.056, 6.280]$ \\
0.6 & 4.471 & $[2.807, 6.136]$ & $[3.082, 6.487]$ \\
0.8 & 4.564 & $[2.767, 6.361]$ & $[3.079, 6.766]$ \\
1 & 4.555 & $[2.828, 6.282]$ & $[3.118, 6.655]$ \\
\bottomrule
\end{tabular}
\end{table}

In an industrial context, the reliability of a device at a specific established time under normal operating conditions is a crucial metric, as is understanding how this reliability is impacted by increased environmental temperatures. For a fixed established time of $t=500h$, the MDPDEs for device reliability and their corresponding direct and transformed approximate confidence intervals (CIs) for various $\beta$ values are presented in Table~\ref{tab:estimates_survival_and_cis} for different constant temperatures.

A significant observation is the substantial decrease in reliability as the stress level (temperature) increases. This finding indicates a faster degradation of the product at higher temperatures. Additionally, the selection of the tuning parameter $\beta$ has a relatively minor effect on the estimated reliability, as all estimates remain quite similar.

Similar to the MTTF estimates, the bounds of the direct CIs for reliability under normal operating conditions are truncated to satisfy the natural constraint that reliability must fall between 0 and 1. While the direct and transformed approximate confidence intervals provide similar bounds at lower temperatures, they show a slight difference at higher temperatures.
\begin{table}[h]
\centering
\caption{Estimated reliability at $t=500$ and asymptotic (direct and transformed) confidence intervals (in hours) for the lightning apparatus under two  constant temperatures.}
\label{tab:estimates_survival_and_cis}
\begin{tabular}{lcccc}
\toprule
& $\hat{R}$(500) & Direct CI & Transformed CI \\
\midrule
\multicolumn{4}{c}{$x_0 = 288$} \\
\midrule
MLE & 0.448 & $[0.247, 0.650]$ & $[0.286, 0.703]$ \\
0.2 & 0.470 & $[0.275, 0.665]$ & $[0.310, 0.712]$ \\
0.4 & 0.477 & $[0.280, 0.673]$ & $[0.316, 0.720]$ \\
0.6 & 0.485 & $[0.285, 0.686]$ & $[0.321, 0.734]$ \\
0.8 & 0.495 & $[0.286, 0.704]$ & $[0.324, 0.756]$ \\
1 & 0.505 & $[0.281, 0.730]$ & $[0.324, 0.788]$ \\
\midrule
\multicolumn{4}{c}{$x_1 = 293$} \\
\midrule
MLE & 0.227 & $[0.065, 0.388]$ & $[0.111, 0.462]$ \\
0.2 & 0.341 & $[0.156, 0.527]$ & $[0.198, 0.587]$ \\
0.4 & 0.350 & $[0.162, 0.539]$ & $[0.205, 0.600]$ \\
0.6 & 0.361 & $[0.166, 0.557]$ & $[0.211, 0.620]$ \\
0.8 & 0.373 & $[0.165, 0.581]$ & $[0.213, 0.651]$ \\
1 & 0.379 & $[0.152, 0.606]$ & $[0.208, 0.689]$ \\
\bottomrule
\end{tabular}
\end{table}

Finally, let us consider the results for the 0.5-quantile estimates. The 50\%-quantile represents the time by which 50\% of the lightning apparatus are expected to have survived under constant temperature conditions. Table \ref{tab:estimates_median_and_cis} presents the MDPDEs for the 0.5-quantile and their corresponding approximate confidence intervals (CIs).

As expected, all 0.5-quantile estimates show a strong inverse relationship with temperature; the time at which 50\% of the devices are still functioning is much shorter at higher temperatures. When comparing direct and transformed approximate CIs for the 0.5-quantile, the transformed CIs tend to have higher upper and lower bounds.

\begin{table}[htb]

	\centering
	\caption{MDPDEs for the $50\%-$quantile (in minutes) and their corresponding asymptotic (direct and transformed) confidence intervals for the lightning apparatus under two constant temperatures.}
	\label{tab:estimates_median_and_cis}
\begin{tabular}{lcccc}	

\toprule
& $\hat{Q}_{0.5}$ & Direct CI & Transformed CI \\
\midrule
\multicolumn{4}{c}{$x_0 = 288$} \\
\midrule
MLE & 4.746 & $[3.798, 5.694]$ & $[3.886, 5.795]$ \\
0.2 & 4.718 & $[2.948, 6.487]$ & $[3.242, 6.865]$ \\
0.4 & 4.778 & $[2.973, 6.584]$ & $[3.275, 6.972]$ \\
0.6 & 4.860 & $[2.973, 6.748]$ & $[3.296, 7.167]$ \\
0.8 & 4.952 & $[2.932, 6.972]$ & $[3.294, 7.446]$ \\
1 & 5.047 & $[3.051, 7.043]$ & $[3.398, 7.495]$ \\
\midrule
\multicolumn{4}{c}{$x_1 = 293$} \\
\midrule
MLE & 3.807 & $[3.158, 4.457]$ & $[3.210, 4.515]$ \\
0.2 & 3.703 & $[2.445, 4.962]$ & $[2.636, 5.203]$ \\
0.4 & 3.773 & $[2.470, 5.076]$ & $[2.671, 5.329]$ \\
0.6 & 3.852 & $[2.461, 5.243]$ & $[2.685, 5.527]$ \\
0.8 & 3.935 & $[2.401, 5.469]$ & $[2.665, 5.811]$ \\
1 & 4.074 & $[2.474, 5.674]$ & $[2.751, 6.033]$ \\

\bottomrule
\end{tabular}
\end{table}

\section{Conclusions}
In this paper, we have continued to explore the performance of the DPD-based method for robust parameter estimation under continuous monitoring of a SSALT model. In this study, we have extended the methodology to a more complex lifetime distribution, specifically the Weibull distribution, which allows for a more flexible modelling of the device lifetimes. Explicit expressions for the MDPDE, as well as the asymptotic variance-covariance of the parameter estimates and other related quantities, have been derived.

Our numerical experiments highlight that the DPD approach provides improved estimation performance compared to classical maximum likelihood methods when extreme or influential observations are present. In particular, the impact of contamination on the variance and bias of the parameter estimates is substantially mitigated, confirming the robustness of the methodology. The observed performance suggests that deriving the Influence Function (IF) for the resulting MDPDE in this generalized setting would be a natural and valuable extension, as it would allow for a formal theoretical assessment of robustness under various contamination schemes.  

Furthermore, having access to the asymptotic variance of the MDPDE opens the possibility of constructing robust inferential procedures, such as Rao-type and Wald-type test statistics. Evaluating the performance of these tests in the presence of outliers would be of practical interest, particularly in reliability and life-testing applications where extreme values can heavily distort classical inference. Additional research could also explore extensions to other parametric families, such as log-normal, gamma, or log-logistic distributions, under the mixed distribution formulation, thereby broadening the applicability of the DPD framework in life-stress modeling.  

Finally, the consideration of more realistic data scenarios, such as type-II censoring or interval-censored observations, remains an open research direction. Adapting the DPD-based estimators to accommodate these censoring schemes would be essential for practical implementation in real-life industrial experiments. In summary, this work reinforces the practical relevance of robust estimation using DPD, illustrates its advantages in complex distributions, and outlines multiple avenues for future theoretical and applied research in robust life-testing methodology.

\clearpage
\bibliographystyle{plainnat}  
\bibliography{sn-bibliography}

\clearpage
\section{Appendix}
\subsection{Proof Theorem 3}
Following from \citet{basu1997}, and supposing that the Weibull distribution is the lifetime distribution, the expression of $\bm{J}_{\beta}$, $\bm{K}_{\beta}$ and $\bm{\xi}_{\beta}$ are:
\begin{align*}
\bm{J}_{\beta}\left(\bm{\theta}\right)&=\int \bm{u}_{\bm{\theta}}(t)\bm{u}_{\bm{\theta}}(t)^{\intercal}f(t|\bm{\theta})^{\beta+1}dt
\\
\bm{K}_{\beta}\left(\bm{\theta}\right)&=\int \bm{u}_{\bm{\theta}}(t)\bm{u}_{\bm{\theta}}(t)^{\intercal}f(t|\bm{\theta})^{2\beta+1}dt-\bm{\xi}_{\beta}\left(\bm{\theta}\right)\bm{\xi}_{\beta}\left(\bm{\theta}\right)^{\intercal}
\end{align*}
being
\begin{equation*}
\bm{u}_{\bm{\theta}}(t)=\frac{\partial \log\left(f_{T}\left(t|\bm{\theta}\right)\right)}{\partial \bm{\theta}}.
\end{equation*}
The elements of the diagonal for the matrix $\mathbf{J}_{\beta}$ are:
\begin{align*}
J_{ii}^{\beta}(\theta_i)&=\int_0^{\tau_1}\left(\frac{\partial\log(f_{x_1}(t|\lambda_1,\eta))}{\partial \theta_i}\right)^2 f_{x_1}(t|\lambda_1,\eta)^{\beta+1}dt+\int_{\tau_1}^{\tau_2}\left(\frac{\partial\log(f_{x_2}(t+h|\lambda_1,\lambda_2,\eta))}{\partial \theta_i}\right)^2 f_{x_2}(t+h|\lambda_1,\lambda_2,\eta)^{\beta+1}dt
\\
&+\left(\frac{\partial\log(1-F_{x_2}(\tau_2+h|\lambda_1,\lambda_2,\eta))}{\partial \theta_i}\right)^2 (1-F_{x_2}(\tau_2+h|\lambda_1,\lambda_2,\eta))^{\beta+1}, \quad i=1,2,3;
\end{align*}
the elements outside the diagonal for the matrix $\mathbf{J}_{\beta}$ are:
\begin{align*}
J_{ij}^{\beta}(\theta_i,\theta_j)&=\int_0^{\tau_1}\left(\frac{\partial\log(f_{x_1}(t|\lambda_1,\eta))}{\partial \theta_i}\right) \left(\frac{\partial\log(f_{x_1}(t|\lambda_1,\eta))}{\partial \theta_j}\right) f_{x_1}(t|\lambda_1,\eta)^{\beta+1}dt
\\
&+\int_{\tau_1}^{\tau_2}\left(\frac{\partial\log(f_{x_2}(t+h|\lambda_1,\lambda_2,\eta))}{\partial \theta_i}\right) \left(\frac{\partial\log(f_{x_2}(t+h|\lambda_1,\lambda_2,\eta))}{\partial \theta_j}\right)f_{x_2}(t+h|\lambda_1,\lambda_2,\eta)^{\beta+1}dt
\\
&+\left(\frac{\partial\log(1-F_{x_2}(\tau_2+h|\lambda_1,\lambda_2,\eta))}{\partial \theta_i}\right) \left(\frac{\partial\log(1-F_{x_2}(\tau_2+h|\lambda_1,\lambda_2,\eta))}{\partial \theta_j}\right)(1-F_{x_2}(\tau_2+h|\lambda_1,\lambda_2,\eta))^{\beta+1},
\\
& \quad i,j=1,2,3, \quad i \neq j;
\end{align*}
and the elements of $\mathbf{\xi}_{\beta}$ are:
\begin{align*}
\xi_{i}^{\beta}(\theta_i)&=\int_0^{\tau_1}\left(\frac{\partial\log(f_{x_1}(t|\lambda_1,\eta))}{\partial \theta_i}\right) f_{x_1}(t|\lambda_1,\eta)^{\beta+1}dt+\int_{\tau_1}^{\tau_2}\left(\frac{\partial\log(f_{x_2}(t+h|\lambda_1,\lambda_2,\eta))}{\partial \theta_i}\right) f_{x_2}(t+h|\lambda_1,\lambda_2,\eta)^{\beta+1}dt
\\
&+\left(\frac{\partial\log(1-F_{x_2}(\tau_2+h|\lambda_1,\lambda_2,\eta))}{\partial \theta_i}\right) (1-F_{x_2}(\tau_2+h|\lambda_1,\lambda_2,\eta))^{\beta+1}, \quad i=1,2,3.
\end{align*}
\subsection{Proof Lemma 1}

\begin{equation*}
\zeta_{\alpha,\beta}^{\tau_1}\left(a_0,a_1,\eta\right)=\int_{0}^{\tau_1} \left(\frac{t}{\lambda_1}\right)^{\alpha} \frac{\eta^{\beta+1}}{\lambda_1^{\frac{\eta}{\beta+1}}}t^{(\eta-1)(\beta+1)}\exp\left(-\left(\frac{t}{\lambda_1}\right)^{\eta}(\beta+1)\right). 
\end{equation*}
Doing the following change of variable:
\begin{align*}
\left(\frac{t}{\lambda_1}\right)^{\eta}(\beta+1)=l & \Rightarrow t^{\eta}=\frac{l}{\beta+1}\lambda_1^{\eta} \Rightarrow t=\frac{l^{\frac{1}{\eta}}}{(\beta+1)^{\frac{1}{\eta}}}\lambda_1  \\
&\Rightarrow dt= \frac{\lambda_1}{(\beta+1)^{\frac{1}{\eta}}}\frac{1}{\eta}  l^{\frac{1}{\eta}-1}dl.
\end{align*}
Then we have:
\begin{align*}
\zeta_{\alpha,\beta}^{\tau_1}\left(a_0,a_1,\eta\right)=\int_0^{\left(\frac{\tau_1}{\lambda_1}\right)^{\eta}(\beta+1)} &\frac{1}{\lambda^{\alpha}}\frac{l^{\frac{\alpha}{\eta}}}{(\beta+1)^{\frac{\alpha}{\eta}}}\lambda_1^{\alpha} \frac{\eta^{\beta+1}}{\lambda_1^{\eta(\beta+1)}}\frac{l^{(\eta-1)(\beta+1)}}{(\beta+1)^{\frac{(\eta-1)(\beta+1)}{\eta}}}
\lambda_1^{(\eta-1)(\beta+1)}\\
& \cdot \exp\left(-l\right) \frac{\lambda_1}{(\beta+1)^{\frac{1}{\eta}}} \frac{1}{\eta} l^{\frac{1}{\eta}-1}dl.
\end{align*}
Then, we have the following exponents:
\begin{align*}
l &\rightarrow \frac{\alpha}{\eta}+ \frac{(\beta+1)(\eta-1)}{\eta}+\frac{1}{\eta}-1=\frac{\alpha+(\beta+1)(\eta-1)-1}{\eta}-1
\\
\lambda_1 &\rightarrow -\alpha+\alpha -\eta(\beta+1) + (\eta-1)(\beta+1)+1=-\beta
\\
\eta & \rightarrow \beta+1-1=\beta
\\
(\beta+1) & \rightarrow -\frac{\alpha}{\eta} -\frac{(\eta-1)(\beta+1)}{\eta} - \frac{1}{\eta} = -\frac{\alpha + (\eta-1)(\beta+1)+1}{\eta}.
\end{align*}
Then we have:
\begin{align*}
\zeta_{\alpha,\beta}^{\tau_1}\left(a_0,a_1,\eta\right)&=\left(\frac{\eta}{\lambda_1}\right)^\beta \frac{1}{(\beta+1)^{\frac{\alpha+(\eta-1)(\beta+1)+1}{\eta}}}\int_0^{\left(\frac{\tau_1}{\lambda_1}\right)^{\eta}(\beta+1)} \exp\left(-l\right) l^{\frac{\alpha+(\eta-1)(\beta+1)+1}{\eta}-1} dl
\\
&=\left(\frac{\eta}{\lambda_1}\right)^\beta \frac{1}{(\beta+1)^{\frac{\alpha+(\eta-1)(\beta+1)+1}{\eta}}}\gamma\left(\frac{\alpha+(\eta-1)(\beta+1)+1}{\eta},\left(\frac{\tau_1}{\lambda_1}\right)^{\eta}(\beta+1)\right).
\end{align*}
We also have to prove the following result:
\begin{align*}
\zeta_{\alpha,\beta}^{\tau_1,\tau_2}\left(a_0,a_1,\eta\right)&=\int_{\tau_1}^{\tau_2}\left(\frac{t+\frac{\lambda_2}{\lambda_1}\tau_1-\tau_1}{\lambda_2} \right)^{\alpha}g_{x_2}(t)^{\beta+1}dt
\\
&=\int_{\tau_1}^{\tau_2}\left(\frac{t+\frac{\lambda_2}{\lambda_1}\tau_1-\tau_1}{\lambda_2} \right)^{\alpha} \frac{\eta^{\beta+1}}{\lambda_2^{\eta(\beta+1)}} \left(t+\frac{\lambda_2}{\lambda_1}\tau_1-\tau_1\right)^{(\eta-1)(\beta+1)}
\\
& \cdot \exp\left(-\frac{1}{\lambda_2^{\eta}}\left(t+\frac{\lambda_2}{\lambda_1}\tau_1-\tau_1\right)^{\eta}(\beta+1)\right)dt.
\end{align*}
Doing the following change of variable:
\begin{align*}
\left( \frac{t+\frac{\lambda_2}{\lambda_1}\tau_1-\tau_1}{\lambda_2}\right)^{\eta} (\beta+1)&= l \Rightarrow \left(t+\frac{\lambda_2}{\lambda_1}\tau_1-\tau_1\right)^{\eta}=\frac{l \lambda_2 ^{\eta}}{\beta+1}
\\
\Rightarrow  t+\frac{\lambda_2}{\lambda_1}\tau_1-\tau_1&=\frac{l^{\frac{1}{\eta}} \lambda_2}{(\beta+1)^{\frac{1}{\eta}}} \Rightarrow
 t=\frac{l^{\frac{1}{\eta}} \lambda_2}{(\beta+1)^{\frac{1}{\eta}}}+\frac{\lambda_2}{\lambda_1}\tau_1+\tau_1
 \\
 dt&=\frac{1}{\eta}l^{\frac{1}{\eta}-1}\lambda_2\frac{1}{(\beta+1)^{\frac{1}{\eta}}}dl.
\end{align*}
And the integral can be expressed as
\begin{align*}
\zeta_{\alpha,\beta}^{\tau_1,\tau_2}\left(a_0,a_1,\eta\right)&=\int_{\left(\frac{\tau_1}{\lambda_1}\right)^{\eta}(\beta+1)}^{\left(\tau_2+\frac{\lambda_2}{\lambda_1}\tau_1-\tau_1\right)^{\eta}\frac{1}{\lambda_2^{\eta}}(\beta+1)}
\frac{1}{\lambda_2^{\alpha}} \frac{l^{\frac{\alpha}{\eta}}\lambda_2^{\alpha}}{(\beta+1)^{\frac{\alpha}{\eta}}}\frac{\eta^{\beta+1}}{\lambda_2^{\eta(\beta+1)}}\frac{l^{\frac{(\beta+1)(\eta-1)}{\eta}}\lambda_2^{(\beta+1)(\eta-1)}}{(\beta+1)^{\frac{(\beta+1)(\eta-1)}{\eta}}}
\\
& \cdot \exp\left(-l\right)\frac{1}{\eta} l^{\frac{1}{\eta}-1}\lambda_2 \frac{1}{(\beta+1)^{\frac{1}{\eta}}}dl.
\end{align*}
Then, we have the following exponents:
\begin{align*}
\eta & \rightarrow (\beta+1)-1=\beta
\\
l  &\rightarrow \frac{\alpha}{\eta} +\frac{(\beta+1)(\eta-1)}{\eta} + \frac{1}{\eta} -1 = \frac{\alpha+(\beta+1)(\eta-1)+1}{\eta}-1
\\
\lambda_2 & \rightarrow -\alpha+\alpha+(\beta+1)(\eta-1)+1-\eta(\beta+1)=-\beta
\\
\beta+1 & \rightarrow -\frac{\alpha}{\eta}- \frac{(\beta+1)(\eta-1)}{\eta}-\frac{1}{\eta}=-\frac{\alpha+(\beta+1)(\eta-1)-1}{\eta}.
\end{align*}
So the integral can be expressed as
\begin{align*}
&\left(\frac{\eta}{\lambda_2}\right)^{\beta}\frac{1}{(\beta+1)^{\frac{\alpha+(\beta+1)(\eta-1)-1}{\eta}}}\int_{\left(\frac{\tau_1}{\lambda_1}\right)^{\eta}(\beta+1)}^{\frac{1}{\lambda_2^{\eta}}\left(\tau_2+\frac{\lambda_2}{\lambda_1}\tau_1-\tau_1\right)} l^{\frac{\alpha+(\beta+1)(\eta-1)+1}{\eta}}\exp\left(-l\right) dl 
\\
&=\left(\frac{\eta}{\lambda_2}\right)^{\beta}\frac{1}{(\beta+1)^{\frac{\alpha+(\beta+1)(\eta-1)-1}{\eta}}} \left\{ \gamma\left(\frac{\alpha+(\beta+1)(\eta-1)+1}{\eta},\frac{1}{\lambda_2}\left(\tau_2+\frac{\lambda_2}{\lambda_1}\tau_1-\tau_1\right) ^{\eta}(\beta+1)\right) \right.
\\
& \left.-\gamma\left(\frac{\alpha+(\beta+1)(\eta-1)+1}{\eta},\left(\frac{\tau_1}{\lambda_1}\right)^{\eta}(\beta+1)\right)\right\}.
\end{align*}

\subsection{Proof Proposition 2}
\begin{align*}
h_{n}^{1,\beta}(a_0,a_1,\eta)&=\int_0^{\tau_1}f_{x_1}\left(t|\lambda_1,\eta\right)^{\beta+1}dt+\int_{\tau_1}^{\tau_2}f_{x_2}\left(t+h|\lambda_1,\lambda_2,\eta\right)^{\beta+1}dt+\left(1- F_{x_2}\left(\tau_2+h|\lambda_1,\lambda_2,\eta\right)\right)^{\beta+1}
\\
&=\zeta_{0,\beta}^{\tau_1}(a_0,a_1,\eta)+\zeta_{0,\beta}^{\tau_1,\tau_2}(a_0,a_1,\eta)+\exp\left(-\frac{\beta+1}{\lambda_2^{\eta}}\left(\tau_2+\frac{\lambda_2}{\lambda_1}\tau_1-\tau_1\right)\right)
\\
&=\left(\frac{\eta}{\lambda_1}\right)^{\beta}\frac{1}{(\beta+1)^{\frac{(\beta+1)(\eta-1)+1}{\eta}}} \gamma\left(\frac{(\beta+1)(\eta-1)+1}{\eta}, \left(\frac{\tau_1}{\lambda_1}\right)^{\eta}(\beta+1)\right)
\\
&+\left(\frac{\eta}{\lambda_1}\right)^{\beta}\frac{1}{(\beta+1)^{\frac{(\beta+)(\eta-1)+1}{\eta}}}\left\{\gamma\left(\frac{(\beta+1)(\eta-1)+1}{\eta},\frac{1}{\lambda_2^{\eta}}\left(\tau_2+\frac{\lambda_2}{\lambda_1}\tau_1-\tau_1\right)^{\eta}(\beta+1)\right)\right.
\\
&-\left. \gamma\left(\frac{(\eta-1)(\beta+1)}{\eta},\left(\frac{\tau_1}{\lambda_1}\right)^{\eta}(\beta+1)\right) \right\}+\exp\left(-\frac{\beta+1}{\lambda_2^{\eta}}\left(\tau_2+\frac{\lambda_2}{\lambda_1}\tau_1-\tau_1\right)\right).
\end{align*}
\subsection{Proof Proposition 4}

For  $J_{11}^{\beta}(a_0)$
\begin{align*}
J_{11}^{\beta}(a_0)&=\int_{0}^{\tau_1}\left(\frac{\partial \log\left(f_{x_1}(t|\lambda_1,\eta)\right)}{\partial a_0}\right)^2 f_{x_1}(t|\lambda_1,\eta)^{\beta+1}dt+\int_{\tau_1}^{\tau_2}\left(\frac{\partial \log\left(f_{x_2}(t+h|\lambda_1,\lambda_2,\eta)\right)}{\partial a_0}\right)^2 f_{x_2}(t+h|\lambda_1,\lambda_2,\eta)^{\beta+1}dt
\\
&+\left(\frac{\partial \log\left(1-F_{x_2}(\tau_2+h|\lambda_1,\lambda_2,\eta)\right)}{\partial a_0}\right)^2 \left(1-F_{x_2}(\tau_2+h|\lambda_1,\lambda_2,\eta)\right)^{\beta+1}
\\
&=J_{11,\tau_1}^{\beta}(a_0)+J_{11,\tau_1 \tau_2}^{\beta} (a_0)+ J_{11,\tau_2}^{\beta}(a_0).
\end{align*}
For $J_{11,\tau_1}^{\beta}(a_0)$:
\begin{equation*}
J_{11,\tau_1}^{\beta}(a_0)=\int_{0}^{\tau_1}\left(\frac{\partial \log\left(f_{x_1}(t|\lambda_1,\eta)\right)}{\partial a_0}\right)^2 f_{x_1}(t|\lambda_1,\eta)^{\beta+1}dt
\end{equation*}
\begin{align}
\log \left(f_{x_1}(t|\lambda_1,\eta)\right)&=\log(\eta)-\eta\log(\lambda_1)+(\eta-1)\log(t)-\left(\frac{t}{\lambda_1}\right)^{\eta} \nonumber
\\
\Rightarrow \frac{\partial  \log \left(f_{x_1}(t|\lambda_1,\eta)\right)}{\partial a_0}&=-\eta +\frac{t^{\eta}}{\lambda_1^{\eta}}=\eta\left(-1+\left(\frac{t}{\lambda_1}\right)^{\eta}\right) \nonumber
\\
\Rightarrow \left( \frac{\partial  \log \left(f_{x_1}(t|\lambda_1,\eta)\right)}{\partial a_0}\right)^2&=\eta^2\left(1-2\left(\frac{t}{\lambda_1}\right)^\eta+\left(\frac{t}{\lambda_1}\right)^{2\eta}\right).
\label{eq:dev_f1_a0}
\end{align}
So
\begin{align*}
J_{11,\tau_1}^{\beta}(a_0)&=\eta^2\left\{ \int_{0}^{\tau_1} f_{x_1}(t|\lambda_1,\eta)^{\beta+1}dt+ \int_{0}^{\tau_1} \left(\frac{t}{\lambda_1}\right)^{2\eta} f_{x_1}(t|\lambda_1,\eta)^{\beta+1}dt\right.
\\
&\left.-2\int_{0}^{\tau_1} \left(\frac{t}{\lambda_1}\right)^{\eta} f_{x_1}(t|\lambda_1,\eta)^{\beta+1}dt\right\}=\eta^2\left\{ \zeta_{0,\beta}^{\tau_1}(a_0,a_1,\eta)+\zeta_{2\eta,\beta}^{\tau_1}(a_0,a_1,\eta)-2\zeta_{\eta,\beta}^{\tau_1}(a_0,a_1,\eta)\right\}.
\end{align*}
In relation to $J_{11,\tau_1\tau_2}^{\beta}(a_0)$ we have
\begin{equation*}
J_{11,\tau_1 \tau_2}^{\beta} (a_0)=\int_{\tau_1}^{\tau_2}\left(\frac{\partial \log\left(f_{x_2}(t+h|\lambda_1,\lambda_2,\eta)\right)}{\partial a_0}\right)^2 f_{x_2}(t+h|\lambda_1,\lambda_2,\eta)^{\beta+1}dt
\end{equation*}
\begin{align*}
\log \left(f_{x_2}(t+h|\lambda_1,\lambda_2,\eta)\right)&=\log(\eta)-\eta\log(\lambda_2)+(\eta-1)\log\left(t+\frac{\lambda_2}{\lambda_1}\tau_1-\tau_1\right)-\left(\frac{t+\frac{\lambda_2}{\lambda_1}\tau_1-\tau_1}{\lambda_2}\right)^{\eta}
\\
\Rightarrow \frac{\partial  \log \left(f_{x_2}(t+h|\lambda_1,\lambda_2,\eta)\right)}{\partial a_0}&=-\eta +(\eta-1)\frac{1}{t+\frac{\lambda_2}{\lambda_1}\tau_1-\tau_1}\tau_1\frac{\partial}{\partial a_0}\left(\frac{\lambda_2}{\lambda_1}\right)
\\
&-\frac{\eta\left(t+\frac{\lambda_2}{\lambda_1}\tau_1-\tau_1\right)^{\eta-1}\tau_1\frac{\partial}{\partial a_0}\left(\frac{\lambda_2}{\lambda_1}\right)\lambda_2^{\eta}}{\lambda_2^{2\eta}}+\frac{\eta\lambda_2^{\eta-1}\frac{\partial \lambda_2}{\partial a_0}\left(t+\frac{\lambda_2}{\lambda_1}\tau_1-\tau_1\right)^{\eta}}{\lambda_2^{2\eta}}
\end{align*}
\begin{align}
\lambda_2&=\exp(a_0+a_1x_2),\quad \lambda_1=\exp(a_0+a_1x_1) \Rightarrow \frac{\lambda_2}{\lambda_1}=\exp(a_1(x_2-x_1)) \nonumber
\\
&\Rightarrow \frac{\partial}{\partial a_0}\left(\frac{\lambda_2}{\lambda_1}\right)=0 \nonumber
\\
&\Rightarrow\frac{\partial  \log \left(f_{x_2}(t+h|\lambda_1,\lambda_2,\eta)\right)}{\partial a_0}=-\eta+\frac{\eta}{\lambda_2^{\eta}}\left(t+\frac{\lambda_2}{\lambda_1}\tau_1-\tau_1\right)^{\eta}=\eta\left(-1+\left(\frac{t+\frac{\lambda_2}{\lambda_1}\tau_1-\tau_1}{\lambda_2}\right)^{\eta}\right)
\label{eq:dev_f2_a0}
\end{align}
\begin{align*}
\Rightarrow \left( \frac{\partial  \log \left(f_{x_2}(t+h|\lambda_1,\lambda_2,\eta)\right)}{\partial a_0}\right)^2&=\eta^2\left(1+\left(\frac{t+\frac{\lambda_2}{\lambda_1}\tau_1-\tau_1}{\lambda_2}\right)^{2\eta}-2\left(\frac{t+\frac{\lambda_2}{\lambda_1}\tau_1-\tau_1}{\lambda_2}\right)^{\eta}\right).
\end{align*}
Then
\begin{align*}
J_{11,\tau_1\tau_2}^{\beta}(a_0)&=\eta^2\left\{\int_{\tau_1}^{\tau_2}f_{x_2}(t+h|\lambda_1,\lambda_2,\eta)^{\beta+1}dt+\int_{\tau_1}^{\tau_2}\left(\frac{t+\frac{\lambda_2}{\lambda_1}\tau_1-\tau_1}{\lambda_2}\right)^{2\eta}f_{x_2}(t+h|\lambda_1,\lambda_2,\eta)^{\beta+1}dt \right.
\\
&\left. -2\int_{\tau_1}^{\tau_2} \left(\frac{t+\frac{\lambda_2}{\lambda_1}\tau_1-\tau_1}{\lambda_2}\right)^\eta f_{x_2}(t+h|\lambda_1,\lambda_2,\eta)^{\beta+1}dt\right\}
\\
&=\eta^{2}\left\{\zeta_{0,\beta}^{\tau_1,\tau_2}(a_0,a_1,\eta)+\zeta_{2\eta,\beta}^{\tau_1,\tau_2}(a_0,a_1,\eta)-2\zeta_{\eta,\beta}^{\tau_1,\tau_2}(a_0,a_1,\eta)\right\}.
\end{align*}
In relation to $J_{11,\tau_2}^{\beta}(a_0)$ we have
\begin{equation*}
J_{11,\tau_2}^{\beta} (a_0)=\left(\frac{\partial \log\left(1-F_{x_2}(\tau_2+h|\lambda_1,\lambda_2,\eta)\right)}{\partial a_0}\right)\left(1-F_{x_2}(\tau_2+h|\lambda_1,\lambda_2,\eta)\right)^{\beta+1}
\end{equation*}
\begin{align*}
&\log\left(1-F_{x_2}(\tau_2+h|\lambda_1,\lambda_2,\eta)\right)=-\left(\frac{\tau_2+\frac{\lambda_2}{\lambda_1}\tau_1-\tau_1}{\lambda_2}\right)^\eta
\\
&\Rightarrow \frac{\partial \log\left(1-F_{x_2}(\tau_2+h|\lambda_1,\lambda_2,\eta)\right)}{\partial a_0}=-\eta \left(\frac{\tau_2+\frac{\lambda_2}{\lambda_1}\tau_1-\tau_1}{\lambda_2}\right)^{\eta-1} \frac{\frac{\partial}{\partial a_0}\left(\frac{\lambda_2}{\lambda_1}\right)-\frac{\partial\lambda_2}{\partial a_0}(\tau_2+\frac{\lambda_2}{\lambda_1}\tau_1-\tau_1)}{\lambda_2^2}.
\end{align*}
Noting that
\begin{align*}
\frac{\partial}{\partial a_0}\left(\frac{\lambda_2}{\lambda_1}\right)&=0;
\\
\frac{\partial \lambda_2}{\partial a_0}&=\frac{\partial}{\partial a_0} \exp(a_0+a_1x_2) = \exp(a_0+a_1x_2)=\lambda_2.
\end{align*}
We have
\begin{align*}
&\frac{\partial \log\left(1-F_{x_2}(\tau_2+h|\lambda_1,\lambda_2,\eta)\right)}{\partial a_0}=\eta \left(\frac{\tau_2+\frac{\lambda_2}{\lambda_1}\tau_1-\tau_1}{\lambda_2}\right)^{\eta-1} \frac{\tau_2+ \frac{\lambda_2}{\lambda_1}\tau_1-\tau_1}{\lambda_2}=\eta \left(\frac{\tau_2+\frac{\lambda_2}{\lambda_1}\tau_1-\tau_1}{\lambda_2}\right)^{\eta} 
\\
&\Rightarrow \left(\frac{\partial \log\left(1-F_{x_2}(\tau_2+h|\lambda_1,\lambda_2,\eta)\right)}{\partial a_0}\right)^2=\eta^2 \left(\frac{\tau_2+\frac{\lambda_2}{\lambda_1}\tau_1-\tau_1}{\lambda_2}\right)^{2\eta} 
\\
&\Rightarrow  J_{11,\tau_2}^{\beta}(a_0)=\eta^2 \left(\frac{\tau_2+\frac{\lambda_2}{\lambda_1}\tau_1-\tau_1}{\lambda_2}\right)^{2\eta} \exp\left(-\left(\frac{\tau_2+\frac{\lambda_2}{\lambda_1}\tau_1-\tau_1}{\lambda_2}\right)^{\eta}(\beta+1)\right).
\end{align*}

\subsection{Proof Proposition 5}

For $J_{22}^{\beta}(a_1)$ we have:
\begin{align}
J_{22}^{\beta}(a_1)&=\int_{0}^{\tau_1}\left(\frac{\partial \log\left(f_{x_1}(t|\lambda_1,\eta)\right)}{\partial a_1}\right)^2 f_{x_1}(t|\lambda_1,\eta)^{\beta+1}dt+\int_{\tau_1}^{\tau_2}\left(\frac{\partial \log\left(f_{x_2}(t+h|\lambda_1,\lambda_2,\eta)\right)}{\partial a_1}\right)^2 f_{x_2}(t+h|\lambda_1,\lambda_2,\eta)^{\beta+1}dt \nonumber
\\
&+\left(\frac{\partial \log\left(1-F_{x_2}(\tau_2+h|\lambda_1,\lambda_2,\eta)\right)}{\partial a_1}\right)^2 \left(1-F_{x_2}(\tau_2+h|\lambda_1,\lambda_2,\eta)\right)^{\beta+1} \nonumber
\\
&=J_{22,\tau_1}^{\beta}(a_1)+J_{22,\tau_1 \tau_2}^{\beta} (a_1)+ J_{22,\tau_2}^{\beta}(a_1).
\label{eq:dev_F2_a0}
\end{align}
In relation to $J_{22,\tau_1}^{\beta}(a_1)$ we have:
\begin{align}
J_{22,\tau_1}^{\beta}&=\int_{0}^{\tau_1}\left(\frac{\partial \log\left(f_{x_1}(t|\lambda_1,\eta)\right)}{\partial a_1}\right)^2 f_{x_1}(t|\lambda_1,\eta)^{\beta+1}dt \nonumber
\\
\log\left(f_{x_1}(t|\lambda_1,\eta)\right)&=\log(\eta)-\eta\log(\lambda_1)+(\eta-1)\log(t)-\left(\frac{t}{\lambda_1}\right)^{\eta} \nonumber
 \\
 \Rightarrow \frac{\partial \log\left(f_{x_1}(t|\lambda_1,\eta)\right)}{\partial a_1}&=-\eta x_1+t^{\eta}\eta\lambda_1^{\eta-1}\lambda_1x_1=\eta x_1 \left(-1+\left(\frac{t}{\lambda_1}\right)^{\eta}\right) \nonumber
 \\ 
 \Rightarrow \left(\frac{\partial \log\left(f_{x_1}(t|\lambda_1,\eta)\right)}{\partial a_1}\right)^2&=\eta^2 x_1^2 \left\{1+\left(\frac{t}{\lambda_1}\right)^{2\eta}-2\left(\frac{t}{\lambda_1}\right)^{\eta}\right\} \nonumber
 \\
 \Rightarrow J_{22,\tau_1}^{\beta}(a_1)&=\eta^2 x_1^2\left\{\int_{0}^{\tau_1} f_{x_1}(t|\lambda_1,\eta)^{\beta+1}dt+\int_{0}^{\tau_1}\left(\frac{t}{\lambda_1}\right)^{2\eta} f_{x_1}(t|\lambda_1,\eta)^{\beta+1}dt  \right. \nonumber
 \\
 &\left. -2\int_{0}^{\tau_1} \left(\frac{t}{\lambda_1}\right)^{\eta}f_{x_1}(t|\lambda_1,\eta)^{\beta+1}dt\right\} \nonumber
 \\
 \Rightarrow J_{22,\tau_1}^{\beta}(a_1)&=x_1^2 J_{11,\tau_1}^{\beta}(a_0).
 \label{eq:dev_f1_a1}
\end{align}
In relation with $J_{22,\tau_1 \tau_2}^{\beta}(a_1)$
\begin{align*}
J_{22,\tau_1 \tau_2}^{\beta}(a_1)&=\int_{\tau_1}^{\tau_2}  \frac{\partial \log\left(f_{x_2}(t+h|\lambda_1,\lambda_2,\eta)\right)}{\partial a_1}f_{x_2}(t+h|\lambda_1,\lambda_2,\eta)^{\beta+1}dt
\\
\log\left(f_{x_2}(t+h|\lambda_1,\lambda_2,\eta)\right)&=\log(\eta)-\eta \log(\lambda_2)+(\eta-1)\log\left(t+\frac{\lambda_2}{\lambda_1}\tau_1-\tau_1\right)-\frac{1}{\lambda_2^{\eta}}\left(t+\frac{\lambda_2}{\lambda_1}\tau_1-\tau_1\right)^\eta
\end{align*}
\begin{align*}
\begin{cases}
\lambda_1&=\exp(a_0+a_1x_1) 
\\
\lambda_2&=\exp(a_0+a_1x_2) 
\end{cases}
\Rightarrow
\begin{cases}
 \frac{\partial \lambda_1}{\partial a_0}& = \lambda_1; \quad  \frac{\partial \lambda_1}{\partial a_1} = x_1\lambda_1;\quad  \frac{\partial}{\partial a_0}\left(\frac{\lambda_2}{\lambda_1}\right)=0
 \\
  \frac{\partial \lambda_2}{\partial a_0}& = \lambda_2; \quad  \frac{\partial \lambda_2}{\partial a_1} = x_2\lambda_2;\quad  \frac{\partial}{\partial a_1}\left(\frac{\lambda_2}{\lambda_1}\right)=\frac{\lambda_2}{\lambda_1}(x_2-x_1).
 \end{cases}
\end{align*}
Then, we have
\begin{align}
\frac{\partial \log\left(f_{x_2}(t+h|\lambda_1,\lambda_2,\eta)\right)}{\partial a_1}&=-\eta x_2 + (\eta-1)\frac{\tau_1}{t+\frac{\lambda_2}{\lambda_1}\tau_1-\tau_1}\frac{\partial}{\partial a_1}\left(\frac{\lambda_2}{\lambda_1}\right)\\
&-\frac{\eta \lambda_2^\eta\left(t+\frac{\lambda_2}{\lambda_1}\tau_1-\tau_1\right)^{\eta-1}\tau_1\frac{\partial}{\partial a_1}\left(\frac{\lambda_2}{\lambda_1}\right)-\eta\lambda_2^{\eta-1}x_2\lambda_2\left(t+\frac{\lambda_2}{\lambda_1}\tau_1-\tau_1\right)^{\eta}}{\lambda_2^{2\eta}} \nonumber
\\
&=-\eta x_2 +(\eta-1)\frac{\tau_1 \frac{\lambda_2}{\lambda_1}(x_2-x_1)}{t+\frac{\lambda_2}{\lambda_1}\tau_1-\tau_1}-\frac{\eta \tau_1}{\lambda_2^{\eta+1}}\left(\frac{t+\frac{\lambda_2}{\lambda_1}\tau_1-\tau_1}{\lambda_2}\right)^{\eta-1}\lambda_2^{\eta}\frac{\lambda_2}{\lambda_1}(x_2-x_1) \nonumber
\\
&+\frac{\eta \lambda_2^{\eta}x_2}{\lambda_2^{\eta}}\left(\frac{t+\frac{\lambda_2}{\lambda_1}\tau_1-\tau_1}{\lambda_2}\right)^{\eta} \nonumber
\\
&=-\eta x_2 + \frac{\eta-1}{\lambda_1}\tau_1(x_2-x_1)\frac{1}{\frac{t+\frac{\lambda_2}{\lambda_1}\tau_1-\tau_1}{\lambda_2}}-\frac{\eta\tau_1(x_2-x_1)}{\lambda_1}\left(\frac{t+\frac{\lambda_2}{\lambda_1}\tau_1-\tau_1}{\lambda_2}\right)^{\eta-1} \nonumber
\\
&+\eta x_2 \left(\frac{t+\frac{\lambda_2}{\lambda_1}\tau_1-\tau_1}{\lambda_2}\right)^{\eta} 
 \label{eq:dev_f2_a1}
\end{align}
\begin{align*}
\Rightarrow \left(\frac{\partial \log\left(f_{x_2}(t+h|\lambda_1,\lambda_2,\eta)\right)}{\partial a_1}\right)^{2}&=\eta^2 x_2^2 +\frac{(\eta-1)^2\tau_1^2}{\lambda_1^2}(x_2-x_1)^2\frac{1}{\left(\frac{t+\frac{\lambda_2}{\lambda_1}\tau_1-\tau_1}{\lambda_2}\right)}
\\
&+\frac{\eta^2\tau_1^2}{\lambda_1^2}(x_2-x_1)^2\left(\frac{t+\frac{\lambda_2}{\lambda_1}\tau_1-\tau_1}{\lambda_2}\right)^{2(\eta-1)}+\eta^2x_2^2\left(\frac{t+\frac{\lambda_2}{\lambda_1}\tau_1-\tau_1}{\lambda_2}\right)^{2\eta}
\\
&-2\eta x_2\frac{(\eta-1)\tau_1}{\lambda_1}(x_2-x_1)\frac{1}{\left(\frac{t+\frac{\lambda_2}{\lambda_1}\tau_1-\tau_1}{\lambda_2}\right)}
\\
&+2\frac{\eta^2\tau_1 x_2(x_2-x_1)}{\lambda_1}\left(\frac{t+\frac{\lambda_2}{\lambda_1}\tau_1-\tau_1}{\lambda_2}\right)^{\eta-1} -2\eta^2x_2^2\left(\frac{t+\frac{\lambda_2}{\lambda_1}\tau_1-\tau_1}{\lambda_2}\right)^{\eta}
\\
&-2\frac{(\eta-1) \eta\tau_1^2(x_2-x_1)^2}{\lambda_1^2}\left(\frac{t+\frac{\lambda_2}{\lambda_1}\tau_1-\tau_1}{\lambda_2}\right)^{\eta-2}
\\
&+2\frac{(\eta-1)\tau_1}{\lambda_1}\eta x_2 (x_2-x_1)\left(\frac{t+\frac{\lambda_2}{\lambda_1}\tau_1-\tau_1}{\lambda_2}\right)^{\eta-1}
\\
&-2\frac{\eta^2\tau_1 x_2 (x_2-x_1)}{\lambda_1}\left(\frac{t+\frac{\lambda_2}{\lambda_1}\tau_1-\tau_1}{\lambda_2}\right)^{2\eta-1}.
\end{align*}
So, this derivative consists of 10 summands. Each summand is labeled $c_i$, $i=1,..., 10$ according to the order in which it appears.

Then, we have
\begin{align*}
J_{22,\tau_1\tau_2}^{\beta}(a_1)=\sum_{i=i}^{10}\int_{\tau_1}^{\tau_2} c_i f_{x_2}(t+h|\lambda_1,\lambda_2,\eta)^{\beta+1}dt=\sum_{i=1}^{10}c_i^*.
\end{align*}
Now we are going to obtain each value of $c_i^*, i=1, ..., 10$
\begin{align*}
c_1^*&=\eta^2 x_2^2\int_{\tau_1}^{\tau_2}  f_{x_2}(t+h|\lambda_1,\lambda_2,\eta)^{\beta+1}dt=\eta^2 x_2^2 \zeta_{0,\beta}^{\tau_1,\tau_2}(a_0,a_1,\eta)
\\
c_2^*&=\frac{(\eta-1)^2 \tau_1^2}{\lambda_1^2}(x_2-x_1)^2\int_{\tau_1}^{\tau_2} \left(\frac{t+\frac{\lambda_2}{\lambda_1}\tau_1-\tau_1}{\lambda_2}\right)^{-2} f_{x_2}(t+h|\lambda_1,\lambda_2,\eta)^{\beta+1}dt
	\\
	&=\frac{(\eta-1)^2 \tau_1^2}{\lambda_1^2}(x_2-x_1)^2\zeta_{-2,\beta}^{\tau_1,\tau_2}(a_0,a_1,\eta)
\\
c_3^*&=\frac{\eta^2\tau_1^2}{\lambda_1^2}(x_2-x_1)^2\int_{\tau_1}^{\tau_2} \left(\frac{t+\frac{\lambda_2}{\lambda_1}\tau_1-\tau_1}{\lambda_2}\right)^{2(\eta-1)} f_{x_2}(t+h|\lambda_1,\lambda_2,\eta)^{\beta+1}dt
	\\
	&=\frac{\eta^2\tau_1^2}{\lambda_1^2}(x_2-x_1)^2\zeta_{2(\eta-1),\beta}^{\tau_1,\tau_2}(a_0,a_1,\eta)
\\
c_4^*&=\eta^2 x_2^2\int_{\tau_1}^{\tau_2} \left(\frac{t+\frac{\lambda_2}{\lambda_1}\tau_1-\tau_1}{\lambda_2}\right)^{2\eta} f_{x_2}(t+h|\lambda_1,\lambda_2,\eta)^{\beta+1}dt
	\\
	&=\eta^2 x_2^2 \zeta_{2\eta,\beta}^{\tau_1,\tau_2}(a_0,a_1,\eta)
\\
c_5^*&=-2\eta x_2 \frac{(\eta-1)\tau_1}{\lambda_1}(x_2-x_1)\int_{\tau_1}^{\tau_2} \left(\frac{t+\frac{\lambda_2}{\lambda_1}\tau_1-\tau_1}{\lambda_2}\right)^{-1} f_{x_2}(t+h|\lambda_1,\lambda_2,\eta)^{\beta+1}dt
	\\
	&=-2\eta x_2 \frac{(\eta-1)\tau_1}{\lambda_1}(x_2-x_1) \zeta_{-1,\beta}^{\tau_1,\tau_2}(a_0,a_1,\eta)
\\
c_6^*&=2\frac{\eta^2\tau_1x_2(x_2-x_1)}{\lambda_1}\int_{\tau_1}^{\tau_2} \left(\frac{t+\frac{\lambda_2}{\lambda_1}\tau_1-\tau_1}{\lambda_2}\right)^{\eta-1} f_{x_2}(t+h|\lambda_1,\lambda_2,\eta)^{\beta+1}dt
	\\
	&=2\frac{\eta^2\tau_1x_2(x_2-x_1)}{\lambda_1} \zeta_{\eta-1,\beta}^{\tau_1,\tau_2}(a_0,a_1,\eta)
\\
c_7^*&=-2\eta^2 x_2^2\int_{\tau_1}^{\tau_2} \left(\frac{t+\frac{\lambda_2}{\lambda_1}\tau_1-\tau_1}{\lambda_2}\right)^{\eta} f_{x_2}(t+h|\lambda_1,\lambda_2,\eta)^{\beta+1}dt
	\\
	&=-2\eta^2 x_2^2 \zeta_{\eta,\beta}^{\tau_1,\tau_2}(a_0,a_1,\eta)
\\
c_8^*&=-2\frac{(\eta-1) \eta\tau_1^2(x_2-x_1)^2}{\lambda_1^2}\int_{\tau_1}^{\tau_2} \left(\frac{t+\frac{\lambda_2}{\lambda_1}\tau_1-\tau_1}{\lambda_2}\right)^{\eta-2} f_{x_2}(t+h|\lambda_1,\lambda_2,\eta)^{\beta+1}dt
	\\
	&=-2\frac{(\eta-1) \eta\tau_1^2(x_2-x_1)^2}{\lambda_1^2} \zeta_{\eta-2,\beta}^{\tau_1,\tau_2}(a_0,a_1,\eta)
\\
c_9^*&=2\frac{(\eta-1)\tau_1}{\lambda_1}\eta x_2 (x_2-x_1)\int_{\tau_1}^{\tau_2} \left(\frac{t+\frac{\lambda_2}{\lambda_1}\tau_1-\tau_1}{\lambda_2}\right)^{\eta-1} f_{x_2}(t+h|\lambda_1,\lambda_2,\eta)^{\beta+1}dt
	\\
	&=2\frac{(\eta-1)\tau_1}{\lambda_1}\eta x_2 (x_2-x_1) \zeta_{\eta-1,\beta}^{\tau_1,\tau_2}(a_0,a_1,\eta)
\\
c_{10}^*&=-2\frac{\eta^2\tau_1}{\lambda_1}(x_2-x_1)x_2\int_{\tau_1}^{\tau_2} \left(\frac{t+\frac{\lambda_2}{\lambda_1}\tau_1-\tau_1}{\lambda_2}\right)^{2\eta-1} f_{x_2}(t+h|\lambda_1,\lambda_2,\eta)^{\beta+1}dt
	\\
	&=-2\frac{\eta^2\tau_1}{\lambda_1}(x_2-x_1)x_2 \zeta_{2\eta-1,\beta}^{\tau_1,\tau_2}(a_0,a_1,\eta).
\end{align*}
And for $J_{22,\tau_2}^{\beta}(a_1)$ we have
\begin{align*}
J_{22,\tau_2}^{\beta}(a_1)=\left(\frac{\partial \log\left(1-F_{x_2}(\tau_2+h|\lambda_1,\lambda_2,\eta)\right)}{\partial a_1}\right)^{2} \left(1-F_{x_2}(\tau_2+h|\lambda_1,\lambda_2,\eta)\right)^{\beta+1}
\end{align*}
\begin{align}
\log\left(1-F_{x_2}(\tau_2+h|\lambda_1,\lambda_2,\eta)\right)&=-\left(\frac{\tau_2+\frac{\lambda_2}{\lambda_1}\tau_1-\tau_1}{\lambda_2}\right)^{\eta} \nonumber
\\
\Rightarrow \frac{\partial \log\left(1-F_{x_2}(\tau_2+h|\lambda_1,\lambda_2,\eta)\right)}{\partial a_1}&=-\eta\left(\frac{\tau_2+\frac{\lambda_2}{\lambda_1}\tau_1-\tau_1}{\lambda_2}\right)^{\eta-1} \frac{\lambda_2\frac{\lambda_2}{\lambda_1}(x_2-x_1)\tau_1-\lambda_2x_2\left(\tau_2+\frac{\lambda_2}{\lambda_1}\tau_1-\tau_1\right)}{\lambda_2^2}  \nonumber
\\
&=\eta\left(\frac{\tau_2+\frac{\lambda_2}{\lambda_1}\tau_1-\tau_1}{\lambda_2}\right)^{\eta-1}\frac{\frac{x_1\tau_1}{\lambda_1}+x_2(\tau_2-\tau_1)}{\lambda_2}
 \label{eq:dev_F2_a1}
\end{align}
\begin{align*}
&\Rightarrow \left(\frac{\partial \log\left(1-F_{x_2}(\tau_2+h|\lambda_1,\lambda_2,\eta)\right)}{\partial a_1}\right)^2 \left(1-F_{x_2}(\tau_2+h|\lambda_1,\lambda_2,\eta)\right)^{\beta+1}
\\
&=\eta^2\left(\frac{\tau_2+\frac{\lambda_2}{\lambda_1}\tau_1-\tau_1}{\lambda_2}\right)^{2\eta-2}\left(\frac{\frac{x_1\tau_1}{\lambda_1}+x_2(\tau_2-\tau_1)}{\lambda_2}\right)^2\exp\left(-\left(\frac{\tau_2+\frac{\lambda_2}{\lambda_1}\tau_1-\tau_1}{\lambda_2}\right)^{\eta}(\beta+1)\right).
\end{align*}

\subsection{Proof Lemma 6}
For the first function:
\begin{align*}
H_{\alpha,\gamma,\beta}^{\tau_1}(a_0,a_1,\eta)&=\int_{0}^{\tau_1}\left(\frac{t}{\lambda_1}\right)\left(\log\left(\frac{t}{\lambda_1}\right)\right)^{\gamma}g_{x_1}(t)^{\beta+1}dt.
\end{align*}  
Doing the following change of variable:
\begin{align*}
 l&=\frac{t}{\lambda_1} \text{, we have}
\\
H_{\alpha,\gamma,\beta}^{\tau_1}(a_0,a_1,\eta)&=\int_{0}^{\frac{\tau_1}{\lambda_1}} l^{\alpha}  \left(\log\left(l\right)\right)^{\gamma}\frac{\eta^{\beta+1}}{\lambda_1^{\beta+1}}\lambda_1^{(\beta+1)(\eta-1)} l^{(\beta+1)(\eta-1)}\exp\left(-l^{\eta}(\beta+1)\right)\lambda_1dl
\\
&=\lambda_1\left(\frac{\eta}{\lambda_1}\right)^{\beta+1}\int_{0}^{\frac{\tau_1}{\lambda_1}}l^{\alpha+(\eta-1)(\beta+1)}\left(\log\left(l\right)\right)^{\gamma}\exp\left(-l^{\eta}(\beta+1)\right)dl.
\end{align*}  
For the second function and the considered density function for the Weibull lifetime we have
\begin{align*}
H_{\alpha,\gamma,\beta}^{\tau_1,\tau_2}(a_0,a_1,\eta)&=\int_{\tau_1}^{\tau_2}\left(\frac{t+\frac{\lambda_2}{\lambda_1}\tau_1-\tau_1}{\lambda_2}\right)^{\alpha}\left(\log\left(\frac{t+\frac{\lambda_2}{\lambda_1}\tau_1-\tau_1}{\lambda_2}\right)\right)^{\gamma}\frac{\eta^{\beta+1}}{\lambda_2^{(\beta+1)\eta}}\left(t+\frac{\lambda_2}{\lambda_1}\tau_1-\tau_1\right)^{(\beta+1)(\eta-1)}
\\
& \cdot \exp\left(-\frac{1}{\lambda_2^{\eta}}\left(t+\frac{\lambda_2}{\lambda_1}\tau_1-\tau_1\right)^{\eta}(\beta+1)\right)dt.
\end{align*}
Doing the following change of variable:
\begin{align*}
l=\frac{t+\frac{\lambda_2}{\lambda_1}\tau_1-\tau_1}{\lambda_2} \rightarrow dt=\lambda_2dl \text{, so we have}
\end{align*}
\begin{align*}
t=\tau_1& \Rightarrow l=\frac{\tau_1+\frac{\lambda_2}{\lambda_1}\tau_1-\tau_1}{\lambda_2}=\frac{\tau_1}{\lambda_1}
\\
t=\tau_2&  \Rightarrow l=\frac{\tau_2+\frac{\lambda_2}{\lambda_1}\tau_1-\tau_1}{\lambda_2}
\\
H_{\alpha,\gamma,\beta}^{\tau_1,\tau_2}(a_0,a_1,\eta)&=\int_{\frac{\tau_1}{\lambda_1}}^{\frac{1}{\lambda_2}\left(\tau_2+\frac{\lambda_2}{\lambda_1}\tau_1-\tau_1\right)}l^{\alpha}\left(\log\left(l\right)\right)^{\gamma}\frac{\eta^{\beta+1}}{\lambda_2^{\eta(\beta+1)}}\lambda_2^{(\eta-1)(\beta+1)}\exp\left(-l^{\eta}(\beta+1)\right)\lambda_2dl
\\
&=\lambda_2\left(\frac{\eta}{\lambda_2}\right)^{\beta+1}\int_{\frac{\tau_1}{\lambda_1}}^{\left(\tau_2+\frac{\lambda_2}{\lambda_1}\tau_1-\tau_1\right)\frac{1}{\lambda_2}}l^{\alpha+(\eta-1)(\beta+1)}\left(\log\left(l\right)\right)^{\gamma}\exp\left(-l^{\eta}(\beta+1)\right)dl.
\end{align*}

\subsection{Proof Propostion 7}
For $J_{33}^{\beta}(\eta)$
\begin{align*}
J_{33}^{\beta}(\eta)&=\int_{0}^{\tau_1}\left(\frac{\partial \log(f_{x_1}(t|\lambda_1,\eta))}{\partial \eta} \right)^2f_{x_1}(t|\lambda_1,\eta)^{\beta+1}dl
\\
&+\int_{\tau_1}^{\tau_2}\left(\frac{\partial \log(f_{x_2}(t+h|\lambda_1,\lambda_2,\eta))}{\partial \eta} \right)^2f_{x_2}(t+h|\lambda_1,\lambda_2,\eta)^{\beta+1}dl
\\
&+\left(\frac{\partial \log(1-F_{x_2}(\tau_2+h|\lambda_1,\lambda_2,\eta))}{\partial \eta} \right)^2 \log(1-F_{x_2}(\tau_2+h|\lambda_1,\lambda_2,\eta)^{\beta+1}
\\&=J_{33,\tau_1}^{\beta}(\eta)+J_{33,\tau_1\tau_2}^{\beta}(\eta)+J_{33,\tau_2}^{\beta}(\eta).
\end{align*}
For $J_{33,\tau_1}^{\beta}(\eta)$
\begin{align}
f_{x_1}(t|\lambda_1,\eta)&=\frac{\eta}{\lambda_1^{\eta}}t^{\eta-1}\exp\left(-\left(\frac{t}{\lambda_1}\right)^{\eta}\right) \nonumber
\\
\Rightarrow \log\left(f_{x_1}(t|\lambda_1,\eta)\right)&=\log(\eta)-\log(\lambda_1)+(\eta-1)\log\left(\frac{t}{\lambda_1}\right)-\left(\frac{t}{\lambda_1}\right)^{\eta} \nonumber
\\
\Rightarrow \frac{\partial \log\left(f_{x_1}(t|\lambda_1,\eta)\right)}{\partial \eta}&=\frac{1}{\eta}+\log\left(\frac{t}{\lambda_1}\right)-\left(\frac{t}{\lambda_1}\right)^{\eta}\log\left(\frac{t}{\lambda_1}\right) \nonumber
\\
\Rightarrow \left(\frac{\partial \log\left(f_{x_1}(t|\lambda_1,\eta)\right)}{\partial \eta}\right)^2&=\frac{1}{\eta^2}+\left(\log\left(\frac{t}{\lambda_1}\right)\right)^2+\left(\frac{t}{\lambda_1}\right)^{2\eta}\left(\log\left(\frac{t}{\lambda_1}\right)\right)^2+2\frac{1}{\eta}\log\left(\frac{t}{\lambda_1}\right) \nonumber
\\
&-2\frac{1}{\eta}\left(\frac{t}{\lambda_1}\right)^{\eta}\log\left(\frac{t}{\lambda_1}\right)-2\left(\frac{t}{\lambda_1}\right)^{\eta}\left(\log\left(\frac{t}{\lambda_1}\right)\right)^2 .
 \label{eq:dev_f1_eta}
\end{align}
Then
\begin{align*}
J_{33,\tau_1}^{\beta}(\eta)&=\frac{1}{\eta^2}\int_{0}^{\tau_1}f_{x_1}(t|\lambda_1,\eta)^{\beta+1}dt+\int_{0}^{\tau_1}\left(\log\left(\frac{t}{\lambda_1}\right)\right)^2f_{x_1}(t|\lambda_1,\eta)^{\beta+1}dt
\\
&+\int_{0}^{\tau_1}\left(\frac{t}{\lambda_1}\right)^{2\eta}\left(\log\left(\frac{t}{\lambda_1}\right)\right)^2f_{x_1}(t|\lambda_1,\eta)^{\beta+1}dt+\frac{2}{\eta}\int_{0}^{\tau_1}\log\left(\frac{t}{\lambda_1}\right)f_{x_1}(t|\lambda_1,\eta)^{\beta+1}dt
\\
&-\frac{2}{\eta}\int_{0}^{\tau_1}\left(\frac{t}{\lambda_1}\right)^{\eta}\log\left(\frac{t}{\lambda_1}\right)f_{x_1}(t|\lambda_1,\eta)^{\beta+1}dt-2\int_{0}^{\tau_1}\left(\frac{t}{\lambda_1}\right)^{\eta}\left(\log\left(\frac{t}{\lambda_1}\right)\right)^{2}f_{x_1}(t|\lambda_1,\eta)^{\beta+1}dt
\\
&=A_1+A_2+A_3+A_4+A_5+A_6.
\end{align*}
Then, we have
\begin{align*}
A_1&=\frac{1}{\eta^2}\int_{0}^{\tau_1}f_{x_1}(t|\lambda_1,\eta)^{\beta+1}dt=\frac{1}{\eta^2}H_{0,0,\beta}^{\tau_1}(a_0,a_1,\eta)
\\
A_2&=\int_{0}^{\tau_1}\left(\log\left(\frac{t}{\lambda_1}\right)\right)^2f_{x_1}(t|\lambda_1,\eta)^{\beta+1}dt=H_{0,2,\beta}^{\tau_1}(a_0,a_1,\eta)
\\
A_3&=\int_{0}^{\tau_1}\left(\frac{t}{\lambda_1}\right)^{2\eta}\left(\log\left(\frac{t}{\lambda_1}\right)\right)^2f_{x_1}(t|\lambda_1,\eta)^{\beta+1}dt=H_{2\eta,2,\beta}^{\tau_1}(a_0,a_1,\eta)
 \\
 A_4&=\frac{2}{\eta}\int_{0}^{\tau_1}\log\left(\frac{t}{\lambda_1}\right)f_{x_1}(t|\lambda_1,\eta)^{\beta+1}dt=\frac{2}{\eta}H_{0,1,\beta}^{\tau_1}(a_0,a_1,\eta)
 \\
 A_5&=-\frac{2}{\eta}\int_{0}^{\tau_1}\left(\frac{t}{\lambda_1}\right)^{\eta}\log\left(\frac{t}{\lambda_1}\right)f_{x_1}(t|\lambda_1,\eta)^{\beta+1}dt=-\frac{2}{\eta}H_{\eta,1,\beta}^{\tau_1}(a_0,a_1,\eta)
 \\
 A_6&=-2\int_{0}^{\tau_1}\left(\frac{t}{\lambda_1}\right)^{\eta}\left(\log\left(\frac{t}{\lambda_1}\right)\right)^{2}f_{x_1}(t|\lambda_1,\eta)^{\beta+1}dt=-2H_{\eta,2,\beta}^{\tau_1}(a_0,a_1,\eta).
\end{align*}
Finally, we have
\begin{align*}
J_{33,\tau_1}^{\beta}(\eta)=&\frac{1}{\eta^2}H_{0,0,\beta}^{\tau_1}(a_0,a_1,\eta)+H_{0,2,\beta}^{\tau_1}(a_0,a_1,\eta)+H_{2\eta,2,\beta}^{\tau_1}(a_0,a_1,\eta)
\\
&\frac{2}{\eta}H_{0,1,\beta}^{\tau_1}(a_0,a_1,\eta)-\frac{2}{\eta}H_{\eta,1,\beta}^{\tau_1}(a_0,a_1,\eta)-2H_{\eta,2,\beta}^{\tau_1}(a_0,a_1,\eta).
\end{align*}
For the term $J_{33,\tau_1\tau_2}^{\beta}(\eta)$
\begin{align}
f_{x_2}(t+h|\lambda_1,\lambda_2,\eta)&=\frac{\eta}{\lambda_2}\left(\frac{t+\frac{\lambda_2}{\lambda_1}\tau_1-\tau_1}{\lambda_2}\right)^{\eta-1}\exp\left(-\left(\frac{t+\frac{\lambda_2}{\lambda_1}\tau_1-\tau_1}{\lambda_2}\right)^{\eta}\right) \nonumber
\\
\Rightarrow \log\left(f_{x_2}(t+h|\lambda_1,\lambda_2,\eta)\right)&=\log(\eta)-\log(\lambda_2)+(\eta-1)\log\left(\frac{t+\frac{\lambda_2}{\lambda_1}\tau_1-\tau_1}{\lambda_2}\right)-\left(\frac{t+\frac{\lambda_2}{\lambda_1}\tau_1-\tau_1}{\lambda_2}\right)^{\eta} \nonumber
\\
\Rightarrow \frac{\partial \log\left(f_{x_2}(t+h|\lambda_1,\lambda_2,\eta)\right)}{\partial \eta}&=\frac{1}{\eta}+\log\left(\frac{t+\frac{\lambda_2}{\lambda_1}\tau_1-\tau_1}{\lambda_2}\right)-\left(\frac{t+\frac{\lambda_2}{\lambda_1}\tau_1-\tau_1}{\lambda_2}\right)^{\eta}\log\left(\frac{t+\frac{\lambda_2}{\lambda_1}\tau_1-\tau_1}{\lambda_2}\right) \nonumber
\\
\Rightarrow \left(\frac{\partial \log\left(f_{x_2}(t+h|\lambda_1,\lambda_2,\eta)\right)}{\partial \eta}\right)^2&=\frac{1}{\eta^2}+\left(\log\left(\frac{t+\frac{\lambda_2}{\lambda_1}\tau_1-\tau_1}{\lambda_2}\right)\right)^2 \nonumber
\\
&+\left(\frac{t+\frac{\lambda_2}{\lambda_1}\tau_1-\tau_1}{\lambda_2}\right)^{2\eta}\left(\log\left(\frac{t+\frac{\lambda_2}{\lambda_1}\tau_1-\tau_1}{\lambda_2}\right)\right)^2+\frac{2}{\eta}\log\left(\frac{t+\frac{\lambda_2}{\lambda_1}\tau_1-\tau_1}{\lambda_2}\right) \nonumber
\\
&-\frac{2}{\eta}\left(\frac{t+\frac{\lambda_2}{\lambda_1}\tau_1-\tau_1}{\lambda_2}\right)^{\eta}\log\left(\frac{t+\frac{\lambda_2}{\lambda_1}\tau_1-\tau_1}{\lambda_2}\right)\nonumber
\\
&-2\left(\frac{t+\frac{\lambda_2}{\lambda_1}\tau_1-\tau_1}{\lambda_2}\right)^{\eta}\left(\log\left(\frac{t+\frac{\lambda_2}{\lambda_1}\tau_1-\tau_1}{\lambda_2}\right)\right)^2.
 \label{eq:dev_f2_eta}
\end{align}
Then
\begin{align*}
J_{33,\tau_1\tau_2}^{\beta}(\eta)&=\frac{1}{\eta^2}\int_{\tau_1}^{\tau_2}f_{x_2}(t+h|\lambda_1,\lambda_2,\eta)^{\beta+1}dt+\int_{\tau_1}^{\tau_2}\left(\log\left(\frac{t+\frac{\lambda_2}{\lambda_1}\tau_1-\tau_1}{\lambda_2}\right)\right)^2f_{x_2}(t+h|\lambda_1,\lambda_2,\eta)^{\beta+1}dt
\\
&+\int_{\tau_1}^{\tau_2}\left(\frac{t+\frac{\lambda_2}{\lambda_1}\tau_1-\tau_1}{\lambda_2}\right)^{2\eta}\left(\log\left(\frac{t+\frac{\lambda_2}{\lambda_1}\tau_1-\tau_1}{\lambda_2}\right)\right)^2f_{x_2}(t+h|\lambda_1,\lambda_2,\eta)^{\beta+1}dt
\\
&+\frac{2}{\eta}\int_{\tau_1}^{\tau_2}\log\left(\frac{t+\frac{\lambda_2}{\lambda_1}\tau_1-\tau_1}{\lambda_2}\right)f_{x_2}(t+h|\lambda_1,\lambda_2,\eta)^{\beta+1}dt
\\
&-\frac{2}{\eta}\int_{\tau_1}^{\tau_2}\left(\frac{t+\frac{\lambda_2}{\lambda_1}\tau_1-\tau_1}{\lambda_2}\right)^{\eta}\log\left(\frac{t+\frac{\lambda_2}{\lambda_1}\tau_1-\tau_1}{\lambda_2}\right)f_{x_2}(t+h|\lambda_1,\lambda_2,\eta)^{\beta+1}dt
\\
&-2\int_{\tau_1}^{\tau_2}\left(\frac{t+\frac{\lambda_2}{\lambda_1}\tau_1-\tau_1}{\lambda_2}\right)^{\eta}\left(\log\left(\frac{t+\frac{\lambda_2}{\lambda_1}\tau_1-\tau_1}{\lambda_2}\right)\right)^2f_{x_2}(t+h|\lambda_1,\lambda_2,\eta)^{\beta+1}dt
\\
&=B_1+B_2+B_3+B_4+B_5+B_6.
\end{align*}
Then, calculating each term
\begin{align*}
B_1&=\frac{1}{\eta^2}\int_{\tau_1}^{\tau_2}f_{x_2}(t+h|\lambda_1,\lambda_2,\eta)^{\beta+1}dt=\frac{1}{\eta^2}H_{0,0,\beta}^{\tau_1,\tau_2}(a_0,a_1,\eta)
\\
B_2&=\int_{\tau_1}^{\tau_2}\left(\log\left(\frac{t+\frac{\lambda_2}{\lambda_1}\tau_1-\tau_1}{\lambda_2}\right)\right)^2f_{x_2}(t+h|\lambda_1,\lambda_2,\eta)^{\beta+1}dt=H_{0,2,\beta}^{\tau_1,\tau_2}(a_0,a_1,\eta)
\\
B_3&=\int_{\tau_1}^{\tau_2}\left(\frac{t+\frac{\lambda_2}{\lambda_1}\tau_1-\tau_1}{\lambda_2}\right)^{2\eta}\left(\log\left(\frac{t+\frac{\lambda_2}{\lambda_1}\tau_1-\tau_1}{\lambda_2}\right)\right)^2f_{x_2}(t+h|\lambda_1,\lambda_2,\eta)^{\beta+1}dt=H_{2\eta,2,\beta}^{\tau_1,\tau_2}(a_0,a_1,\eta)
\\
B_4&=\frac{2}{\eta}\int_{\tau_1}^{\tau_2}\log\left(\frac{t+\frac{\lambda_2}{\lambda_1}\tau_1-\tau_1}{\lambda_2}\right)f_{x_2}(t+h|\lambda_1,\lambda_2,\eta)^{\beta+1}dt=\frac{2}{\eta}H_{0,1,\beta}^{\tau_1,\tau_2}(a_0,a_1,\eta)
\\
B_5&=-\frac{2}{\eta}\int_{\tau_1}^{\tau_2}\left(\frac{t+\frac{\lambda_2}{\lambda_1}\tau_1-\tau_1}{\lambda_2}\right)^{\eta}\log\left(\frac{t+\frac{\lambda_2}{\lambda_1}\tau_1-\tau_1}{\lambda_2}\right)f_{x_2}(t+h|\lambda_1,\lambda_2,\eta)^{\beta+1}dt=-\frac{2}{\eta}H_{\eta,1,\beta}^{\tau_1,\tau_2}(a_0,a_1,\eta)
\\
B_6&=-2\int_{\tau_1}^{\tau_2}\left(\frac{t+\frac{\lambda_2}{\lambda_1}\tau_1-\tau_1}{\lambda_2}\right)^{\eta}\left(\log\left(\frac{t+\frac{\lambda_2}{\lambda_1}\tau_1-\tau_1}{\lambda_2}\right)\right)^2f_{x_2}(t+h|\lambda_1,\lambda_2,\eta)^{\beta+1}dt=-2H_{\eta,2,\beta}^{\tau_1,\tau_2}(a_0,a_1,\eta).
\end{align*}
Then, we have
\begin{align*}
J_{33,\tau_1\tau_2}^{\beta}(\eta)&=\frac{1}{\eta^2}H_{0,0,\beta}^{\tau_1,\tau_2}(a_0,a_1,\eta)+H_{0,2,\beta}^{\tau_1,\tau_2}(a_0,a_1,\eta)+H_{2\eta,2,\beta}^{\tau_1,\tau_2}(a_0,a_1,\eta)+\frac{2}{\eta}H_{0,1,\beta}^{\tau_1,\tau_2}(a_0,a_1,\eta)
\\
&-\frac{2}{\eta}H_{\eta,1,\beta}^{\tau_1,\tau_2}(a_0,a_1,\eta)-2H_{\eta,2,\beta}^{\tau_1,\tau_2}(a_0,a_1,\eta).
\end{align*}
And for $J_{33,\tau_2}^{\beta}(\eta)$
\begin{align*}
J_{33,\tau_2}^{\beta}(\eta)=\left(\frac{\partial \log\left(1-F_{x_2}(\tau_2+h|\lambda_1,\lambda_2,\eta)\right)}{\partial \eta}\right)^2\left(1-F_{x_2}(\tau_2+h|\lambda_1,\lambda_2,\eta)\right)^{\beta+1}.
\end{align*}
Then
\begin{align}
1-F_{x_2}(\tau_2+h|\lambda_1,\lambda_2,\eta)&=\exp\left(-\left(\frac{\tau_2+\frac{\lambda_2}{\lambda_1}\tau_1-\tau_1}{\lambda_2}\right)^{\eta}\right) \nonumber
\\
\Rightarrow \log\left(1-F_{x_2}(\tau_2+h|\lambda_1,\lambda_2,\eta)\right)&=-\left(\frac{\tau_2+\frac{\lambda_2}{\lambda_1}\tau_1-\tau_1}{\lambda_2}\right)^{\eta} \nonumber
\\
\Rightarrow \frac{\partial \log\left(1-F_{x_2}(\tau_2+h|\lambda_1,\lambda_2,\eta)\right)}{\partial \eta}&=-\left(\frac{\tau_2+\frac{\lambda_2}{\lambda_1}\tau_1-\tau_1}{\lambda_2}\right)^{\eta}\log\left(\frac{\tau_2+\frac{\lambda_2}{\lambda_1}\tau_1-\tau_1}{\lambda_2}\right) \nonumber
\\
\Rightarrow \left( \frac{\partial \log\left(1-F_{x_2}(\tau_2+h|\lambda_1,\lambda_2,\eta)\right)}{\partial \eta}\right)^2&=\left(\frac{\tau_2+\frac{\lambda_2}{\lambda_1}\tau_1-\tau_1}{\lambda_2}\right)^{2\eta}\left(\log\left(\frac{\tau_2+\frac{\lambda_2}{\lambda_1}\tau_1-\tau_1}{\lambda_2}\right)\right)^2 \nonumber
\\
\Rightarrow J_{33,\tau_2}^{\beta}(\eta)=\left(\frac{\tau_2+\frac{\lambda_2}{\lambda_1}\tau_1-\tau_1}{\lambda_2}\right)^{2\eta}&\left(\log\left(\frac{\tau_2+\frac{\lambda_2}{\lambda_1}\tau_1-\tau_1}{\lambda_2}\right)\right)^2\exp\left(-\left(\frac{\tau_2+\frac{\lambda_2}{\lambda_1}\tau_1-\tau_1}{\lambda_2}\right)^\eta(\beta+1)\right) .
\label{eq:dev_F2_eta}
\end{align}

\subsection{Proof Proposition 8}
For $J_{12}^{\beta}(a_0,a_1)$
\begin{align*}
J_{12}^{\beta}(a_0,a_1)&=\int_{0}^{\tau_1}\frac{\partial \log\left(f_{x_1}(t|\lambda_1,\eta)\right)}{\partial a_0}\frac{\partial \log\left(f_{x_1}(t|\lambda_1,\eta)\right)}{\partial a_1}f_{x_1}(t|\lambda_1,\eta)^{\beta+1}dt
\\
&+\int_{\tau_1}^{\tau_2}\frac{\partial \log\left(f_{x_2}(t+h|\lambda_1,\lambda_2,\eta)\right)}{\partial a_0}\frac{\partial \log\left(f_{x_2}(t+h|\lambda_1,\lambda_2,\eta)\right)}{\partial a_1}f_{x_2}(t+h|\lambda_1,\lambda_2,\eta)^{\beta+1}dt
\\
&+\frac{\partial \log\left(1-F_{x_2}(\tau_2+h|\lambda_1,\lambda_2,\eta)\right)}{\partial a_0}\frac{\partial \log\left(1-F_{x_2}(\tau_2+h|\lambda_1,\lambda_2,\eta)\right)}{\partial a_1}\left(1-F_{x_2}(\tau_2+h|\lambda_1,\lambda_2,\eta)\right)^{\beta+1}
\\
&=J_{12,\tau_1}^{\beta}(a_0,a_1)+J_{12,\tau_1\tau_2}^{\beta}(a_0,a_1)+J_{12,\tau_2}^{\beta}(a_0,a_1).
\end{align*}
For  $J_{12,\tau_1}^{\beta}(a_0,a_1)$, using the derivatives calculated in \ref{eq:dev_f1_a0} and \ref{eq:dev_f1_a1}.
\begin{align*}
\frac{\partial \log\left(f_{x_1}(t|\lambda_1,\eta)\right)}{\partial a_0}=\eta\left(-1+\left(\frac{t}{\lambda_1}\right)^{\eta}\right); \quad \frac{\partial \log\left(f_{x_1}(t|\lambda_1,\eta)\right)}{\partial a_1}=\eta x_1\left(-1+\left(\frac{t}{\lambda_1}\right)^{\eta}\right).
\end{align*}
Then
\begin{align*}
\frac{\partial \log\left(f_{x_1}(t|\lambda_1,\eta)\right)}{\partial a_0}\frac{\partial \log\left(f_{x_1}(t|\lambda_1,\eta)\right)}{\partial a_1}=\eta^2 x_1 \left(-1+\left(\frac{t}{\lambda_1}\right)^{2\eta}-2\left(\frac{t}{\lambda_1}\right)^{\eta}\right).
\end{align*}
So
\begin{align*}
J_{12,\tau_1}^{\beta}(a_0,a_1)&=\eta^2x_1\left\{\int_{0}^{\tau_1}f_{x_1}(t|\lambda_1,\eta)^{\beta+1}dt+\int_{0}^{\tau_1}\left(\frac{t}{\lambda_1}\right)^2f_{x_1}(t|\lambda_1,\eta)^{\beta+1}dt \right.
\\
&\left.-2\int_{0}^{\tau_1}\left(\frac{t}{\lambda_1}\right)^{\eta}f_{x_1}(t|\lambda_1,\eta)^{\beta+1}dt\right\}
\\&=\eta^2x_1\left\{\zeta_{0,\beta}^{\tau_1}(a_0,a_1,\eta)+\zeta_{2\eta,\beta}^{\tau_1}(a_0,a_1,\eta)-2\zeta_{\eta,\beta}^{\tau_1}(a_0,a_1,\eta)\right\}.
\end{align*}
For $J_{12,\tau_1\tau_2}^{\beta}(a_0,a_1)$
\begin{align*}
J_{12,\tau_1\tau_2}^{\beta}(a_0,a_1)=\int_{\tau_1}^{\tau_2}\frac{\partial \log\left(f_{x_2}(t+h|\lambda_1,\lambda_2,\eta)\right)}{\partial a_0}\frac{\partial \log\left(f_{x_2}(t+h|\lambda_1,\lambda_2,\eta)\right)}{\partial a_1}f_{x_2}(t+h|\lambda_1,\lambda_2,\eta)^{\beta+1}dt.
\end{align*}
For the derivatives calculated in \eqref{eq:dev_f2_a0} and \eqref{eq:dev_f2_a1}.
\begin{align*}
\frac{\partial \log\left(f_{x_2}(t+h|\lambda_1,\lambda_2,\eta)\right)}{\partial a_0}&=\eta\left(-1+\left(\frac{t+\frac{\lambda_2}{\lambda_1}\tau_1-\tau_1}{\lambda_2}\right)^{\eta}\right)
\\
\frac{\partial \log\left(f_{x_2}(t+h|\lambda_1,\lambda_2,\eta)\right)}{\partial a_1}&=-\eta x_2 + \frac{\eta-1}{\lambda_1}\tau_1(x_2-x_1)\frac{1}{\frac{t+\frac{\lambda_2}{\lambda_1}\tau_1-\tau_1}{\lambda_2}}
\\
&-\frac{\eta\tau_1(x_2-x_1)}{\lambda_1}\left(\frac{t+\frac{\lambda_2}{\lambda_1}\tau_1-\tau_1}{\lambda_2}\right)^{\eta-1}
\\
&+\eta x_2 \left(\frac{t+\frac{\lambda_2}{\lambda_1}\tau_1-\tau_1}{\lambda_2}\right)^{\eta}
\\
\Rightarrow \frac{\partial \log\left(f_{x_2}(t+h|\lambda_1,\lambda_2,\eta)\right)}{\partial a_0}\frac{\partial \log\left(f_{x_2}(t+h|\lambda_1,\lambda_2,\eta)\right)}{\partial a_1}&=\eta^2x_2-\eta(\eta-1)\frac{\frac{\tau_1}{\lambda_1}(x_2-x_1)}{\left(\frac{t+\frac{\lambda_2}{\lambda_1}\tau_1-\tau_1}{\lambda_2}\right)}
\\
&+\eta^2\frac{\tau_1}{\lambda_1}(x_2-x_1)\left(\frac{t+\frac{\lambda_2}{\lambda_1}\tau_1-\tau_1}{\lambda_2}\right)^{\eta-1}
\\
&-\eta^2\left(\frac{t+\frac{\lambda_2}{\lambda_1}\tau_1-\tau_1}{\lambda_2}\right)^{\eta}
\\
&-\eta^2x_2\left(\frac{t+\frac{\lambda_2}{\lambda_1}\tau_1-\tau_1}{\lambda_2}\right)^{\eta}
\\
&+\eta(\eta-1)\frac{\tau_1}{\lambda_1}(x_2-x_1)\left(\frac{t+\frac{\lambda_2}{\lambda_1}\tau_1-\tau_1}{\lambda_2}\right)^{\eta-1}
\\
&-\eta^2\frac{\tau_1}{\lambda_1}(x_2-x_1)\left(\frac{t+\frac{\lambda_2}{\lambda_1}\tau_1-\tau_1}{\lambda_2}\right)^{2\eta-1}
\\
&+\eta^2x_2\left(\frac{t+\frac{\lambda_2}{\lambda_1}\tau_1-\tau_1}{\lambda_2}\right)^{2\eta}.
\end{align*}
So
\begin{align*}
J_{12,\tau_1\tau_2}^{\beta}(a_0,a_1)&=\eta^2x_2\int_{\tau_1}^{\tau_2}f_{x_2}(t+h|\lambda_1,\lambda_2,\eta)^{\beta+1}dt
\\
&-\eta(\eta-1)\frac{\tau_1}{\lambda_1}(x_2-x_1)\int_{\tau_1}^{\tau_2}\left(\frac{t+\frac{\lambda_2}{\lambda_1}\tau_1-\tau_1}{\lambda_2}\right)^{-1}f_{x_2}(t+h|\lambda_1,\lambda_2,\eta)^{\beta+1}dt
\\
&+\eta(2\eta-1)\frac{\tau_1}{\lambda_1}(x_2-x_1)\int_{\tau_1}^{\tau_2}\left(\frac{t+\frac{\lambda_2}{\lambda_1}\tau_1-\tau_1}{\lambda_2}\right)^{\eta-1}f_{x_2}(t+h|\lambda_1,\lambda_2,\eta)^{\beta+1}dt
\\
&-2\eta^2x_2\int_{\tau_1}^{\tau_2}\left(\frac{t+\frac{\lambda_2}{\lambda_1}\tau_1-\tau_1}{\lambda_2}\right)^{\eta}f_{x_2}(t+h|\lambda_1,\lambda_2,\eta)^{\beta+1}dt
\\
&-\eta^2\frac{\tau_1}{\lambda_1}(x_2-x_1)\int_{\tau_1}^{\tau_2}\left(\frac{t+\frac{\lambda_2}{\lambda_1}\tau_1-\tau_1}{\lambda_2}\right)^{2\eta-1}f_{x_2}(t+h|\lambda_1,\lambda_2,\eta)^{\beta+1}dt
\\
&+\eta^2x_2\int_{\tau_1}^{\tau_2}\left(\frac{t+\frac{\lambda_2}{\lambda_1}\tau_1-\tau_1}{\lambda_2}\right)^{2\eta}f_{x_2}(t+h|\lambda_1,\lambda_2,\eta)^{\beta+1}dt
\\
&=\eta^2x_2\zeta_{0,\beta}^{\tau_1,\tau_2}(a_0,a_1,\eta)
-\eta(\eta-1)\frac{\tau_1}{\lambda_1}(x_2-x_1)\zeta_{-1,\beta}^{\tau_1,\tau_2}(a_0,a_1,\eta)
\\
&+\eta(2\eta-1)\frac{\tau_1}{\lambda_1}(x_2-x_1)\zeta_{\eta-1,\beta}^{\tau_1,\tau_2}(a_0,a_1,\eta)
-2\eta^2x_2\zeta_{\eta,\beta}^{\tau_1,\tau_2}(a_0,a_1,\eta)
\\
&-\eta^2\frac{\tau_1}{\lambda_1}(x_2-x_1)\zeta_{2\eta-1,\beta}^{\tau_1,\tau_2}(a_0,a_1,\eta)
+\eta^2x_2\zeta_{2\eta,\beta}^{\tau_1,\tau_2}(a_0,a_1,\eta).
\end{align*}
And for $J_{12,\tau_2}^{\beta}(a_0,a_1)$
\begin{align*}
J_{12,\tau_2}^{\beta}(a_0,a_1)=\frac{\partial \log\left(1-F_{x_2}(\tau_2+h|\lambda_1,\lambda_2,\eta)\right)}{\partial a_0}\frac{\partial \log\left(1-F_{x_2}(\tau_2+h|\lambda_1,\lambda_2,\eta)\right)}{\partial a_1}\left(1-F_{x_2}(\tau_2+h|\lambda_1,\lambda_2,\eta)\right)^{\beta+1}.
\end{align*}
Using the derivatives calculated in \eqref{eq:dev_F2_a0} and \eqref{eq:dev_F2_a1}.
\begin{align*}
\frac{\partial \log\left(1-F_{x_2}(\tau_2+h|\lambda_1,\lambda_2,\eta)\right)}{\partial a_0}&=\eta\left(\frac{\tau_2+\frac{\lambda_2}{\lambda_1}\tau_1-\tau_1}{\lambda_2}\right)^{\eta}
\\
\frac{\partial \log\left(1-F_{x_2}(\tau_2+h|\lambda_1,\lambda_2,\eta)\right)}{\partial a_1}&=\eta\left(\frac{\tau_2+\frac{\lambda_2}{\lambda_1}\tau_1-\tau_1}{\lambda_2}\right)^{\eta-1}\frac{\frac{x_1\tau_1}{\lambda_1}+x_2(\tau_2-\tau_1)}{\lambda_2}
\\
\Rightarrow  J_{12,\tau_2}^{\beta}(a_0,a_1)&=\eta^2\left(\frac{\tau_2+\frac{\lambda_2}{\lambda_1}\tau_1-\tau_1}{\lambda_2}\right)^{2\eta-1}\frac{\frac{x_1\tau_1}{\lambda_1}+x_2(\tau_2-\tau_1)}{\lambda_2}
\\
&\cdot\exp\left(-\left(\frac{\tau_2+\frac{\lambda_2}{\lambda_1}\tau_1-\tau_1}{\lambda_2}\right)^{\eta}(\beta+1)\right).
\end{align*}

\subsection{Proof Propostion 9}

For $J_{13}^{\beta}(a_0,\eta)$
\begin{align*}
J_{13}^{\beta}(a_0,\eta)&=\int_{0}^{\tau_1}\left(\frac{\partial \log(f_{x_1}(t|\lambda_1,\eta))}{\partial a_0}\right)\left(\frac{\partial \log(f_{x_1}(t|\lambda_1,\eta))}{\partial \eta}\right)f_{x_1}(t|\lambda_1,\eta)^{\beta+1}dt
\\
+&\int_{\tau_1}^{\tau_2}\left(\frac{\partial \log(f_{x_2}(t+h|\lambda_1,\lambda_2,\eta))}{\partial a_0}\right)\left(\frac{\partial \log(f_{x_2}(t+h|\lambda_1,\lambda_2,\eta))}{\partial \eta}\right)f_{x_2}(t+h|\lambda_1,\lambda_2,\eta)^{\beta+1}dt
\\
+&\frac{\partial \log\left(1-F_{x_2}(\tau_2+h|\lambda_1,\lambda_2,\eta)\right)}{\partial a_0}\frac{\partial \log\left(1-F_{x_2}(\tau_2+h|\lambda_1,\lambda_2,\eta)\right)}{\partial \eta}\left(1-F_{x_2}(\tau_2+h|\lambda_1,\lambda_2,\eta)\right)^{\beta+1}
\\
&=J_{13,\tau_1}^{\beta}(a_0,\eta)+J_{13,\tau_1\tau_2}^{\beta}(a_0,\eta)+J_{13,\tau_2}^{\beta}(a_0,\eta).
\end{align*}
In relation to  $J_{13,\tau_1}^{\beta}(a_0,\eta)$ and using the derivatives calculated in \eqref{eq:dev_f1_a0} and \eqref{eq:dev_f1_eta}.
\begin{align*}
\frac{\partial \log(f_{x_1}(t|\lambda_1,\eta))}{\partial a_0}&=\eta\left(-1+\left(\frac{t}{\lambda_1}\right)^{\eta}\right);
 \\
 \frac{\partial \log(f_{x_1}(t|\lambda_1,\eta))}{\partial \eta}&=\frac{1}{\eta}+\log\left(\frac{t}{\lambda_1}\right)-\left(\frac{t}{\lambda_1}\right)^{\eta}\log\left(\frac{t}{\lambda_1}\right)
 \\
 \Rightarrow \frac{\partial \log(f_{x_1}(t|\lambda_1,\eta))}{\partial a_0} \frac{\partial \log(f_{x_1}(t|\lambda_1,\eta))}{\partial \eta}&=-1-\eta\log\left(\frac{t}{\lambda_1}\right)+\eta\left(\frac{t}{\lambda_1}\right)^{\eta}\log\left(\frac{t}{\lambda_1}\right)
  \\
  &+\left(\frac{t}{\lambda_1}\right)^{\eta}+\eta\left(\frac{t}{\lambda_1}\right)^{\eta}\log\left(\frac{t}{\lambda_1}\right)-\eta\left(\frac{t}{\lambda_1}\right)^{2\eta}\log\left(\frac{t}{\lambda_1}\right).
\end{align*}
Then, we have:
\begin{align*}
J_{13,\tau_1}^{\beta}(a_0,\eta)&=-\int_{0}^{\tau_1}     f_{x_1}(t|\lambda_1,\eta)^{\beta+1}dt
-\eta\int_{0}^{\tau_1}\log\left(\frac{t}{\lambda_1}\right)     f_{x_1}(t|\lambda_1,\eta)^{\beta+1}dt
\\
&+ \eta\int_{0}^{\tau_1}\left(\frac{t}{\lambda_1}\right)^{\eta}     \log\left(\frac{t}{\lambda_1}\right)  f_{x_1}(t|\lambda_1,\eta)^{\beta+1}dt
+\int_{0}^{\tau_1}\left(\frac{t}{\lambda_1}\right)^{\eta}      f_{x_1}(t|\lambda_1,\eta)^{\beta+1}dt
\\
&+\eta \int_{0}^{\tau_1} \left(\frac{t}{\lambda_1}\right)^{\eta}    \log\left(\frac{t}{\lambda_1}\right)   f_{x_1}(t|\lambda_1,\eta)^{\beta+1}dt
-\eta \int_{0}^{\tau_1} \left(\frac{t}{\lambda_1}\right)^{2\eta}     \log\left(\frac{t}{\lambda_1}\right)  f_{x_1}(t|\lambda_1,\eta)^{\beta+1}dt
\\
&=-\zeta_{0,\beta}^{\tau_1}(a_0,a_1,\eta)
-\eta H_{0,1,\beta}^{\tau_1}(a_0,a_1,\eta)
+2\eta H_{ \eta,1,\beta}^{\tau_1}(a_0,a_1,\eta)
\\
&+\zeta_{\eta,\beta}^{\tau_1}(a_0,a_1,\eta)
-\eta H_{2\eta,1,\beta}^{\tau_1}(a_0,a_1,\eta).
\end{align*}
In relation to $J_{13,\tau_1\tau_2}^{\beta}(a_0,\eta)$ and using the derivatives calculated in  \eqref{eq:dev_f2_a0} and \eqref{eq:dev_f2_eta}.
\begin{align*}
\frac{\partial \log\left(f_{x_2}(t+h|\lambda_1,\lambda_2,\eta\right)}{\partial a_0}&=\eta\left(-1+\left(\frac{t+\frac{\lambda_2}{\lambda_1}\tau_1-\tau_1}{\lambda_2}\right)^{\eta}\right);
\\
\frac{\partial \log\left(f_{x_2}(t+h|\lambda_1,\lambda_2,\eta\right)}{\partial \eta}&=\frac{1}{\eta}+\log\left(\frac{t+\frac{\lambda_2}{\lambda_1}\tau_1-\tau_1}{\lambda_2}\right)-\left(\frac{t+\frac{\lambda_2}{\lambda_1}\tau_1-\tau_1}{\lambda_2}\right)^{\eta}\log\left(\frac{t+\frac{\lambda_2}{\lambda_1}\tau_1-\tau_1}{\lambda_2}\right)
\end{align*}
\begin{align*}
\Rightarrow  \frac{\partial \log\left(f_{x_2}(t+h|\lambda_1,\lambda_2,\eta\right)}{\partial a_0}\frac{\partial \log\left(f_{x_2}(t+h|\lambda_1,\lambda_2,\eta\right)}{\partial \eta}&=-1-\eta\log\left(\frac{t+\frac{\lambda_2}{\lambda_1}\tau_1-\tau_1}{\lambda_2}\right)
\\
&+\eta\left(\frac{t+\frac{\lambda_2}{\lambda_1}\tau_1-\tau_1}{\lambda_2}\right)^{\eta}\log\left(\frac{t+\frac{\lambda_2}{\lambda_1}\tau_1-\tau_1}{\lambda_2}\right)
\\
&+\left(\frac{t+\frac{\lambda_2}{\lambda_1}\tau_1-\tau_1}{\lambda_2}\right)^{\eta}
\\
&+\eta\left(\frac{t+\frac{\lambda_2}{\lambda_1}\tau_1-\tau_1}{\lambda_2}\right)^{\eta}\log\left(\frac{t+\frac{\lambda_2}{\lambda_1}\tau_1-\tau_1}{\lambda_2}\right)
\\
&-\eta\left(\frac{t+\frac{\lambda_2}{\lambda_1}\tau_1-\tau_1}{\lambda_2}\right)^{2\eta}\log\left(\frac{t+\frac{\lambda_2}{\lambda_1}\tau_1-\tau_1}{\lambda_2}\right).
\end{align*}
Then
\begin{align*}
J_{13,\tau_1\tau_2}^{\beta}(a_0,\eta)&=-\int_{\tau_1}^{\tau_2}f_{x_2}(t+h|\lambda_1,\lambda_2,\eta)^{\beta+1}dt
-\eta \int_{\tau_1}^{\tau_2}\log\left(\frac{t+\frac{\lambda_2}{\lambda_1}\tau_1-\tau_1}{\lambda_2}\right)f_{x_2}(t+h|\lambda_1,\lambda_2,\eta)^{\beta+1}dt
\\
&+\eta \int_{\tau_1}^{\tau_2}\left(\frac{t+\frac{\lambda_2}{\lambda_1}\tau_1-\tau_1}{\lambda_2}\right)^{\eta}\log\left(\frac{t+\frac{\lambda_2}{\lambda_1}\tau_1-\tau_1}{\lambda_2}\right)f_{x_2}(t+h|\lambda_1,\lambda_2,\eta)^{\beta+1}dt
\\
&+\int_{\tau_1}^{\tau_2}\left(\frac{t+\frac{\lambda_2}{\lambda_1}\tau_1-\tau_1}{\lambda_2}\right)^{\eta}f_{x_2}(t+h|\lambda_1,\lambda_2,\eta)^{\beta+1}dt
\\
&+\eta\int_{\tau_1}^{\tau_2}\left(\frac{t+\frac{\lambda_2}{\lambda_1}\tau_1-\tau_1}{\lambda_2}\right)^{\eta}\log\left(\frac{t+\frac{\lambda_2}{\lambda_1}\tau_1-\tau_1}{\lambda_2}\right)f_{x_2}(t+h|\lambda_1,\lambda_2,\eta)^{\beta+1}dt
\\
&-\eta\int_{\tau_1}^{\tau_2}\left(\frac{t+\frac{\lambda_2}{\lambda_1}\tau_1-\tau_1}{\lambda_2}\right)^{2\eta}\log\left(\frac{t+\frac{\lambda_2}{\lambda_1}\tau_1-\tau_1}{\lambda_2}\right)f_{x_2}(t+h|\lambda_1,\lambda_2,\eta)^{\beta+1}dt
\\
&=-\zeta_{0,\beta}^{\tau_1,\tau_2}(a_0,a_1,\eta)-\eta H_{0,1,\beta}^{\tau_1,\tau_2}(a_0,a_1,\eta)+2\eta H_{\eta,1,\beta}^{\tau_1,\tau_2}(a_0,a_1,\beta)
\\
&+\zeta_{\eta,\beta}^{\tau_1,\tau_2}(a_0,a_1,\eta)-\eta H_{2\eta,1,\beta}^{\tau_1,\tau_2}(a_0,a_1,\eta).
\end{align*}
And for $J_{13,\tau_2}^{\beta}(a_0,a_1,\eta)$
\begin{align*}
J_{13,\tau_2}^{\beta}(a_0,\eta)=\frac{\partial \log\left(1-F_{x_2}(\tau_2+h|\lambda_1,\lambda_2,\eta)\right)}{\partial a_0}\frac{\partial \log\left(1-F_{x_2}(\tau_2+h|\lambda_1,\lambda_2,\eta)\right)}{\partial \eta}\left(1-F_{x_2}(\tau_2+h|\lambda_1,\lambda_2,\eta)\right)^{\beta+1}.
\end{align*}
Using the derivatives calculated in  \ref{eq:dev_F2_a0} and \ref{eq:dev_F2_eta}.
\begin{align*}
\frac{\partial \log\left(1-F_{x_2}(\tau_2+h|\lambda_1,\lambda_2,\eta)\right)}{\partial a_0}&=\eta\left(\frac{\tau_2+\frac{\lambda_2}{\lambda_1}\tau_1-\tau_1}{\lambda_2}\right)^{\eta}
\\
\frac{\partial \log\left(1-F_{x_2}(\tau_2+h|\lambda_1,\lambda_2,\eta)\right)}{\partial \eta}&=-\log\left(\frac{\tau_2+\frac{\lambda_2}{\lambda_1}\tau_1-\tau_1}{\lambda_2}\right)\left(\frac{\tau_2+\frac{\lambda_2}{\lambda_1}\tau_1-\tau_1}{\lambda_2}\right)^{\eta}
\end{align*}
\begin{align*}
\Rightarrow \frac{\partial \log\left(1-F_{x_2}(\tau_2+h|\lambda_1,\lambda_2,\eta)\right)}{\partial a_0}\frac{\partial \log\left(1-F_{x_2}(\tau_2+h|\lambda_1,\lambda_2,\eta)\right)}{\partial \eta}=& -\eta\log\left(\frac{\tau_2+\frac{\lambda_2}{\lambda_1}\tau_1-\tau_1}{\lambda_2}\right)
\\
&\cdot \left(\frac{\tau_2+\frac{\lambda_2}{\lambda_1}\tau_1-\tau_1}{\lambda_2}\right)^{2\eta}
\end{align*}
\begin{align*}
\Rightarrow  J_{13,\tau_2}^{\beta}(a_0,a_1,\eta)&=\frac{\partial \log\left(1-F_{x_2}(\tau_2+h|\lambda_1,\lambda_2,\eta)\right)}{\partial a_0}\frac{\partial \log\left(1-F_{x_2}(\tau_2+h|\lambda_1,\lambda_2,\eta)\right)}{\partial \eta}
\\
&\cdot \left(1-F_{x_2}(\tau_2+h|\lambda_1,\lambda_2,\eta)\right)^{\beta+1}
\\
&=-\eta\log\left(\frac{\tau_2+\frac{\lambda_2}{\lambda_1}\tau_1-\tau_1}{\lambda_2}\right)\left(\frac{\tau_2+\frac{\lambda_2}{\lambda_1}\tau_1-\tau_1}{\lambda_2}\right)^{2\eta}
\\
&\cdot \exp\left(-\left(\frac{\tau_2+\frac{\lambda_2}{\lambda_1}\tau_1-\tau_1}{\lambda_2}\right)^\eta(\beta+1)\right).
\end{align*}

\subsection{Proof Proposition 10}
\begin{proof}
For $J_{23}^{\beta}(a_1,\eta)$
\begin{align*}
J_{23}^{\beta}(a_1,\eta)&=\int_{0}^{\tau_1}\frac{\partial \log\left(f_{x_1}(t|\lambda_1,\eta)\right)}{\partial a_1}\frac{\partial \log\left(f_{x_1}(t|\lambda_1,\eta)\right)}{\partial \eta}f_{x_1}(t|\lambda_1,\eta)^{\beta+1}dt
\\
&+\int_{\tau_1}^{\tau_2}\frac{\partial \log\left(f_{x_2}(t+h|\lambda_1,\lambda_2,\eta)\right)}{\partial a_1}\frac{\partial \log\left(f_{x_2}(t+h|\lambda_1,\lambda_2,\eta)\right)}{\partial \eta}f_{x_2}(t+h|\lambda_1,\lambda_2,\eta)^{\beta+1}dt
\\
&+\frac{\partial \log\left(1-F_{x_2}(\tau_2+h|\lambda_1,\lambda_2,\eta)\right)}{\partial a_1}\frac{\partial \log\left(1-F_{x_2}(\tau_2+h|\lambda_1,\lambda_2,\eta)\right)}{\partial \eta}\left(1-F_{x_2}(\tau_2+h|\lambda_1,\lambda_2,\eta)\right)^{\beta+1}
\\
&=J_{23,\tau_1}^{\beta}(a_1,\eta)+J_{23,\tau_1\tau_2}^{\beta}(a_1,\eta)+J_{23,\tau_2}^{\beta}(a_1,\eta).
\end{align*}
In relation to  $J_{23,\tau_1}^{\beta}(a_0,\eta)$ and using the derivatives calculated in  \eqref{eq:dev_f1_a1} and \eqref{eq:dev_f1_eta}.
\begin{align*}
\frac{\partial \log(f_{x_1}(t|\lambda_1,\eta))}{\partial a_1}&=\eta x_1 \left(-1+\left(\frac{t}{\lambda_1}\right)^{\eta}\right)=x_2\frac{\partial \log(f_{x_1}(t|\lambda_1,\eta))}{\partial a_0}
 \\
 \frac{\partial \log(f_{x_1}(t|\lambda_1,\eta))}{\partial \eta}&=\frac{1}{\eta}+\log\left(\frac{t}{\lambda_1}\right)-\left(\frac{t}{\lambda_1}\right)^{\eta}\log\left(\frac{t}{\lambda_1}\right)
 \\
 \Rightarrow \frac{\partial \log(f_{x_1}(t|\lambda_1,\eta))}{\partial a_1} \frac{\partial \log(f_{x_1}(t|\lambda_1,\eta))}{\partial \eta}&=x_1\frac{\partial \log(f_{x_1}(t|\lambda_1,\eta))}{\partial a_0} \frac{\partial \log(f_{x_1}(t|\lambda_1,\eta))}{\partial \eta}.
\end{align*}
Then, we have:
\begin{align*}
J_{23,\tau_1}^{\beta}(a_0,\eta)&=x_1 J_{13,\tau_1}^{\beta}(a_0,\eta)
\\
&=x_1\left\{ -\zeta_{0,\beta}^{\tau_1}(a_0,a_1,\eta)
-\eta H_{0,1,\beta}^{\tau_1}(a_0,a_1,\eta)
+2\eta H_{ \eta,1,\beta}^{\tau_1}(a_0,a_1,\eta) \right.
\\
&\left.+\zeta_{\eta,\beta}^{\tau_1}(a_0,a_1,\eta)
-\eta H_{2\eta,1,\beta}^{\tau_1}(a_0,a_1,\eta)\right\}.
\end{align*}
In relation to $J_{23,\tau_1\tau_2}^{\beta}(a_0,\eta)$ and using the derivatives calculated in  \eqref{eq:dev_f2_a1} and \eqref{eq:dev_f2_eta}.
\begin{align*}
\frac{\partial \log\left(f_{x_2}(t+h|\lambda_1,\lambda_2,\eta\right)}{\partial a_1}&=-\eta x_2 + \frac{\eta-1}{\lambda_1}\tau_1(x_2-x_1)\left( \frac{t+\frac{\lambda_2}{\lambda_1}\tau_1-\tau_1}{\lambda_2}\right)^{-1}-\frac{\eta\tau_1(x_2-x_1)}{\lambda_1}\left(\frac{t+\frac{\lambda_2}{\lambda_1}\tau_1-\tau_1}{\lambda_2}\right)^{\eta-1}
\\
&+\eta x_2 \left(\frac{t+\frac{\lambda_2}{\lambda_1}\tau_1-\tau_1}{\lambda_2}\right)^{\eta}
\\
\frac{\partial \log\left(f_{x_2}(t+h|\lambda_1,\lambda_2,\eta\right)}{\partial \eta}&=\frac{1}{\eta}+\log\left(\frac{t+\frac{\lambda_2}{\lambda_1}\tau_1-\tau_1}{\lambda_2}\right)-\left(\frac{t+\frac{\lambda_2}{\lambda_1}\tau_1-\tau_1}{\lambda_2}\right)^{\eta}\log\left(\frac{t+\frac{\lambda_2}{\lambda_1}\tau_1-\tau_1}{\lambda_2}\right)
\end{align*}
\begin{align*}
\Rightarrow  \frac{\partial \log\left(f_{x_2}(t+h|\lambda_1,\lambda_2,\eta\right)}{\partial a_1}\frac{\partial \log\left(f_{x_2}(t+h|\lambda_1,\lambda_2,\eta\right)}{\partial \eta}&=-x_2-\eta x_2 \log\left(\frac{t+\frac{\lambda_2}{\lambda_1}\tau_1-\tau_1}{\lambda_2}\right)
\\
&+\eta x_2 \left(\frac{t+\frac{\lambda_2}{\lambda_1}\tau_1 -\tau_1}{\lambda_2}\right)^{\eta}\log\left(\frac{t+\frac{\lambda_2}{\lambda_1}\tau_1-\tau_1}{\lambda_2}\right)
\\
&+\frac{\eta-1}{\eta}\frac{\tau_1}{\lambda_1}(x_2-x_1)\frac{1}{\frac{t+\frac{\lambda_2}{\lambda_1}\tau_1 -\tau_1}{\lambda_2}}
\\
&+(\eta-1)\frac{\tau_1}{\lambda_1}(x_2-x_1)\frac{1}{\frac{t+\frac{\lambda_2}{\lambda_1}\tau_1 -\tau_1}{\lambda_2}}\log\left(\frac{t+\frac{\lambda_2}{\lambda_1}\tau_1 -\tau_1}{\lambda_2}\right)
\\
&-(\eta-1)\frac{\tau_1}{\lambda_1}(x_2-x_1)\left(\frac{t+\frac{\lambda_2}{\lambda_1}\tau_1-\tau_1}{\lambda_2}\right)^{\eta-1}\log\left(\frac{t+\frac{\lambda_2}{\lambda_1}\tau_1-\tau_1}{\lambda_2}\right)
\\
&-\frac{\tau_1}{\lambda_1}(x_2-x_1)\left(\frac{t+\frac{\lambda_2}{\lambda_1}\tau_1-\tau_1}{\lambda_2}\right)^{\eta-1}
\\
&-\eta\frac{\tau_1}{\lambda_1}(x_2-x_1)\left(\frac{t+\frac{\lambda_2}{\lambda_1}\tau_1-\tau_1}{\lambda_2}\right)^{\eta-1}\log\left(\frac{t+\frac{\lambda_2}{\lambda_1}\tau_1-\tau_1}{\lambda_2}\right)
\\
&+\eta\frac{\tau_1}{\lambda_1}(x_2-x_1)\left(\frac{t+\frac{\lambda_2}{\lambda_1}\tau_1-\tau_1}{\lambda_2}\right)^{2\eta-1}\log\left(\frac{t+\frac{\lambda_2}{\lambda_1}\tau_1-\tau_1}{\lambda_2}\right)
\\
&+x_2\left(\frac{t+\frac{\lambda_2}{\lambda_1}\tau_1-\tau_1}{\lambda_2}\right)^{\eta}
\\
&+\eta x_2\left(\frac{t+\frac{\lambda_2}{\lambda_1}\tau_1-\tau_1}{\lambda_2}\right)^{\eta}\log\left(\frac{t+\frac{\lambda_2}{\lambda_1}\tau_1-\tau_1}{\lambda_2}\right)
\\
&-\eta x_2\left(\frac{t+\frac{\lambda_2}{\lambda_1}\tau_1-\tau_1}{\lambda_2}\right)^{2\eta}\log\left(\frac{t+\frac{\lambda_2}{\lambda_1}\tau_1-\tau_1}{\lambda_2}\right).
\end{align*}
Then
\begin{align*}
J_{23,\tau_1\tau_2}^{\beta}(a_1,\eta)&=-x_2\int_{\tau_1}^{\tau_2}f_{x_2}(t+h|\lambda_1,\lambda_2,\eta)^{\beta+1}dt
-\eta x_2 \int_{\tau_1}^{\tau_2}\log\left(\frac{t+\frac{\lambda_2}{\lambda_1}\tau_1-\tau_1}{\lambda_2}\right)f_{x_2}(t+h|\lambda_1,\lambda_2,\eta)^{\beta+1}dt
\\
& +\eta x_2 \int_{\tau_1}^{\tau_2}\left(\frac{t+\frac{\lambda_2}{\lambda_1}\tau_1-\tau_1}{\lambda_2}\right)^{\eta}\log\left(\frac{t+\frac{\lambda_2}{\lambda_1}\tau_1-\tau_1}{\lambda_2}\right)f_{x_2}(t+h|\lambda_1,\lambda_2,\eta)^{\beta+1}dt
\\
& + \frac{\eta-1}{\eta}\frac{\tau_1}{\lambda_1}(x_2-x_1) \int_{\tau_1}^{\tau_2}\left(\frac{t+\frac{\lambda_2}{\lambda_1}\tau_1-\tau_1}{\lambda_2}\right)^{-1}f_{x_2}(t+h|\lambda_1,\lambda_2,\eta)^{\beta+1}dt
\\
&+(\eta-1)\frac{\tau_1}{\lambda_1}(x_2-x_1) \int_{\tau_1}^{\tau_2}\left(\frac{t+\frac{\lambda_2}{\lambda_1}\tau_1-\tau_1}{\lambda_2}\right)^{-1}\log\left(\frac{t+\frac{\lambda_2}{\lambda_1}\tau_1-\tau_1}{\lambda_2}\right)f_{x_2}(t+h|\lambda_1,\lambda_2,\eta)^{\beta+1}dt
\\
&-(\eta-1)\frac{\tau_1}{\lambda_1}(x_2-x_1) \int_{\tau_1}^{\tau_2}\left(\frac{t+\frac{\lambda_2}{\lambda_1}\tau_1-\tau_1}{\lambda_2}\right)^{\eta-1}\log\left(\frac{t+\frac{\lambda_2}{\lambda_1}\tau_1-\tau_1}{\lambda_2}\right)f_{x_2}(t+h|\lambda_1,\lambda_2,\eta)^{\beta+1}dt
\\
&-\frac{\tau_1}{\lambda_1}(x_2-x_1) \int_{\tau_1}^{\tau_2}\left(\frac{t+\frac{\lambda_2}{\lambda_1}\tau_1-\tau_1}{\lambda_2}\right)^{\eta-1}f_{x_2}(t+h|\lambda_1,\lambda_2,\eta)^{\beta+1}dt
\\
&-\eta\frac{\tau_1}{\lambda_1}(x_2-x_1) \int_{\tau_1}^{\tau_2}\left(\frac{t+\frac{\lambda_2}{\lambda_1}\tau_1-\tau_1}{\lambda_2}\right)^{\eta-1}\log\left(\frac{t+\frac{\lambda_2}{\lambda_1}\tau_1-\tau_1}{\lambda_2}\right)f_{x_2}(t+h|\lambda_1,\lambda_2,\eta)^{\beta+1}dt
\\
&+\eta\frac{\tau_1}{\lambda_1}(x_2-x_1) \int_{\tau_1}^{\tau_2}\left(\frac{t+\frac{\lambda_2}{\lambda_1}\tau_1-\tau_1}{\lambda_2}\right)^{2\eta-1}\log\left(\frac{t+\frac{\lambda_2}{\lambda_1}\tau_1-\tau_1}{\lambda_2}\right)f_{x_2}(t+h|\lambda_1,\lambda_2,\eta)^{\beta+1}dt
\\
& +x_2 \int_{\tau_1}^{\tau_2}\left(\frac{t+\frac{\lambda_2}{\lambda_1}\tau_1-\tau_1}{\lambda_2}\right)^{\eta}f_{x_2}(t+h|\lambda_1,\lambda_2,\eta)^{\beta+1}dt
\\
&+\eta x_2 \int_{\tau_1}^{\tau_2}\left(\frac{t+\frac{\lambda_2}{\lambda_1}\tau_1-\tau_1}{\lambda_2}\right)^{\eta}\log\left(\frac{t+\frac{\lambda_2}{\lambda_1}\tau_1-\tau_1}{\lambda_2}\right)f_{x_2}(t+h|\lambda_1,\lambda_2,\eta)^{\beta+1}dt
\\
&-\eta x_2  \int_{\tau_1}^{\tau_2}\left(\frac{t+\frac{\lambda_2}{\lambda_1}\tau_1-\tau_1}{\lambda_2}\right)^{2\eta}\log\left(\frac{t+\frac{\lambda_2}{\lambda_1}\tau_1-\tau_1}{\lambda_2}\right)f_{x_2}(t+h|\lambda_1,\lambda_2,\eta)^{\beta+1}dt
\\
&=- x_2 \zeta_{0,\beta}^{\tau_1,\tau_2}(a_0,a_1,\eta)
-\eta x_2H_{0,1,\beta}^{\tau_1,\tau_2}(a_0,a_1,\eta)
 +2\eta x_2 H_{\eta,1,\beta}^{\tau_1,\tau_2}(a_0,a_1,\eta)
 \\
&
+ \frac{\eta-1}{\eta}\frac{\tau_1}{\lambda_1}(x_2-x_1)\zeta_{-1,\beta}^{\tau_1,\tau_2}(a_0,a_1,\eta)
+(\eta-1)\frac{\tau_1}{\lambda_1}(x_2-x_1)H_{-1,1,\beta}^{\tau_1,\tau_2}(a_0,a_1,\eta)
\\
&
-(2\eta-1)\frac{\tau_1}{\lambda_1}(x_2-x_1)H_{\eta-1,1,\beta}^{\tau_1,\tau_2}(a_0,a_1,\eta)
-\frac{\tau_1}{\lambda_1}(x_2-x_1) \zeta_{\eta-1,\beta}^{\tau_1,\tau_2}(a_0,a_1,\eta)
\\
&
+\eta\frac{\tau_1}{\lambda_1}(x_2-x_1)H_{2\eta-1,1,\beta}^{\tau_1,\tau_2}(a_0,a_1,\eta)
 +x_2\zeta_{\eta,\beta}^{\tau_1,\tau_2}(a_0,a_1,\eta)
-\eta x_2 H_{2\eta,1,\beta}^{\tau_1,\tau_2}(a_0,a_1,\eta).
\end{align*}
And for $J_{23,\tau_2}^{\beta}(a_0,a_1,\eta)$
\begin{align*}
J_{23,\tau_2}^{\beta}(a_0,a_1,\eta)=\frac{\partial \log\left(1-F_{x_2}(\tau_2+h|\lambda_1,\lambda_2,\eta)\right)}{\partial a_1}\frac{\partial \log\left(1-F_{x_2}(\tau_2+h|\lambda_1,\lambda_2,\eta)\right)}{\partial \eta}\left(1-F_{x_2}(\tau_2+h|\lambda_1,\lambda_2,\eta)\right)^{\beta+1}.
\end{align*}
Using the derivatives calculated in  \eqref{eq:dev_F2_a1} and \eqref{eq:dev_F2_eta}.
\begin{align*}
\frac{\partial \log\left(1-F_{x_2}(\tau_2+h|\lambda_1,\lambda_2,\eta)\right)}{\partial a_1}&=\eta\left(\frac{\tau_2+\frac{\lambda_2}{\lambda_1}\tau_1-\tau_1}{\lambda_2}\right)^{\eta-1}\frac{\frac{x_1\tau_1\lambda_2}{\lambda_1}+x_2(\tau_2-\tau_1)}{\lambda_2}
\\
\frac{\partial \log\left(1-F_{x_2}(\tau_2+h|\lambda_1,\lambda_2,\eta)\right)}{\partial \eta}&=-\log\left(\frac{\tau_2+\frac{\lambda_2}{\lambda_1}\tau_1-\tau_1}{\lambda_2}\right)\left(\frac{\tau_2+\frac{\lambda_2}{\lambda_1}\tau_1-\tau_1}{\lambda_2}\right)^{\eta}
\end{align*}
\begin{align*}
\Rightarrow \frac{\partial \log\left(1-F_{x_2}(\tau_2+h|\lambda_1,\lambda_2,\eta)\right)}{\partial a_1}\frac{\partial \log\left(1-F_{x_2}(\tau_2+h|\lambda_1,\lambda_2,\eta)\right)}{\partial \eta}&= -\eta\frac{\frac{x_1\tau_1\lambda_2}{\lambda_1}+x_2(\tau_2-\tau_1)}{\lambda_2}\log\left(\frac{\tau_2+\frac{\lambda_2}{\lambda_1}\tau_1-\tau_1}{\lambda_2}\right)
\\
&\log\left(\frac{\tau_2+\frac{\lambda_2}{\lambda_1}\tau_1-\tau_1}{\lambda_2}\right)^{2\eta-1}
\end{align*}
\begin{align*}
\Rightarrow  J_{13,\tau_2}^{\beta}(a_0,a_1,\eta)&=\frac{\partial \log\left(1-F_{x_2}(\tau_2+h|\lambda_1,\lambda_2,\eta)\right)}{\partial a_0}\frac{\partial \log\left(1-F_{x_2}(\tau_2+h|\lambda_1,\lambda_2,\eta)\right)}{\partial \eta}\left(1-F_{x_2}(\tau_2+h|\lambda_1,\lambda_2,\eta)\right)^{\beta+1}
\\
&=\eta\frac{\frac{x_1\tau_1\lambda_2}{\lambda_1}+x_2(\tau_2-\tau_1)}{\lambda_2}\log\left(\frac{\tau_2+\frac{\lambda_2}{\lambda_1}\tau_1-\tau_1}{\lambda_2}\right)\left(\frac{\tau_2+\frac{\lambda_2}{\lambda_1}\tau_1-\tau_1}{\lambda_2}\right)^{2\eta-1}
\\
&\cdot \exp\left(-\frac{\tau_2+\frac{\lambda_2}{\lambda_1}\tau_1-\tau_1}{\lambda_2}(\beta+1)\right).
\end{align*}
\end{proof}

\subsection{Proof Proposition 11}

\begin{proof}
For $\xi_1(a_0)$ we have
\begin{align*}
\xi_1^{\beta}(a_0)&=\int_{0}^{\tau_1}\frac{\partial \log\left(f_{x_1}(t|\lambda_1,\eta)\right)}{\partial a_0}f_{x_1}(t|\lambda_1,\eta)^{\beta+1}dt+\int_{\tau_1}^{\tau_2}\frac{\partial \log\left(f_{x_2}(t+h|\lambda_1,\lambda_2,\eta)\right)}{\partial a_0} f_{x_2}(t+h|\lambda_1,\lambda_2,\eta)^{\beta+1}dt
\\
&+\frac{\partial \log\left(1-F_{x_2}(\tau_2+h|\lambda_1,\lambda_2,\eta)\right)}{\partial a_0}\left(1-F_{x_2}(\tau_2+h|\lambda_1,\lambda_2,\eta)\right)^{\beta+1}
\\
&=\xi_{1,\tau_1}^{\beta}(a_0)+\xi_{1,\tau_1 \tau_2}^{\beta} (a_0)+ \xi_{1,\tau_2}^{\beta}(a_0).
\end{align*}
For $\xi_{1,\tau_1}^{\beta}(a_0)$ we have taking the derivatives calculated in \eqref{eq:dev_f1_a0}.
\begin{align*}
\frac{\partial \log\left(f_{x_1}(t|\lambda_1,\eta)\right)}{\partial a_0}&=\eta\left(-1+\left(\frac{t}{\lambda_1}\right)^{\eta}\right)
\\
\Rightarrow  \xi_{1,\tau_1}^{\beta}(a_0)&=\eta\left\{-\int_{0}^{\tau_1}f_{x_1}(t|\lambda_1,\eta)^{\beta+1}dt+\int_{0}^{\tau_1}\left(\frac{t}{\lambda_1}\right)^{\eta}f_{x_1}(t|\lambda_1,\eta)^{\beta+1}dt \right\}
\\
&=\eta\left\{-\zeta_{0,\beta}^{\tau_1}(a_0,a_1,\eta)+\zeta_{\eta,\beta}^{\tau_1}(a_0,a_1,\eta)\right\}.
\end{align*}
For $\xi_{1,\tau_1\tau_2}^{\beta}(a_0)$ we have taking the derivatives calculated in \eqref{eq:dev_f2_a0}.
\begin{align*}
\frac{\partial \log\left(f_{x_2}(t+h|\lambda_1,\lambda_2,\eta)\right)}{\partial a_0}&=\eta\left(-1+\left(\frac{t+\frac{\lambda_2}{\lambda_1}\tau_1-\tau_1}{\lambda_2}\right)^{\eta}\right)
\\
\Rightarrow  \xi_{1,\tau_1\tau_2}^{\beta}(a_0)&=\eta\left\{-\int_{0}^{\tau_1}f_{x_2}(t+h|\lambda_1,\lambda_2,\eta)^{\beta+1}dt+\int_{0}^{\tau_1}\left(\frac{t+\frac{\lambda_2}{\lambda_1}\tau_1-\tau_1}{\lambda_2}\right)^{\eta}f_{x_2}(t+h|\lambda_1,\lambda_2,\eta)^{\beta+1}dt \right\}
\\
&=\eta\left\{-\zeta_{0,\beta}^{\tau_1,\tau_2}(a_0,a_1,\eta)+\zeta_{\eta,\beta}^{\tau_1,\tau_2}(a_0,a_1,\eta)\right\}.
\end{align*}
And for $\xi_{1,\tau_2}^{\beta}(a_0)$ we have taking the derivatives calculated in \eqref{eq:dev_F2_a0}.
\begin{align*}
\frac{\partial \log\left(1-F_{x_2}(\tau_2+h|\lambda_1,\lambda_2,\eta)\right)}{\partial a_0}&=\eta\left(\frac{\tau_2+\frac{\lambda_2}{\lambda_1}\tau_1-\tau_1}{\lambda_2}\right)^{\eta}
\\
\Rightarrow  \xi_{1,\tau_2}^{\beta}(a_0)&=\eta \left(\frac{\tau_2+\frac{\lambda_2}{\lambda_1}\tau_1-\tau_1}{\lambda_2}\right)^{\eta} \exp\left(-\left(\frac{\tau_2+\frac{\lambda_2}{\lambda_1}\tau_1-\tau_1}{\lambda_2}\right)^{\eta}(\beta+1)\right).
\end{align*}
\end{proof}

\subsection{Proof Proposition 12}

\begin{proof}
For $\xi_2(a_1)$ we have
\begin{align*}
\xi_2^{\beta}(a_1)&=\int_{0}^{\tau_1}\frac{\partial \log\left(f_{x_1}(t|\lambda_1,\eta)\right)}{\partial a_1}f_{x_1}(t|\lambda_1,\eta)^{\beta+1}dt+\int_{\tau_1}^{\tau_2}\frac{\partial \log\left(f_{x_2}(t+h|\lambda_1,\lambda_2,\eta)\right)}{\partial a_1} f_{x_2}(t+h|\lambda_1,\lambda_2,\eta)^{\beta+1}dt
\\
&+\frac{\partial \log\left(1-F_{x_2}(\tau_2+h|\lambda_1,\lambda_2,\eta)\right)}{\partial a_1}\left(1-F_{x_2}(\tau_2+h|\lambda_1,\lambda_2,\eta)\right)^{\beta+1}
\\
&=\xi_{2,\tau_1}^{\beta}(a_1)+\xi_{2,\tau_1 \tau_2}^{\beta} (a_1)+ \xi_{2,\tau_2}^{\beta}(a_1).
\end{align*}
For $\xi_{2,\tau_1}^{\beta}(a_1)$ we have taking the derivatives calculated in \eqref{eq:dev_f1_a1}.
\begin{align*}
\frac{\partial \log\left(f_{x_1}(t|\lambda_1,\eta)\right)}{\partial a_1}&=\eta x_1\left(-1+\left(\frac{t}{\lambda_1}\right)^{\eta}\right)
\\
\Rightarrow  \xi_{2,\tau_1}^{\beta}(a_1)&=\eta x_1 \left\{-\int_{0}^{\tau_1}f_{x_1}(t|\lambda_1,\eta)^{\beta+1}dt+\int_{0}^{\tau_1}\left(\frac{t}{\lambda_1}\right)^{\eta}f_{x_1}(t|\lambda_1,\eta)^{\beta+1}dt \right\}
\\
&=\eta x_1\left\{-\zeta_{0,\beta}^{\tau_1}(a_0,a_1,\eta)+\zeta_{\eta,\beta}^{\tau_1}(a_0,a_1,\eta)\right\}.
\end{align*}
For $\xi_{2,\tau_1\tau_2}^{\beta}(a_1)$ we have taking the derivatives calculated in \eqref{eq:dev_f2_a1}.
\begin{align*}
\frac{\partial \log\left(f_{x_2}(t+h|\lambda_1,\lambda_2,\eta)\right)}{\partial a_1}&=-\eta x_2 +(\eta-1) \frac{\tau_1}{\lambda_1}(x_2-x_1)\left(\frac{t+\frac{\lambda_2}{\lambda_1}\tau_1-\tau_1}{\lambda_2}\right)^{-1}-\eta\frac{\tau_1}{\lambda_1}(x_2-x_1)\left(\frac{t+\frac{\lambda_2}{\lambda_1}\tau_1-\tau_1}{\lambda_2}\right)^{\eta-1}\nonumber
\\
&+\eta x_2\left(\frac{t+\frac{\lambda_2}{\lambda_1}\tau_1-\tau_1}{\lambda_2}\right)^{\eta} \nonumber
\\
\Rightarrow  \xi_{2,\tau_1\tau_2}^{\beta}(a_1)&=-\eta x_2 \int_{0}^{\tau_1}f_{x_2}(t+h|\lambda_1,\lambda_2,\eta)^{\beta+1}dt
\\
&-\eta\frac{\tau_1}{\lambda_1}(x_2-x_1)\int_{0}^{\tau_1}\left(\frac{t+\frac{\lambda_2}{\lambda_1}\tau_1-\tau_1}{\lambda_2}\right)^{-1}f_{x_2}(t+h|\lambda_1,\lambda_2,\eta)^{\beta+1}dt 
\\
&-\eta\frac{\tau_1}{\lambda_1}(x_2-x_1)\int_{0}^{\tau_1}\left(\frac{t+\frac{\lambda_2}{\lambda_1}\tau_1-\tau_1}{\lambda_2}\right)^{\eta-1}f_{x_2}(t+h|\lambda_1,\lambda_2,\eta)^{\beta+1}dt 
\\
&+\eta x_2 \int_{0}^{\tau_1}\left(\frac{t+\frac{\lambda_2}{\lambda_1}\tau_1-\tau_1}{\lambda_2}\right)^{\eta}f_{x_2}(t+h|\lambda_1,\lambda_2,\eta)^{\beta+1}dt 
\\
&=-\eta x_2 \zeta_{0,\beta}^{\tau_1,\tau_2}(a_0,a_1,\eta)+(\eta-1)\frac{\tau_1}{\lambda_1}(x_2-x_1)\zeta_{-1,\beta}^{\tau_1,\tau_2}(a_0,a_1,\eta)
\\
&-\eta\frac{\tau_1}{\lambda_1}(x_2-x_1)\zeta_{\eta-1,\beta}^{\tau_1,\tau_2}(a_0,a_1,\eta)+\eta x_2(x_2-x_1)\zeta_{\eta,\beta}^{\tau_1,\tau_2}(a_0,a_1,\eta).
\end{align*}
And for $\xi_{2,\tau_2}^{\beta}(a_1)$ we have taking the derivatives calculated in \eqref{eq:dev_F2_a1}.
\begin{align*}
\frac{\partial \log\left(1-F_{x_2}(\tau_2+h|\lambda_1,\lambda_2,\eta)\right)}{\partial a_1}&=\eta\left(\frac{\tau_2+\frac{\lambda_2}{\lambda_1}\tau_1-\tau_1}{\lambda_2}\right)^{\eta-1}\frac{\frac{x_1\tau_1\lambda_2}{\lambda_1}+x_2(\tau_2-\tau_1)}{\lambda_2}
\\
\Rightarrow  \xi_{2,\tau_2}^{\beta}(a_1)&=\eta \left(\frac{\tau_2+\frac{\lambda_2}{\lambda_1}\tau_1-\tau_1}{\lambda_2}\right)^{\eta-1} \exp\left(-\left(\frac{\tau_2+\frac{\lambda_2}{\lambda_1}\tau_1-\tau_1}{\lambda_2}\right)^{\eta}(\beta+1)\right)\frac{\frac{x_1\tau_1\lambda_2}{\lambda_1}+x_2(\tau_2-\tau_1)}{\lambda_2}.
\end{align*}

\end{proof}
\subsection{Proof Proposition 13}
For $\xi_3(\eta)$ we have
\begin{align*}
\xi_3^{\beta}(\eta)&=\int_{0}^{\tau_1}\frac{\partial \log\left(f_{x_1}(t|\lambda_1,\eta)\right)}{\partial \eta}f_{x_1}(t|\lambda_1,\eta)^{\beta+1}dt+\int_{\tau_1}^{\tau_2}\frac{\partial \log\left(f_{x_2}(t+h|\lambda_1,\lambda_2,\eta)\right)}{\partial \eta} f_{x_2}(t+h|\lambda_1,\lambda_2,\eta)^{\beta+1}dt
\\
&+\frac{\partial \log\left(1-F_{x_2}(\tau_2+h|\lambda_1,\lambda_2,\eta)\right)}{\partial \eta}\left(1-F_{x_2}(\tau_2+h|\lambda_1,\lambda_2,\eta)\right)^{\beta+1}
\\
&=\xi_{3,\tau_1}^{\beta}(\eta)+\xi_{3,\tau_1 \tau_2}^{\beta} (\eta)+ \xi_{3,\tau_2}^{\beta}(\eta).
\end{align*}
For $\xi_{3,\tau_1}^{\beta}(\eta)$ we have taking the derivatives calculated in \eqref{eq:dev_f1_eta}.
\begin{align*}
\frac{\partial \log\left(f_{x_1}(t|\lambda_1,\eta)\right)}{\partial \eta}&=\frac{1}{\eta}+\log\left(\frac{t}{\lambda_1}\right)-\left(\frac{t}{\lambda_1}\right)^{\eta}\log\left(\frac{t}{\lambda_1}\right)
\\
\Rightarrow  \xi_{3,\tau_1}^{\beta}(\eta)&=\frac{1}{\eta} \int_{0}^{\tau_1}f_{x_1}(t|\lambda_1,\eta)^{\beta+1}dt
+\int_{0}^{\tau_1}\log\left(\frac{t}{\lambda_1}\right)f_{x_1}(t|\lambda_1,\eta)^{\beta+1}dt
\\
&-\int_{0}^{\tau_1}\left(\frac{t}{\lambda_1}\right)^{\eta}\log\left(\frac{t}{\lambda_1}\right)f_{x_1}(t|\lambda_1,\eta)^{\beta+1}dt 
\\
&=\frac{1}{\eta}\zeta_{0,\beta}^{\tau_1}(a_0,a_1,\eta)+H_{0,1,\beta}^{\tau_1}(a_0,a_1,\eta)-H_{\eta,1,\beta}^{\tau_1}(a_0,a_1,\eta).
\end{align*}
For $\xi_{3,\tau_1\tau_2}^{\beta}(\eta)$ we have taking the derivatives calculated in \eqref{eq:dev_f2_eta}.
\begin{align*}
\frac{\partial \log\left(f_{x_2}(t+h|\lambda_1,\lambda_2,\eta)\right)}{\partial \eta}&=\frac{1}{\eta}+\log\left(\frac{t+\frac{\lambda_2}{\lambda_1}\tau_1-\tau_1}{\lambda_2}\right)-\left(\frac{t+\frac{\lambda_2}{\lambda_1}\tau_1-\tau_1}{\lambda_2}\right)^{\eta}\log\left(\frac{t+\frac{\lambda_2}{\lambda_1}\tau_1-\tau_1}{\lambda_2}\right)
\end{align*}
\begin{align*}
\Rightarrow  \xi_{3,\tau_1\tau_2}^{\beta}(\eta)&=\frac{1}{\eta} \int_{\tau_1}^{\tau_2}f_{x_2}(t+h|\lambda_1,\lambda_2,\eta)^{\beta+1}dt
+\int_{\tau_1}^{\tau_2}\log\left(\frac{t+\frac{\lambda_2}{\lambda_1}\tau_1-\tau_1}{\lambda_2}\right)f_{x_2}(t+h|\lambda_1,\lambda_2,\eta)^{\beta+1}dt
\\
&-\int_{\tau_1}^{\tau_2}\left(\frac{t+\frac{\lambda_2}{\lambda_1}\tau_1-\tau_1}{\lambda_2}\right)^{\eta}\log\left(\frac{t+\frac{\lambda_2}{\lambda_1}\tau_1-\tau_1}{\lambda_2}\right)f_{x_2}(t+h|\lambda_1,\lambda_2,\eta)^{\beta+1}dt 
\\
&=\frac{1}{\eta}\zeta_{0,\beta}^{\tau_1,\tau_2}(a_0,a_1,\eta)+H_{0,1,\beta}^{\tau_1,\tau_2}(a_0,a_1,\eta)-H_{\eta,1,\beta}^{\tau_1,\tau_2}(a_0,a_1,\eta).
\end{align*}
And for $\xi_{3,\tau_2}^{\beta}(\eta)$ we have taking the derivatives calculated in \eqref{eq:dev_F2_eta}.
\begin{align*}
\frac{\partial \log\left(1-F_{x_2}(\tau_2+h|\lambda_1,\lambda_2,\eta)\right)}{\partial \eta}&=-\log\left(\frac{\tau_2+\frac{\lambda_2}{\lambda_1}\tau_1-\tau_1}{\lambda_2}\right)\left(\frac{\tau_2+\frac{\lambda_2}{\lambda_1}\tau_1-\tau_1}{\lambda_2}\right)^{\eta}
\\
\Rightarrow  \xi_{3,\tau_2}^{\beta}(\eta)&=-\log\left(\frac{\tau_2+\frac{\lambda_2}{\lambda_1}\tau_1-\tau_1}{\lambda_2}\right)\left(\frac{\tau_2+\frac{\lambda_2}{\lambda_1}\tau_1-\tau_1}{\lambda_2}\right)^{\eta}
\\
&\cdot \exp\left(-\left(\frac{\tau_2+\frac{\lambda_2}{\lambda_1}\tau_1-\tau_1}{\lambda_2}\right)^{\eta}(\beta+1)\right).
\end{align*}
\subsection{Calculations for integrals}
In this subsection we will get a formula to compute the functions $H_{\alpha,\gamma,\beta}^{\tau_1}(a_0,a_1,\eta)$ and  $H_{\alpha,\gamma,\beta}^{\tau_1,\tau_2}(a_0,a_1,\eta)$.

We have
\begin{align*}
H_{\alpha,\gamma,\beta}^{\tau_1}(a_0,a_1,\eta)=\lambda_1\left(\frac{\eta}{\lambda_1}\right)^{\beta+1}\int_{0}^{\tau_1} l^{\alpha+(\eta-1)(\beta+1)}\left(\log(l)\right)^{\gamma}\exp\left(-l^{\eta}(\beta+1)\right)dl.
\end{align*}
doing
\begin{align*}
l^{\eta}(\beta+1)=u \Rightarrow l&=\left(\frac{u}{\beta+1}\right)^{\frac{1}{\eta}}\Rightarrow dl=\frac{1}{(\beta+1)^{\frac{1}{\eta}}}\frac{1}{\eta}u^{\frac{1-\eta}{\eta}}du.
\end{align*}
\\
The upper limit is $\left(\frac{\tau_1}{\lambda_1}\right)^{\eta}(\beta+1)$
\\
\begin{align*}
H_{\alpha,\gamma,\beta}^{\tau_1}(a_0,a_1,\eta)=&\lambda_1\left(\frac{\eta}{\lambda_1}\right)^{\beta+1}
\\
&\int_{0}^{\left(\frac{\tau_1}{\lambda_1}\right)^{\eta}(\beta+1)} \left(\frac{u}{\beta+1}\right)^{\frac{\alpha+(\eta-1)(\beta+1)}{\eta}}\frac{1}{\eta^{\gamma}}\left(\log(u)-\log(\beta+1)\right)^\gamma \exp(-u)\cdot \frac{1}{(\beta+1)^{\frac{1}{\eta}}}\frac{1}{\eta}u^{\frac{1-\eta}{\eta}}du
\\
&=\lambda_1^{-\beta}(\beta+1)^{-\frac{\alpha+(\eta-1)(\beta+1)+1}{\eta}}\eta^{(\beta-\gamma)}
\\
&\cdot\int_{0}^{\left(\frac{\tau_1}{\lambda_1}\right)^{\eta}(\beta+1)} u^{\frac{\alpha+(\eta-1)(\beta+1)+1}{\eta}-1}\left(\log(u)-\log(\beta+1)\right)^\gamma \exp(-u)du.
\end{align*}
And for $\gamma \in \mathbb{n}$ we can expand $\left(\log(u)-\log(\beta+1)\right)^\gamma$ as a Newton's binomial.
\begin{align*}
H_{\alpha,\gamma,\beta}^{\tau_1}(a_0,a_1,\eta)=&\lambda_1^{-\beta}(\beta+1)^{-\frac{\alpha+(\eta-1)(\beta+1)+1}{\eta}}\eta^{(\beta-\gamma)}
\\
&\cdot \sum_{i=0}^{\gamma}(-1)^{\gamma}\left(\log\left(\beta+1\right)\right)^{\gamma}\int_{0}^{\left(\frac{\tau_1}{\lambda_1}\right)^{\eta}(\beta+1)}u^{\frac{\alpha+(\eta-1)(\beta+1)+1}{\eta}-1}\log(u)^{\gamma-i}\exp(-u)du
\\
&=\lambda_1^{-\beta}(\beta+1)^{-\frac{\alpha+(\eta-1)(\beta+1)+1}{\eta}}\eta^{(\beta-\gamma)}
\\
&\cdot\sum_{i=0}^{\gamma}\binom{\gamma}{i}(-1)^{i}\left(\log\left(\beta+1\right)\right)^{i} \left\{\int_{0}^{\infty}u^{\frac{\alpha+(\eta-1)(\beta+1)+1}{\eta}-1}\log(u)^{\gamma-i}\exp(-u)du\right.
\\
&-\left.\int_{\left(\frac{\tau_1}{\lambda_1}\right)^{\eta}(\beta+1)}^{\infty}u^{\frac{\alpha+(\eta-1)(\beta+1)+1}{\eta}-1}\log(u)^{\gamma-i}\exp(-u)du\right\}
\\
&=\lambda_1^{-\beta}(\beta+1)^{-\frac{\alpha+(\eta-1)(\beta+1)+1}{\eta}}\eta^{(\beta-\gamma)} \cdot\sum_{i=0}^{\gamma}\binom{\gamma}{i}(-1)^{i}\left(\log\left(\beta+1\right)\right)^{i}
\\
 &\cdot \left\{\frac{\partial^{\gamma-i}\Gamma\left(\frac{\alpha+(\eta-1)(\beta+1)+1}{\eta}\right)}{\partial  \frac{\alpha+(\eta-1)(\beta+1)+1}{\eta}} \right.\left .-  \frac{\partial^{\gamma-i}\Gamma\left(\frac{\alpha+(\eta-1)(\beta+1)+1}{\eta},\left(\frac{\tau_1}{\lambda_1}\right)^{\eta}(\beta+1) \right)}{\partial  \frac{\alpha+(\eta-1)(\beta+1)+1}{\eta}}   \right\}.
\end{align*}
And 
\begin{align*}
\frac{\partial^{\gamma-i}\Gamma\left(\frac{\alpha+(\eta-1)(\beta+1)+1}{\eta}\right)}{\partial  \frac{\alpha+(\eta-1)(\beta+1)+1}{\eta}}=&\Gamma\left(\frac{\alpha+(\eta-1)(\beta+1)+1}{\eta}\right)
\\
&\cdot B_{\gamma-i}\left(\psi \left(\frac{\alpha+(\eta-1)(\beta+1)+1}{\eta}\right), ..., \psi^{\gamma-i-1}\left(\frac{\alpha+(\eta-1)(\beta+1)+1}{\eta}\right)\right),
\end{align*}
where $B_{\gamma-i}$ is the $\gamma-i^{th}$ Bell's polynomial,
and
\begin{align*}
 \frac{\partial^{\gamma-i}\Gamma\left(\frac{\alpha+(\eta-1)(\beta+1)+1}{\eta},x \right)}{\partial  \frac{\alpha+(\eta-1)(\beta+1)+1}{\eta}} =&\log(x)^{\gamma-i}\Gamma\left(\frac{\alpha+(\eta-1)(\beta+1)+1}{\eta},x \right)+ \gamma x
 \\
 &\cdot\sum_{n=0}^{\gamma-i-1}P_n^{\gamma-i-1}\log(x)^{\gamma-i-1}T\left(3+n \big|  \frac{\alpha+(\eta-1)(\beta+1)+1}{\eta}\right),
\end{align*}
being
\begin{align*}
P_n^{\gamma-i-1}=\binom{\gamma-i-1}{n} n!,
\end{align*}
and
\begin{align*}
T\left(3+n \big| \frac{\alpha+(\eta-1)(\beta+1)+1}{\eta}\right)=G_{2+n,3+n}^{3+n,0}\left( \begin{matrix} 0, \dots, 0 \\  \frac{\alpha+(\eta-1)(\beta+1)+1}{\eta},-1, \dots, -1 \end{matrix} \middle| z \right)
\end{align*}
where the right term the Meijer-G function.
With an analogous reasoning, the only difference with $H_{\alpha,\gamma,\beta}^{\tau_1,\tau_2}(a_0,a_1,\eta)$ are the limits of integration and changing $\lambda_1$ for $\lambda_2$. So after doing the same substitution we have:
\begin{align*}
H_{\alpha,\gamma,\beta}^{\tau_1,\tau_2}(a_0,a_1,\eta)&=\lambda_2^{-\beta}(\beta+1)^{-\frac{\alpha+(\eta-1)(\beta+1)+1}{\eta}}\eta^{(\beta-\gamma)} \cdot\sum_{i=0}^{\gamma}\binom{\gamma}{i}(-1)^{i}\left(\log\left(\beta+1\right)\right)^{i}
\\
 &\cdot \left\{\frac{\partial^{\gamma-i}\Gamma\left(\frac{\alpha+(\eta-1)(\beta+1)+1}{\eta},\left(\frac{\tau_1}{\lambda_1}\right)^{\eta}(\beta+1) \right)}{\partial  \frac{\alpha+(\eta-1)(\beta+1)+1}{\eta}}  \right.\left .-  \frac{\partial^{\gamma-i}\Gamma\left(\frac{\alpha+(\eta-1)(\beta+1)+1}{\eta},\left(\frac{\tau_2+\frac{\lambda_1}{\lambda_2}\tau_1-\tau_1}{\lambda_2}\right)^{\eta}(\beta+1) \right)}{\partial  \frac{\alpha+(\eta-1)(\beta+1)+1}{\eta}}   \right\}.
\end{align*}
Then, we have the expressions for  $H_{\alpha,\gamma,\beta}^{\tau_1}(a_0,a_1,\eta)$ and  $H_{\alpha,\gamma,\beta}^{\tau_1,\tau_2}(a_0,a_1,\eta)$.
\end{document}